\documentclass[envcountsame,envcountsect,undated,nolinenumbers]{lnthi}

\usepackage{epsf}

\def\FcopyforlaterA#1#2#3#4{%
\long\def#1{\begin{#2}\label{#3}#4\end{#2}}%
}

\def\FcopyforlaterB#1#2#3#4{%
\renewcommand{\thetheorem}{\ref{#3}}%
\begin{#2}#4\end{#2}%
\renewcommand{\thetheorem}{\arabic{section}.\arabic{theorem}}%
}

\def\Fcopyforlater#1#2#3#4{
  \FcopyforlaterA{#1}{#2}{#3}{#4}
  \FcopyforlaterB{#1}{#2}{#3}{#4}
}

\def\Fpasteinsection#1{%
#1
}

\title{Succinct Choice Dictionaries}

\author{Torben Hagerup and Frank Kammer}
\institute{Institut f\"ur Informatik, Universit\"at Augsburg, 86135
Augsburg, Germany
\email{\{hagerup,kammer\}@informatik.uni-augsburg.de}}

\usepackage{mathdots}

\usepackage{amsmath}

\setpagewiselinenumbers

\def\Tvn#1{\hbox{\textit{#1\/}}}
\def\Tfloor#1{\lfloor #1\rfloor}
\def\Tceil#1{\lceil #1\rceil}
\def\TbbbN{\mathbb{N}}
\def\Ttwodots{\mathinner{\ldotp\ldotp}}%
\def\Tsup#1{^{\mbox{\scriptsize #1}}}%
\gdef\Tsub#1{_{\mbox{\scriptsize #1}}}%
\def\cchoice{\overline{\Tvn{choice}}}
\def\citerate{\overline{\Tvn{iterate}}}
\def\jj{{\bar j}}
\def\Trho{R}
\def\rbar{\overline{r}}
\newbox\Twbox\setbox\Twbox=\hbox{\small($\mskip 1mu\lower0.0pt\hbox{\epsffile{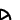}}\mskip 1mu,\mskip 2mu$3%
 $\mskip 2mu:\mskip 2mu$2$\mskip 0mu$)}
\def\Tmyw{\copy\Twbox}
\newbox\Tsbox\setbox\Tsbox=\hbox{$\lower0.0pt\hbox{\epsffile{weight.eps}}\mskip 1mu$}
\def\Tmys{\copy\Tsbox}
\def\Ttvn#1{t_{\mathit{#1}}}

\makeatletter
\def\eqalign#1{%
 \null\,\vcenter{\openup\jot\m@th\ialign{\strut\hfil$\displaystyle{%
  ##}$&$\displaystyle{{}##}$\hfil\crcr#1\crcr}}\,}%
\makeatother

\begin{document}

\maketitle{}%

\begin{abstract}%
The \emph{choice dictionary} is introduced as a
data structure that can be initialized
with a parameter $n\in\TbbbN=\{1,2,\ldots\}$
and subsequently maintains an initially empty subset $S$
of $\{1,\ldots,n\}$ under insertion, deletion, membership queries
and an operation \Tvn{choice} that returns
an arbitrary element of $S$.
The choice dictionary appears to be fundamental
in space-efficient computing.
We show that there is a choice dictionary
that can be initialized
with $n$ and an additional parameter $t\in\TbbbN$
and subsequently
occupies $n+O(n({t/w})^t+\log n)$ bits of memory
and executes each of the four operations \Tvn{insert},
\Tvn{delete}, \Tvn{contains} (i.e., a membership query)
and \Tvn{choice} in $O(t)$ time
on a word RAM with a word length of $w=\Omega(\log n)$ bits.
In particular, with $w=\Theta(\log n)$,
we can support \Tvn{insert}, \Tvn{delete},
\Tvn{contains} and \Tvn{choice} in constant time using
$n+O({n/{(\log n)^t})}$ bits for arbitrary fixed~$t$.
We extend our results to maintaining several
pairwise disjoint subsets of $\{1,\ldots,n\}$.

A static representation of a subset $S$ of $\{1,\ldots,n\}$
that consists of $n+s$ bits $b_1,\ldots,b_{n+s}$
is called \emph{systematic} if
$b_\ell=1\Leftrightarrow \ell\in S$ for $\ell=1,\ldots,n$
and is said to have \emph{redundancy} $s$.
We extend the former definition to dynamic data structures
and prove that the minimum redundancy of a systematic choice
dictionary with parameter $n$ that executes
every operation in $O(t)$ time on a
$w$-bit word RAM is $\Theta({n/{(t w)}})$,
provided that $t w=O({n/{\log n}})$.
Allowing
a redundancy of
$\Theta({{n\log(t\log n)}/{(t\log n)}}+n^\epsilon)$
for arbitrary fixed $\epsilon>0$,
we can support additional $O(t)$-time operations
\Tvn{p-rank} and \Tvn{p-select} that
realize a bijection from $S$ to $\{1,\ldots,|S|\}$
and its inverse.
The bijection may be chosen arbitrarily by the
data structure, but must remain fixed
as long as $S$ is not changed.
In particular, an element of $S$ can be
drawn uniformly at random in constant time
with a redundancy of $O({{n\log\log n}/{\log n}})$.

We study additional space-efficient
data structures for subsets $S$ of $\{1,\ldots,n\}$, including
one that supports only insertion and an
operation \Tvn{extract-choice} that returns
and deletes an arbitrary element of~$S$.
All our main data structures can be
initialized in constant time and
support efficient iteration over the set~$S$, and we can allow
changes to~$S$ while an iteration over~$S$ is in progress.
We use these abilities crucially in designing the most
space-efficient algorithms known for solving
a number of graph and other combinatorial problems
in linear time.
In particular, given an undirected graph $G$ with
$n$ vertices and $m$ edges, we can output a spanning
forest of $G$ in $O(n+m)$ time with at most
$(1+\epsilon)n$ bits of
working memory for arbitrary fixed $\epsilon>0$,
and if $G$ is connected,
we can output a shortest-path spanning
tree of $G$ rooted at a designated vertex
in $O(n+m)$ time
with $n\log_2 3+O({n/{(\log n)^t}})$ bits
of working memory for arbitrary fixed $t\in\TbbbN$.\\

\end{abstract}

\pagestyle{plain}
\thispagestyle{plain}

\section{Introduction}
\label{sec:intro}%
The \emph{redundancy} of
a data structure $D$ capable of representing an arbitrary
object in a nonempty set $\mathcal{S}$
is the (worst-case) number of bits
of memory occupied by $D$ beyond
the so-called in\-for\-ma\-tion-theoretic lower bound,
i.e., beyond
$\Tceil{\log|\mathcal{S}|}$---in this paper
``$\log$'' always denotes the
binary logarithm function $\log_2$.
If $\mathcal{S}$ depends on one or more size
parameters, $D$ is said to be \emph{succinct}
if its redundancy is $o(\log|\mathcal{S}|)$.
Whereas constant factors have traditionally been ignored
for both time and space bounds in the theoretical
analysis of algorithms and data structures,
in recent years there has been increased
interest in succinct data
structures \cite{BarAHM12,BroM99,Cla96,FerMM09,GeaRRR06,GolGGRR07,Jac89,MunRRR12,Pag01,Pat08,RamRS07}.
Most of the succinct data structures
developed to date are \emph{static}, i.e., they support
certain queries about the object $S\in\mathcal{S}$ stored,
but no updates of~$S$, and, in fact, even the time
to construct the data structure from $S$ has
frequently been ignored.
Of the dynamic succinct data structures developed
to date, a major part is
concerned with navigation in trees
\cite{Arr08,DavR11,FarM11,MunRS01,NavS14,RamR03},
and there are only few other contributions in
areas such as text processing
\cite{HeM10,LeeP09,MakN08,RusNO11}
and the maintenance of arrays, dictionaries and prefix sums
\cite{BroCDMS99,GupHSV07b,RamRR01}.
We add to the rather small collection of
known dynamic succinct data structures
that represent structures other than trees.

Data structures that represent an
(arbitrary) subset $S$ of a universe
of the form $U=\{1,\ldots,n\}$ and support various
sets of operations
have been studied in computer science for decades
\cite{Ajt88,AjtFK84,AndT07,BeaF02,BosDDHM13,DieKMHRT94,BoaKZ77,FiaN93,FreKS84,%
FreS89,FreW93,GupHSV07a,Joh82,PatT14,Ram96,Ruz08,TarY79,Tho00b,Wil84,Yao81}.
Our work continues this tradition and suggests
new sets of operations to be supported.
In the setting under consideration,
the condition of succinctness translates into space requirements
of $n+o(n)$ bits.
A powerful dynamic data type that we now call a
\emph{ragged dictionary}
was introduced in~\cite{ElmHK15}
and shown there to have a number of applications in
space-efficient graph algorithms.
In many situations the full power of the ragged
dictionary is not needed,
and the currently known construction of
ragged dictionaries is so involved that its
description is still in preparation.
In this paper we trim the ragged dictionary,
retaining only a set of operations that is simpler to
implement, allows a succinct realization,
and suffices in
most---but not all---applications of ragged dictionaries.
The resulting data type is characterized formally below.

\begin{definition}
\label{def:choice}%
A \emph{choice dictionary} is a data type that
can be initialized with an arbitrary integer
$n\in\TbbbN=\{1,2,\ldots\}$,
subsequently maintains
an initially empty subset $S$ of $U=\{1,\ldots,n\}$
and supports the following operations,
whose preconditions are stated in
parentheses:

\begin{tabbing}
\quad\=\hskip 2.4cm\=\hskip 2cm\=\kill
\>$\Tvn{insert}(\ell)$\>($\ell\in U$):\>
Replaces $S$ by $S\cup\{\ell\}$.\\
\>$\Tvn{delete}(\ell)$\>($\ell\in U$):\>
Replaces $S$ by $S\setminus\{\ell\}$.\\
\>$\Tvn{contains}(\ell)$\>($\ell\in U$):\>
Returns 1 if $\ell\in S$, 0 otherwise.\\
\>$\Tvn{choice}$:\>
\>Returns an (arbitrary) element of $S$
if $S\not=\emptyset$, 0 otherwise.
\end{tabbing}
\end{definition}

As is common and convenient, we use the term
``choice dictionary'' also to denote data structures
that implement the choice-dictionary data type.
Following the initialization of a choice dictionary
$D$ with an integer~$n$,
we call (the constant) $n$ the \emph{universe size}
of $D$ and (the variable) $S$ its
\emph{client set}.
The operation \Tvn{choice}, named so by
analogy with the axiom of choice, is central
and lends its name to the entire data type
as its most characteristic feature.
The operation is unusual in that a client set $S$ is
not mapped deterministically to a unique
prescribed return value;
instead, many return values may be legal for a given~$S$.
The operation, while not exactly new,
appears not to have been considered often in the past.
In fact, it is not uncommon for algorithms to comprise
steps that could be implemented via calls of \Tvn{choice}.
For many classic data structures, however, finding
an (arbitrary) element is no easier than finding
a certain specific element
(such as the minimum or the element most
recently inserted), for which reason such steps are often
overspecified by being formulated as queries for
specific elements.
In our setting, the flexibility inherent in
\Tvn{choice} is crucial to obtaining the most
efficient choice dictionaries and algorithms.

For integers $n_1$ and $n_2$ with $n_1\le n_2$,
the \emph{bit-vector representation over
$U'=\{n_1,\ldots,n_2\}$} of a subset
$S$ of $U'$ is the sequence
$(b_{n_1},\ldots,b_{n_2})$ of $|U'|$ bits with
$b_\ell=1\Leftrightarrow \ell\in S$, for $\ell=n_1,\ldots,n_2$,
or its obvious layout in $|U'|$ successive bits in memory.
If only the operations \Tvn{insert}, \Tvn{delete}
and \Tvn{contains} are to be supported,
a subset $S$ of $U=\{1,\ldots,n\}$ can be
stored simply as its bit-vector representation over $U$.
On the other hand, if
the operation \Tvn{delete} is omitted, the three
remaining operations are trivial to support in constant time
with close to $n$ bits.
It is the combination of \Tvn{insert} and
\Tvn{delete} with \Tvn{choice} that makes
the choice dictionary useful and its design interesting.

It is often possible to equip a choice dictionary
with facilities beyond the four core operations.
One of the most useful extensions is an operation
\Tvn{iterate}, which allows a user to process
the elements of~$S$ one by one.
In fact, we consider \Tvn{iterate} as a virtual
operation that is a shorthand for three concrete operations:
$\Tvn{iterate}.\Tvn{init}$, which prepares
for a new iteration over~$S$,
$\Tvn{iterate}.\Tvn{next}$, which
yields the next element $\ell$ of~$S$
(we say that $\ell$ is \emph{enumerated};
if all elements have already been enumerated,
0 is returned), and
$\Tvn{iterate}.\Tvn{more}$, which returns 1 if
one or more elements of~$S$ remain to be enumerated
and 0 otherwise.
When stating that a choice dictionary allows
iteration in a certain time $t$, what we mean
is that each of the three operations
$\Tvn{iterate}.\Tvn{init}$,
$\Tvn{iterate}.\Tvn{next}$ and
$\Tvn{iterate}.\Tvn{more}$
runs in time bounded by~$t$.
Our iterations are \emph{robust}, by
which we mean the following:
First and foremost, changes to the
client set $S$ through insertions
and deletions can be tolerated during an iteration.
Second, every element of $U$ present in $S$ during
the entire iteration is certain to be enumerated
by the iteration, while on the other hand no
element is enumerated more than once or at a time when it
does not belong to~$S$---in particular, if an element
does not belong to $S$ at any time during the
iteration, it is certain not to be enumerated.

Another useful extension is the ability to work
not only with the client set $S$, but also with
its complement $\overline{S}=U\setminus S$.
This involves an operation
$\cchoice$, which returns an arbitrary element
of $\overline{S}$ (0 if $\overline{S}=\emptyset$),
and possibly a virtual operation $\citerate$,
whose three concrete suboperations enumerate
$\overline{S}$.
Viewing membership in $S$ and in $\overline{S}$
as two different \emph{colors}, we call a choice
dictionary extended in this way a $2$-color choice dictionary,
whereas the original bare-bones choice dictionary
will be said to be \emph{colorless}.
We extend the concept of color
to $c$ colors, for integer $c\in\TbbbN$.
A $c$-color choice dictionary maintains a
\emph{semipartition} $(S_0,\ldots,S_{c-1})$
of $U=\{1,\ldots,n\}$, i.e.,
a sequence of (possibly empty)
disjoint subsets of $U$ whose union is~$U$,
called its \emph{client vector}.
The operations \Tvn{insert}, \Tvn{delete}
and \Tvn{contains} are replaced by

\begin{description}
\item[\normalfont$\Tvn{setcolor}(j,\ell)$]
($j\in\{0,\ldots,c-1\}$ and $\ell\in U$):
Changes the color of $\ell$ to $j$, i.e.,
moves $\ell$ to $S_j$ (if it is not already there).
\item[\normalfont$\Tvn{color}(\ell)$:]
Returns the color of~$\ell$, i.e.,
the unique $j\in\{0,\ldots,c-1\}$ with $\ell\in S_j$.
\end{description}

Moreover,
the operations \Tvn{choice} and \Tvn{iterate}
(with its three suboperations) take an additional
(first) argument $j\in\{0,\ldots,c-1\}$
that indicates the set $S_j$ to which the
operations are to apply;
e.g., $\Tvn{choice}(j)$ returns an arbitrary
element of $S_j$ (0 if $S_j=\emptyset$).
In applications $c$ is often a small constant.
To emphasize this view, we may write the argument $j$
as a subscript of the operation name
(e.g., $\Tvn{iterate}_j.\Tvn{next}$
instead of $\Tvn{iterate}(j).\Tvn{next}$).
Initially all elements of $U$ belong to~$S_0$.
In the special case $c=2$, we may write, e.g.,
\Tvn{choice} and $\cchoice$ or
$\Tvn{choice}_1$ and $\Tvn{choice}_0$, as convenient.
We have not attempted to optimize our results
for large values of~$c$.
Formally, we allow $c=1$, but a choice dictionary
with just one color is trivial and useless,
and in proofs we tacitly assume $c\ge 2$.
A view equivalent to that of a semipartition
$(S_0,\ldots,S_{c-1})$ of $\{1,\ldots,n\}$
is that a $c$-color choice dictionary with
universe size~$n$ must
maintain an array of $n$ values drawn from
$\{0,\ldots,c-1\}$ under certain
obvious operations.

Of course, all operations of the
colorless choice
dictionary with universe size~$n$ and many more can be supported in
$O(\log n)$ time by a balanced binary tree.
Our interest, however, lies with data structures that
are more efficient than binary trees in terms of
both time and space.
Our model of computation is a word RAM~\cite{AngV79,Hag98}
with a word length of $w\in\TbbbN$ bits, where we assume that $w$ is
large enough to allow all memory words in use to be addressed.
As part of ensuring this, in the context of a universe or
an input of size $n$, we always assume that
$w\ge\log n$.
The word RAM has constant-time operations
for addition, subtraction and multiplication
modulo $2^w$, division with truncation
($(x,y)\mapsto\Tfloor{{x/y}}$ for $y>0$),
left shift modulo $2^w$
($(x,y)\mapsto (x\ll y)\bmod 2^w$,
where $x\ll y=x\cdot 2^y$),
right shift
($(x,y)\mapsto x\gg y=\Tfloor{{x/{2^y}}}$),
and bitwise Boolean operations
($\textsc{and}$, $\textsc{or}$ and $\textsc{xor}$
(exclusive or)).
We also assume a constant-time operation to
load an integer that deviates from $\sqrt{w}$
by at most a constant factor---this enables the
proof of Lemma~\ref{lem:word}(a).
We do not assume the availability of
constant-time exponentiation,
a feature that would simplify some of
our data structures.
When nothing else is clear from the context,
integers manipulated algorithmically are assumed
to be of $O(w)$ bits, so that they can be
operated on in constant time.
Integers for which this assumption is not
made may be qualified as ``multiword''.
Multiword integers are assumed to be represented
in the positional system with base
$2^w$, i.e., in
a sequence of words, $w$ bits per word.

Our most surprising
result, proved in Section~\ref{sec:nonsystematic},
yields a colorless choice dictionary
that can be initialized for universe size~$n$
in constant time,
that executes \Tvn{insert}, \Tvn{delete},
\Tvn{contains} and \Tvn{choice} in constant time
and whose redundancy is $O({n/{(\log n)^t}})$
for arbitrary fixed $t\in\TbbbN$,
significantly better
than the best bound of $O(n)$
known for ragged dictionaries
used as choice dictionaries.
We generalize to several colors and to
an upper-bound tradeoff between time and space:

\Fcopyforlater{\unsystematic}{theorem}{thm:unsystematic-c}{%
For every fixed $\epsilon>0$,
there is a choice dictionary that, for arbitrary
$n,c,t\in\TbbbN$, can be initialized
for universe size $n$, $c$ colors and tradeoff
parameter $t$ in constant time and
subsequently occupies
$n\log_2 c+O(cn({{\epsilon c^2(\log c)t}/{\log (n+1)}})^t+c^{3 c}n^\epsilon)$
bits and supports
\Tvn{color}, \Tvn{setcolor}, \Tvn{choice} and
robust iteration in $O(t)$ time.
}

When $c$ is a power of~2, and in particular
for $c=2$, we achieve a better
space bound of
$n\log c+O(c n({{c^2(\log c)t}/w})^t+c^{\epsilon c^2}+\log n)$
bits, albeit with a time bound for
\Tvn{setcolor} of $O(t+c)$ instead of $O(t)$
(Theorem~\ref{thm:unsystematic-f}).
For $c=2$, this yields a redundancy of essentially
$O(n({t/w})^t)$ for execution times
of $O(t)$, the same as that achieved
by P\v atra\c scu for a different
problem~\cite[Theorem~4]{Pat08}.
Interestingly, we employ
P\v atra\c scu's technique,
as extended by Dodis,
P\u atra\c scu and Thorup~\cite{DodPT10},
in the proof of Theorem~\ref{thm:unsystematic-c},
but not in that of Theorem~\ref{thm:unsystematic-f}.
At a technical level, the problem of realizing
\Tvn{choice} can be viewed as that of finding an arbitrary
leaf with a given color in a tree with colored leaves,
but practically no space available for
navigational information at inner nodes.
Our solution forms groups of leaves and exploits
the fact that if a leaf group lacks some color
completely, it offers a certain potential for storing
foreign (namely navigational) information.
If below an inner node $u$ there is no such ``deficient'' leaf
group, on the other hand, the search can proceed
blindly from $u$---there are no ``dangerous'' subtrees.

For $n\in\TbbbN$,
a static data structure that represents a subset $S$ of
$U=\{1,\ldots,n\}$ is called \emph{systematic} if
its encoding of $S$ has the bit-vector
representation of $S$ over $U$ as a prefix~\cite{GalM07}---in
other words, $S$ is stored as its ``raw'' form,
possibly followed by other information.
The definition can be applied as it is to dynamic
data structures, but then precludes initialization
in $o({n/w})$ time and, more significantly,
prevents the representation from having
a size indication such as
an encoding of the integer~$n$ as a prefix.
We therefore use the following alternative definition:
A dynamic data structure $D$ that represents a
subset $S$ of $U=\{1,\ldots,n\}$ is systematic
if, beginning in a bit position that depends
only on~$n$, it contains a sequence $(b_1,\ldots,b_n)$
of $n$ bits
such that for each $\ell\in U$, $b_\ell=1\Leftrightarrow \ell\in S$
holds at all times after $D$'s first writing to~$b_\ell$, if any.
In a word RAM, the bit $b_\ell$ is part of a
word $H$ in memory, and $D$ first writes to $b_\ell$
when it first stores a value in $H$.
Until that point in time, we assume that $H$ and therefore
$b_\ell$ may contain arbitrary values (``be uninitialized'').
It is sometimes considered desirable
for a data structure to be systematic~\cite{GalM07}.
Our proof of Theorem~\ref{thm:unsystematic-c} does
not yield a systematic data structure,
but in Section~\ref{sec:systematic}
we propose an alternative and
systematic choice dictionary:

\Fcopyforlater{\systematic2}{theorem}{thm:systematic-2}{%
There is a 2-color systematic
choice dictionary that, for arbitrary $n,t\in\TbbbN$,
can be initialized for universe size $n$ and
tradeoff parameter $t$ in constant time
and subsequently occupies
$n+{n/{(t w)}}+O({n/{(t w)^2}}+\log n)$ bits
and executes \Tvn{insert}, \Tvn{delete},
\Tvn{contains}, \Tvn{choice}, $\cchoice$  and
robust iteration over the client set
and its complement in $O(t)$ time.
}

For $t w=O({n/{\log n}})$ (a condition that excludes
only cases of scant interest),
the product of the redundancy
and the execution time per operation is $O({n/w})$
for the choice dictionary of Theorem~\ref{thm:systematic-2}.
We prove in Subsection~\ref{subsec:lower} that this is
optimal in the sense that every systematic choice
dictionary with universe size~$n$
must have a redundancy-time product of $\Omega({n/w})$.
Our result, in fact, is considerably more precise:
In the bit-probe model~\cite{GalM07,Yao81}, 
if a systematic choice dictionary with universe size~$n$
has redundancy $s$ and inspects at most $t$ bits during
each execution of an operation, then
$(s+O(1))t\ge\alpha n$, where
$\alpha={1/{(e\ln 2)}}\approx 0.53$, and
we argue that this statement
does not hold if $\alpha$ is replaced by an
arbitrary constant larger than~1.
In a certain sense, therefore, the tradeoff between
redundancy and operation time has been determined to within
a factor of less than~2.
While there are linear or near-linear lower bounds
for the product of redundancy and query time
for certain static systematic data structures,
such as ones that support queries
for the sum, modulo 2, of the bits in prefixes
of a fixed bit string
(the prefix-sum problem)~\cite{GalM07},
we are not aware of nontrivial previous such bounds
for dynamic data structures.

Following the introduction of the ragged dictionary,
another systematic choice dictionary was developed
independently by Banerjee,
Chakraborty and Raman~\cite{BanCR16}.
Their construction is similar to that of the
special case of Theorem~\ref{thm:systematic-2}
obtained by taking $t=\Theta(1)$ and $w=\Theta(\log n)$.
The redundancy is indicated only as $o(n)$,
however, and an inspection of the proof shows
the redundancy to be $\Theta({{n\log\log n}/{\log n}})$,
not the optimal $O({n/{\log n}})$ of
Theorem~\ref{thm:systematic-2}.
Moreover, the data structure of~\cite{BanCR16}
supports neither robust iteration nor
$\cchoice$, and it cannot be initialized
in constant time.
An early choice dictionary
(with the \Tvn{choice} operation called \Tvn{choose-one})
was described by Briggs and Torczon~\cite{BriT93}.
Their data structure requires $\Theta(n\log n)$ bits.

A first space-efficient algorithm
for (essentially) the problem of linear-time
computation of a
shortest-path tree with a given root in a connected
unweighted graph was indicated in~\cite[Theorem 5.1]{ElmHK15}.
For input graphs with $n$ vertices and $m$ edges,
this took the form of a simple $O(n+m)$-time
reduction
to the problem of executing $O(n+m)$ operations on a
4-color choice dictionary with universe size~$n$.
Given that the interest in~\cite{ElmHK15} was not
with constant factors, a ragged dictionary was used
for the choice dictionary, and the bound on the
necessary amount of working memory
(i.e., memory in addition to read-only memory
that holds the input)
was indicated as $O(n)$ bits.
Restating the reduction and plugging in their own
choice dictionary, Banerjee et al.~\cite{BanCR16} derived a new
space bound for the problem, given as $2 n+o(n)$ bits.
Substituting our superior choice dictionaries of
either Theorem~\ref{thm:systematic-2} or
Theorem~\ref{thm:unsystematic-c}, we could
improve the lower-order term of this bound.
We instead obtain a more substantial
improvement (Theorem~\ref{thm:bfs}) by giving a new reduction of the
shortest-path problem to that of executing
$O(n+m)$ operations on a
choice dictionary that has only 3 colors
but must support robust iteration.
With Theorem~\ref{thm:unsystematic-c}, our space
bound becomes $n\log 3+O({n/{(\log n)^t}})$ bits
for arbitrary fixed $t\in\TbbbN$.

Much previous work has gone into the development of
\emph{rank-select structures} (also known as
\emph{indexable dictionaries}) that support
operations \Tvn{rank} and \Tvn{select}
\cite{Jac89}.
Formulated in terms of a client set
$S\subseteq\{1,\ldots,n\}$,
the two operations
are defined as follows:

\begin{description}
\item[\normalfont$\Tvn{rank}(\ell)$]
($\ell\in\{0,\ldots,n\}$):
Returns $|S\cap\{1,\ldots,\ell\}|$.
\item[\normalfont$\Tvn{select}(k)$]
($k\in\{1,\ldots,n\}$):
Returns the unique $\ell\in S$ with $\Tvn{rank}(\ell)=k$
if $k\le|S|$, 0 otherwise.
\end{description}

P\v atra\c scu showed that for
arbitrary fixed $t\in\TbbbN$, there is a
static rank-select structure that occupies
$n+O({n/{(\log n)^t}})$ bits and executes
both \Tvn{rank} and \Tvn{select} in constant time
\cite[Theorem~2]{Pat08} (his result, in fact, is more general).
For systematic static rank-select structures
with constant query time
the optimal redundancy is known to be
$\Theta({{n\log\log n}/{\log n}})$
\cite{Gol07,RamRS07}.
For the corresponding dynamic data type,
i.e., one that supports \Tvn{insert} and
\Tvn{delete} in addition to \Tvn{rank}
and \Tvn{select}, a lower bound of
$\Omega({{\log n}/{\log w}})$ on the execution
time of the slowest operation~\cite{FreS89}
precludes all hope of achieving a similar performance.
Returning to the setting of $c$ colors,
we introduce ``poor man's substitutes'' for
\Tvn{rank} and \Tvn{select} called
\Tvn{p-rank} and \Tvn{p-select} and show
in Subsection~\ref{subsec:random} that,
for arbitrary fixed $c\in\TbbbN$ and
$\epsilon>0$, for arbitrary
$t\in\TbbbN$ and allowing a redundancy of
$\Theta({{n\log(t\log n)}/{(t\log n)}}+n^\epsilon)$,
we can support \Tvn{p-rank} and \Tvn{p-select}
in $O(t)$ time in addition
to the usual choice-dictionary operations
(Theorem~\ref{thm:p}).
When the client vector is $(S_0,\ldots,S_{c-1})$,
both operations are defined in terms of a
sequence $(\pi_0,\ldots,\pi_{c-1})$, where
$\pi_j$ is a bijection from $S_j$ to
$\{1,\ldots,|S_j|\}$, for $j=0,\ldots,c-1$:

\begin{description}
\item[\normalfont$\Tvn{p-rank}(\ell)$]
($\ell\in U$):
Returns $\pi_j(\ell)$, where $j$ is the
color of $\ell$.
\item[\normalfont$\Tvn{p-select}(j,k)$]
($j\in\{0,\ldots,c-1\}$ and $k\in\{1,\ldots,n\}$):
Returns $\pi_j^{-1}(k)$ if $k\le|S_j|$, and 0 otherwise.
\end{description}

If the bijections $\pi_0,\ldots,\pi_{c-1}$ are viewed
as numbering the elements within each of the sets
$S_0,\ldots,S_{c-1}$, $\Tvn{p-rank}(\ell)$ therefore
returns the number of $\ell$ in its set,
and $\Tvn{p-select}(j,k)$ (that we may write
as $\Tvn{p-select}_j(k)$) returns the element
in $S_j$ numbered~$k$ (0 if there is no such element).
The sequence $(\pi_0,\ldots,\pi_{c-1})$
may be chosen arbitrarily by the choice dictionary,
subject only to the condition that it must remain
unchanged between calls of $\Tvn{setcolor}$
(or of $\Tvn{insert}$ and \Tvn{delete}).
The operations \Tvn{p-rank} and \Tvn{p-select} are
approximate inverses of each other in the sense that
$\Tvn{p-select}(\Tvn{color}(\ell),\Tvn{p-rank}(\ell))=\ell$
for all $\ell\in U$ and
$\Tvn{p-rank}(\Tvn{p-select}(j,k))=k$
for all $j\in\{0,\ldots,c-1\}$ and all
$k\in\{1,\ldots,|S_j|\}$.
The operations \Tvn{rank} and \Tvn{select}
generalize approximately to $c$ colors
as the special cases $\Tvn{rank}_j$ and
$\Tvn{select}_j$ of $\Tvn{p-rank}_j$ and $\Tvn{p-select}_j$
obtained by requiring $\pi_j$ to be the increasing
bijection from $S_j$ to $\{1,\ldots,|S_j|\}$,
for $j=0,\ldots,c-1$.
We obtain our result through a nonobvious combination
(illustrated in Fig.~\ref{fig:select}
on p.~\pageref{fig:select})
of (usual) choice dictionaries and other data structures.

The ``\Tvn{p-}'' in \Tvn{p-rank} and \Tvn{p-select} can be thought
of as an abbreviation for ``pseudo-'' or ``permuted''.
The operations \Tvn{p-rank} and \Tvn{p-select}
are closely related to
the classic ranking and
unranking operations within static but more complicated
classes of combinatorial objects~\cite{KreS99}.
Despite the arbitrariness inherent in
\Tvn{p-rank} and \Tvn{p-select}, the latter operation
has at least one important application, namely
to the generation of random elements:

\begin{description}
\item[\normalfont$\Tvn{uniform-choice}(j)$]
($j\in\{0,\ldots,c-1\}$):
Returns an element drawn uniformly at random
from $S_j$ if $S_j\not=\emptyset$, and 0 otherwise.
\end{description}

The realization of \Tvn{uniform-choice} in
terms of $\Tvn{p-select}$ is obvious:
A call $\Tvn{uniform-choice}(j)$
draws an integer $k$ uniformly at random
from $\{1,\ldots,|S_j|\}$ and returns
$\Tvn{p-select}(j,k)$.
The uncolored version of
\Tvn{uniform-choice} is called \Tvn{sample}
in \cite[Problem 1.3.35]{SedW11}.

We also study choice and choice-like dictionaries
that use fewer than $n$ bits when the number
$m$ of elements of nonzero color is considerably
smaller than~$n$.
In particular, we
show that constant-time
\Tvn{insert}, \Tvn{delete}, \Tvn{contains} and \Tvn{choice}
can be achieved with
$O(c m n^\epsilon+1)$ bits, for arbitrary
fixed $\epsilon>0$ (Theorem~\ref{thm:m-c}),
and in Subsection~\ref{subsec:mlog}
we describe
a data structure that uses
$O(m\log(2+{n/{(m+1)}})+1)$ bits
and supports
constant-time \Tvn{insert} and \Tvn{extract-choice},
where the latter operation
removes and returns an arbitrary element
of the client set~$S$
(if $S$ is empty, 0 is returned).
Our data structure is similar to the \emph{pool}
data structure of~\cite{HerS08}, where the operations
\Tvn{insert} and \Tvn{extract-choice} are called
\Tvn{put} and \Tvn{get}, respectively.
To represent $S$
using little space, our data structure stores $S$
in difference form, i.e., as a sequence of differences
between consecutive elements of~$S$.
For this to make sense, $S$ must be sorted, but this
renders constant-time insertion in $S$ difficult.
We set up a system of sorted reservoirs and
unsorted buffers and merge buffers into reservoirs
before they become too large.
Employing this data structure as the work-horse,
we can compute a spanning forest of an undirected
graph
with $n$ vertices and $m$ edges
in $O(n+m)$
time with at most $(1+\epsilon)n$ bits of
working memory, for arbitrary
fixed $\epsilon>0$ (Theorem~\ref{thm:cc}).
An algorithm that, slightly modified, can solve
the same problem in linear time was described previously
\cite[Theorem~5.1]{ElmHK15}, but the number of bits was
specified only as $O(n)$, and even with our best
choice dictionary the algorithm of
\cite{ElmHK15} would use at least
$n\log 3$ bits.

Our choice dictionaries have found uses elsewhere as
modest but crucial components of space-efficient
algorithms for Euler partition and
edge coloring of bipartite graphs~\cite{HagKL17}
and recognition of
outerplanar graphs~\cite{KamKL16}.
We currently explore their applications in
space-efficient solutions to a number of
vertex-coloring problems.

Although we do not assume the memory allocated
to hold a data structure to have been initialized in
any way---it may hold arbitrary values---all
our main data structures can be initialized in constant time.
Whereas this is standard and trivial to achieve for
data structures such as binary trees, we have to
develop new techniques to achieve the same
for our succinct data structures for universes
of the form $\{1,\ldots,n\}$.
It is a convenient property to have, and it
is essential to some of our algorithmic applications.
What makes initialization in constant time difficult is,
above all, that small instances cannot in general
be handled by means of table lookup.

\section{A Very Simple Choice Dictionary}
\label{sec:quick}%
Before embarking on a more
comprehensive
development, in this section
we indicate the shortest route to one of
our results that, though elementary,
suffices for many applications:
A basic choice dictonary that supports each
of the four core operations in constant time
and uses $n+O({n/{\log n}})$ bits to
maintain a subset of $\{1,\ldots,n\}$.
The description will demonstrate, in particular,
that our choice dictionary not only uses less
space, but is also simpler than the construction
of Banerjee et al.~\cite{BanCR16}.
Readers who want more details, a greater generality,
additional operations or a tighter space bound
are referred to Sections
\ref{sec:preliminaries}--\ref{sec:systematic}.

The result is
obtained by the combination of three ingredients,
each of which is very simple.
One is a choice dictionary that is wasteful in terms
of space, one is a choice
dictionary for very small universes, and the final component is 
the combination of many choice dictionaries
in the standard pattern of a trie.

Recall that a systematic choice dictionary with
universe size~$n$ contains a bit-vector representation
$B$ over $\{1,\ldots,n\}$ of the client set
and that $B$ immediately supports \Tvn{insert},
\Tvn{delete} and \Tvn{contains}, so that the only
remaining problem is to support \Tvn{choice}.
Once the search for a~1 in $B$ has been narrowed
down to a \emph{group} of $\beta\log n$ consecutive bits, for
a suitable constant $\beta>0$, it
can be concluded, e.g., by table lookup
(aiming for constant-time initialization,
we do it differently).
Representing each group by the disjunction of its
constituent bits (altogether
$O({n/{\log n}})$ bits),
we are left with the task of locating
a 1 among the group bits, i.e., the universe size
has been reduced by a factor of $\Theta(\log n)$.
Playing the same trick once more, we have
``supergroups'' of $\Theta((\log n)^2)$ bits each,
some of which are empty, and the task is to find
a nonempty supergroup.
To this end we spend $O({n/{\log n}})$ bits
on storing a permutation $\pi$ that sorts the
supergroups by their \empty{status},
empty or nonempty, and direct \Tvn{choice} to
the supergroup at the ``nonempty'' end of the sorted list.
When a supergroup changes its status, the sorting
can be maintained by first interchanging the
supergroup in question (located with the aid of $\pi^{-1}$)
with a supergroup at the border between
empty and nonempty.

\section{Preliminaries}
\label{sec:preliminaries}%

We view our data structures as
``coming to life'' during an initialization that
fixes certain parameters, typically a
universe size, $n$, and possibly a number
of colors, $c$, and/or a tradeoff parameter, $t$,
that expresses the relative weight to be given
to speed versus economy of space.
After initialization, we may consider these
parameters as constants.
It is natural, e.g., to speak of a choice
dictionary \emph{with} a particular universe size.

When we state that a data structure uses a certain
number of bits of memory, this is a statement about
the number of bits occupied by the data structure
when it is in a quiescent state, i.e., between the
execution of operations.
During the execution of an operation, the data structure
may temporarily need more working space---we speak
of \emph{transient} space requirements.
By definition, the $w$-bit word RAM uses at least
$\Theta(w)$ bits whenever it executes an instruction,
so that every operation of every data structure
has transient space requirements of at least
$\Theta(w)$ bits.
All our operations
get by with $\Theta(w)$ bits of transient space
that will not be mentioned explicitly.
Most of our data structures must store a
constant number of integers such as the parameters
with which the structures were initialized.
In consequence, most of our space bounds include
a term of $O(\log n)$ bits that will not be
discussed in every case.
When several data structures are initialized with
the same parameters and do not need to support
independent iterations, they can generally share the
same $O(\log n)$ bits.

Many operations of a choice dictionary can be faced
with ``unusual'' situations, such as the insertion of
an element that is already present, or \Tvn{choice}
called when the client set is empty.
We have chosen---fairly arbitrarily---to define the
operations so that they either do nothing or return
the special value~0 in such circumstances.
Since the unusual situations can easily be detected,
the operations could be redefined to instead issue an
error message or take some other suitable action.

The following is an attempted formalization
of the standard
``initialization on the fly''
technique of \cite[Exercise 2.12]{AhoHU74}.

\begin{lemma}
\label{lem:2.12}%
There is a data structure with the following
properties:
First, for every $n\in\TbbbN$, it can be
initialized for universe size~$n$ and
subsequently maintains
a function $g$ from
$U=\{1,\ldots,n\}$ to $\{0,\ldots,n\}$,
initially the zero function that maps every
element of $U$ to~0,
under evaluation of $g$
and the following operation:

\begin{description}
\item[\normalfont$\Tvn{allocate}(\ell)$]
($\ell\in U$):
If $g(\ell)=0$,
changes the value of $g$ on $\ell$ from 0 to
an element of $\{1,\ldots,n\}\setminus g(U)$.
Otherwise does nothing.
\end{description}

\noindent
Second, for known $n$, the data structure uses at most
$2 n\Tceil{\log n}+\Tceil{\log(n+1)}$ bits, can be initialized
in constant time, and evaluates $g$ and
supports \Tvn{allocate} in constant time.
(If $n$ is not known, it can be stored in the
data structure in another $O(\log(n+1))$ bits).
\end{lemma}

\begin{proof}
The data structure stores an integer $\mu$,
initially~0, and two arrays
$G[1\Ttwodots n]$ and $G^{-1}[1\Ttwodots n]$ such that
for all $\ell\in U$ with $g(\ell)\not=0$,
$G[\ell]=g(\ell)\le\mu$ and $G^{-1}[G[\ell]]=\ell$.
To execute $\Tvn{allocate}(\ell)$ for $\ell\in U$
when $g(\ell)=0$, increment $\mu$ and store $\mu$ in $G[\ell]$ and
$\ell$ in $G^{-1}[\mu]$.
To evaluate $g(\ell)$ for $\ell\in U$, test whether
$1\le G[\ell]\le \mu$ and $G^{-1}[G[\ell]]=\ell$.
If this is the case, return $G[\ell]$;
otherwise return~0.
\end{proof}

The application of
Lemma~\ref{lem:2.12}
highlighted in \cite[Exercise 2.12]{AhoHU74}
is to the constant-time initialization of all
entries of an array $A[1\Ttwodots n]$ to some
value $\xi_0$.
More generally,
if $\Xi$ is the set of values storable in cells of $A$,
we can allow $\xi_0$ to be an
arbitrary function from $U=\{1,\ldots,n\}$ to $\Xi$
that can be evaluated in constant time using
a negligible amount of memory.
The access to $A$ can take the form of two functions:
$\Tvn{read}(\ell)$, where $\ell\in U$, returns $A[\ell]$,
and $\Tvn{write}(\ell,\xi)$, where $\ell\in U$ and $\xi\in \Xi$,
assigns the value $\xi$ to $A[\ell]$.
If $D$ is an instance of the data structure of
Lemma~\ref{lem:2.12} for universe size~$n$
and $g$ is the function that it maintains,
\Tvn{read} and \Tvn{write} can be realized as follows:

\begin{tabbing}
\quad\=\hskip 3cm\=\kill
\>$\Tvn{read}(\ell)$:
\>\textbf{if} $g(\ell)=0$ \textbf{then return} $\xi_0(\ell)$;
\textbf{else return} $A[g(\ell)]$;\\
\>$\Tvn{write}(\ell,\xi)$:
\>\textbf{if} $g(\ell)=0$ \textbf{then} \=$D.\Tvn{allocate}(\ell)$;\\
\>\>$A[g(\ell)]:=\xi$;
\end{tabbing}

Thus an array of $n$ entries can be assumed initialized
at the price of an additional $O(n\log n)$ bits.
By using such an array with single-bit entries
only to keep track of the
initialization of segments of $A$ of $\Theta(w)$ bits each and
representing $\xi_0(\ell)$ by the bit pattern $00\cdots 0$
for $\ell=1,\ldots,n$, using the vacated bit pattern
for $\xi_0(\ell)$ to represent the value that used
to be represented by $00\cdots 0$ (thus initializing
a segment amounts to clearing an area of
$\Theta(w)$ bits), we can reduce the number
of additional bits to $O({{(N\log n)}/w})$,
where $N$ is the number of bits
occupied by~$A$.
These considerations imply, in particular,
that an array can always
be assumed initialized at the price of a constant-factor
overhead in the space requirements.
Stronger results are known (see~\cite{FreK16}),
but the bound indicated
suffices for our purposes.

In addition to the operations considered in the
introduction, our discussion will refer to
a number of further operations that can be
added to a $c$-color choice dictionary with
universe size $n$ and client vector $(S_0,\ldots,S_{c-1})$
and are collected here for easy reference:

\begin{description}
\item[\normalfont\Tvn{universe-size}:]
Returns~$n$.
\item[\normalfont$\Tvn{size}(j)$]
($j\in\{0,\ldots,c-1\}$):
Returns $|S_j|$.
\item[\normalfont$\Tvn{isempty}(j)$]
($j\in\{0,\ldots,c-1\}$):
Returns 1 if $S_j=\emptyset$, and 0 otherwise.
\item[\normalfont$\Tvn{swap-colors}(j,j')$]
($j,j'\in\{0,\ldots,c-1\}$):
Interchanges $S_j$ and $S_{j'}$ (does nothing
if $j=j'$).
\item[\normalfont$\Tvn{elements}(j)$]
($j\in\{0,\ldots,c-1\}$):
Returns all elements of $S_j$
(packaged, e.g., in an array or a list).
\item[\normalfont$\Tvn{successor}(j,\ell)$]
($j\in\{0,\ldots,c-1\}$ and $\ell$ is an integer):
With $I=\{i\in S_j\mid i>\ell\}$,
returns $\min I$ if $I\not=\emptyset$, and 0 otherwise.
\item[\normalfont$\Tvn{predecessor}(j,\ell)$]
($j\in\{0,\ldots,c-1\}$ and $\ell$ is an integer):
With $I=\{i\in S_j\mid i<\ell\}$,
returns $\max I$ if $I\not=\emptyset$, and 0 otherwise.
\end{description}

The first three operations can be added to an arbitrary
choice dictionary at a very modest price,
namely constant time per call of the new operations,
constant additional time per call of the
original operations, and
$O(c\log(n+1))$ additional bits, used to store
$(n,|S_0|,\ldots,|S_{c-1}|)$ while preserving a
constant initialization time with
Lemma~\ref{lem:2.12}.
Similarly, using Lemma~\ref{lem:2.12},
we can realize \Tvn{swap-colors} in constant time
by storing a permutation that translates
between ``internal'' and ``external'' colors
and needs an additional $O(c\log c)$ bits.
So as not to clutter the picture, these operations
were not included in the repertoire of
Definition~\ref{def:choice}.
On the other hand, they can usually be assumed
to be available.
If the original choice dictionary supports
iteration in constant time,
$\Tvn{elements}(j)$ can carry out its job
in $O(|S_j|+1)$ time by executing a full
iteration over~$S_j$.

Whenever convenient, we can assume that the
argument $\ell$ of $\Tvn{successor}_j$ and
$\Tvn{predecessor}_j$ satisfies $\ell\in\{1,\ldots,n-1\}$:
For $\Tvn{successor}_j$, a value of $\ell$ larger
than $n-1$ is always associated with a return value of~0,
a value of $\ell$ smaller than 0 is
equivalent to $\ell=0$, and $\ell=0$ is
equivalent to $\ell=1$ unless $1\in S_j$,
in which case the return value is~1.

Several reductions among different operations
are obvious.
E.g., \Tvn{choice} reduces to \Tvn{p-select}
in the sense that if \Tvn{p-select} is available,
$\Tvn{choice}(j)$ can be implemented simply as
$\Tvn{p-select}(j,1)$.
We may express this succinctly
by writing

\begin{tabbing}
\quad\=\hskip 3cm\=\kill
\>$\Tvn{choice}(j)$:\>$\Tvn{p-select}(j,1)$;
\end{tabbing}

\noindent
Similarly, \Tvn{choice} reduces to \Tvn{iterate},
except that a call of \Tvn{choice} executed in this
manner interferes with the ongoing iteration, if any:

\begin{tabbing}
\quad\=\hskip 3cm\=\kill
\>$\Tvn{choice}(j)$:\>$\Tvn{iterate}(j).\Tvn{init}$;
 \textbf{return} $\Tvn{iterate}(j).\Tvn{next}$;
\end{tabbing}

\noindent
For colorless data structures,
\Tvn{choice} and \Tvn{extract-choice} are
mutually reducible if insertion and
deletion are available
and calls of 
\Tvn{choice} and \Tvn{extract-choice} can be allowed
to interfere with iteration,
\Tvn{p-rank} and \Tvn{p-select}:

\begin{tabbing}
\quad\=\hskip 3cm\=\kill
\>$\Tvn{choice}$:\>$\ell:=\Tvn{extract-choice}$;
 $\Tvn{insert}(\ell)$; \textbf{return} $\ell$;\\
\>$\Tvn{extract-choice}$:\>$\ell:=\Tvn{choice}$;
 $\Tvn{delete}(\ell)$; \textbf{return} $\ell$;
\end{tabbing}

\noindent
The following reductions were mentioned earlier.
Again, a call of \Tvn{elements} interferes with
any ongoing iteration.
A call $\Tvn{random}(k)$ is assumed to return an integer
drawn uniformly at random from $\{1,\ldots,k\}$.

\begin{tabbing}
\quad\=\hskip 3cm\=\kill
\>$\Tvn{uniform-choice}(j)$:\>
 $\Tvn{p-select}(j,\Tvn{random}(\Tvn{size}(j)))$;\\
\>$\Tvn{elements}(j)$:\>
 $X:=\emptyset$; $\Tvn{iterate}(j).\Tvn{init}$;\\
\>\>\textbf{while} $\Tvn{iterate}(j).\Tvn{more}$
 \textbf{do} $X:=X\cup\{\Tvn{iterate}(j).\Tvn{next}\}$;\\
\>\>\textbf{return} $X$;
\end{tabbing}

\noindent
Finally, if $\Tceil{\log(n+1)}$ additional bits
per iteration are available to hold a
private variable~$\ell$, several simultaneous
iterations reduce to any one of
\Tvn{successor}, \Tvn{predecessor}
and \Tvn{p-select}, the latter only if
robustness of the iteration is not required.
We give the details in the case of \Tvn{successor}.

\begin{tabbing}
\quad\=\hskip 3cm\=\kill
\>$\Tvn{iterate}(j).\Tvn{init}$:\>$\ell:=0$;\\
\>$\Tvn{iterate}(j).\Tvn{next}$:\>
 $\ell:=\Tvn{successor}(j,\ell)$; \textbf{return} $\ell$;\\
\>$\Tvn{iterate}(j).\Tvn{more}$:\>
 \textbf{if} $\Tvn{successor}(j,\ell)=0$ \textbf{then return} $0$;
 \textbf{else return} $1$;
\end{tabbing}

For $m,f\in\TbbbN$, let
$1_{m,f}=\sum_{i=0}^{m-1} 2^{i f}
={{(2^{m f}-1)}/{(2^f-1)}}$.
If the $(m f)$-bit binary representation of
$1_{m,f}$ is divided into $m$ fields of
$f$ bits each, each field contains the value~1.
The possibly multiword integer $1_{m,f}$ can be
computed from $m$ and $f$ in $O(1+{{m f}/w})$
time \cite[Theorem 2.5]{Hag15}.
Given a sequence $A=(a_1,\ldots,a_m)$ of
$m$ integers and an integer $k$,
let $\Tvn{rank}(k,A)=|\{i\in\TbbbN:1\le i\le m$ and $k\ge a_i\}|$.
The following lemma is proved with standard
word-RAM techniques,
more background on which can be found,
e.g., in~\cite{Hag98}.

\begin{lemma}
\label{lem:word}%
Let $m$ and $f$ be given integers
with $1\le m,f<2^w$ and suppose that a
sequence $A=(a_1,\ldots,a_m)$ with
$a_i\in\{0,\ldots,2^f-1\}$ for $i=1,\ldots,m$
is given
in the form of the $(m f)$-bit binary representation
of the integer
$x=\sum_{i=0}^{m-1} 2^{i f}a_{i+1}$.
Then the following holds:
\begin{itemize}
\item[(a)]
Let $I_{>0}=\{i\in\TbbbN:1\le i\le m$ and $a_i>0\}$.
Then, in $O(1+{{m f}/w})$ time,
we can test whether $I_{>0}=\emptyset$ and, if not,
compute $\min I_{>0}$ and $\max I_{>0}$.
\item[(b)]
Let $I_0=\{i\in\TbbbN:1\le i\le m$ and $a_i=0\}$.
Then, in $O(1+{{m f}/w})$ time,
we can test whether $I_0=\emptyset$ and, if not,
compute $\min I_0$ and $\max I_0$.
\item[(c)]
If an additional integer $k\in\{0,\ldots,2^f-1\}$
is given, then $O(1+{{m f}/w})$ time
suffices to compute the integer
$z=\sum_{i=0}^{m-1}2^{i f}b_{i+1}$, where
$b_i=1$ if $k\ge a_i$ and $b_i=0$ otherwise
for $i=1,\ldots,m$.
\item[(d)]
If $m<2^f$ and
an additional integer $k\in\{0,\ldots,2^f-1\}$
is given, then $\Tvn{rank}(k,A)$
can be computed
in $O(1+{{m f}/w})$ time.
\end{itemize}
\end{lemma}

\begin{proof}
(a)
$I_{>0}=\emptyset$ if and only if $x=0$,
which is trivial to test.
Assume that $x>0$.
Then $\max I_{>0}=\Tfloor{{{\Tfloor{\log x}}/f}}+1$,
so the problem of finding $\max I_{>0}$
reduces to that of computing $\Tfloor{\log x}$.
Fredman and Willard~\cite[pp.\ 431--432]{FreW93}
showed how to do this in
constant time for $m f\le w$
(a number of quantities needed by their algorithm,
such as $C_1$, can be computed in constant time
with the methods of~\cite{Hag15}).
Testing $w$ bits
at a time for being zero, it is easy to extend
their algorithm to the general case.
Computing
$\min I_{>0}$ reduces to computing $\max I_{>0}$
if one replaces $x$ by
${{((x\textsc{ xor }(x-1))+1)}/2}$
(cf.~\cite[Eq.\ 7.1.3-(40)]{Knu10}).

(b) If the binary representation of~$x$
is viewed as consisting of $m$ fields of $f$ bits each,
the task is to locate the leftmost or rightmost
zero field in~$x$.
We reduce this problem to that solved
in part~(a) by computing an $(m f)$-bit
integer $\overline{x}$, each of whose fields
is nonzero if and only if the corresponding field
in $x$ is zero:

\begin{tabbing}
\quad\=\kill
\>$t:=1_{m,f}\ll(f-1)$;\\
\>$y:=x\mathbin{\textsc{and}}t$;\\
\>$z:=(x-y)\mathbin{\textsc{or}}(y\gg(f-1))$;\\
\>$\overline{x}:=(t-z)\mathbin{\textsc{and}}t$;
\end{tabbing}

\noindent
Within each field, $t$ has a 1 in the most significant
bit position, called the position of the
\emph{test bit}, and
$y$ has only the test bits of~$x$.
If $f=1$, $z$ equals $x$, while otherwise
all test bits in $z$ are 0.
In either case, a field is nonzero in $z$ if
and only if it is nonzero in~$x$.
It is now easy to see that $\overline{x}$
has the required property.

(c) Reusing the notions of fields and test bits
of the proof of part~(b),
we first compute an integer~$z'$ such that the $i$th
test bit in $z'$, counted from the right,
is 1 if and only if $k\ge a_i$,
for $i=1,\ldots,m$.
Disregarding the values of the test bits in $k$
and $x$, this can be done by replicating the
value of $k$ to all fields through a multiplication
by $1_{m,f}$, setting all test bits in the
resulting integer $k'$ to~1, clearing them in $x$,
and subtracting the latter from the former
to obtain an integer~$y$.
Subsequently the original test-bit values of $k$
and $x$ are incorporated in the test to obtain $z'$
through bitwise manipulation of $k'$, $y$ and $x$.
In detail, the value of a particular test bit
in $z'$ should be 1 exactly if at least one of the
corresponding bits in $k'$ and $y$ is~1 and
either both these bits are~1 or the corresponding
bit in $x$ is~0.
Obtaining $z$ from $z'$ is just a matter of
``masking away'' unwanted bits and shifting
right by $f-1$ bits.

\begin{tabbing}
\quad\=\kill
\>$t:=1_{m,f}\ll(f-1)$;\\
\>$k':=k\cdot 1_{m,f}$;\\
\>$y:=(k'\mathbin{\textsc{or}}t)-(x-(x\mathbin{\textsc{and}}t))$;\\
\>$z':=(k'\mathbin{\textsc{or}}y)\mathbin{\textsc{and}}
((k'\mathbin{\textsc{and}}y)\mathbin{\textsc{or}}
(x\mathbin{\textsc{xor}}t))$;\\
\>$z:=(z'\mathbin{\textsc{and}}t)\gg(f-1)$;
\end{tabbing}

(d) The task reduces to summing the
bits in the integer $z$ of part~(c),
which can be carried out in
$O(1+{{m f}/w})$ time by computing
$((z\cdot 1_{m,f})\gg((m-1)f))\mathbin{\textsc{and}}(2^f-1)$.

For all parts of the lemma,
intermediate results should be produced and
consumed in streams, $O(w)$ bits at a time,
in order to keep the transient space at $O(w)$ bits.
\end{proof}

We assume the memory available to a data structure
to be a single sequence of $w$-bit words.
Occasionally, however, it will be convenient to
assume the availability of $k$ independent
memories, where $k\in\TbbbN$ is a constant.
It is a simple matter to simulate $k$ virtual
memories in the single actual memory.
For $i=1,\ldots,k$, let $s_i$ be the number
of bits used in the $i$th virtual memory
and take $s=\sum_{i=1}^k s_i$.
During times when $s\le w$, we store a unary encoding
of $(s_1,\ldots,s_k)$ followed by the actual
contents of the virtual memories in $O(s)$ bits.
When $s>w$, the actual memory words are
instead distributed among the $k$ virtual
memories in a round-robin fashion.
For $i=1,\ldots,k$, the contents of the $i$th
virtual memory are therefore stored in the
actual memory words numbered $i,i+k,i+2 k,\ldots,$
so that the total number of bits used is
$O(k(s+w))=O(s)$.
Both representations support reading and writing
of virtual memory words in constant time---in the
case of the first representation,
carrying out the necessary unary-to-binary conversion
with an algorithm of Lemma~\ref{lem:word}(a)---and
we can also switch between the two representations
in constant time.
The number of bits used is always $O(s)$.
When employing this technique, we will say that we use
\emph{memory interleaving}.

\section{Tries of Choice Dictionaries}
\label{sec:trie}%

We shall often have occasion to combine several
choice dictionaries in a trie structure to
obtain a choice dictionary
for a larger universe.
This section explains the simple principles
involved without formalizing them completely.

Let $n\ge 2$ and suppose that we have a data
structure that realizes an ordered tree $T$
with the leaf set $U=\{1,\ldots,n\}$ in which
each inner
node $u$ has an associated choice dictionary~$D_u$
whose universe size equals the degree
(number of children) of~$u$
and all leaves have the same depth and that
also maintains a \emph{current node} in~$T$.
Let $r$ be the root of~$T$.
Suppose that the data structure supports
the following operations:

\begin{description}
\item[\normalfont$\Tvn{movetoroot}$:]
Sets the current node to be $r$.
\item[\normalfont$\Tvn{movetoparent}$]
(the current node is not $r$):
Replaces the current node by its parent in~$T$.
\item[\normalfont$\Tvn{movetochild}(i)$]
(the current node $u$ is an inner node in $T$
and $i$ is a positive integer bounded by the
degree of~$u$):
Replaces the current node by its $i$th child
(in the order from left to right).
\item[\normalfont$\Tvn{height}$:]
Returns the height in $T$ of the current node.
\item[\normalfont$\Tvn{degree}$:]
Returns the degree in $T$ of the current node.
\item[\normalfont$\Tvn{leftindex}$:]
Returns one more than the number of nodes in $T$ of the same
height as the current node and strictly to its left.
\item[\normalfont$\Tvn{viachild}(\ell)$]
(the current node $u$ is an inner node in $T$ and
$\ell$ is a leaf descendant of $u$):
Returns the integer $i$ such that the $i$th child
of $u$ is an ancestor of~$\ell$.
\item[\normalfont$\Tvn{data}$:]
Returns the memory address of the choice
dictionary associated with the current node.
\end{description}

Then, after initializing $D_r$,
we can execute the operations 
of a colorless choice dictionary
with universe size~$n$ and client set $S$
as described below.
For each inner node $u$ in~$T$,
the client set of $D_u$ will contain an integer~$i$
if and only if at least one leaf descendant of
$u$'s $i$th child belongs to~$S$.
In the interest of clarity, we indicate the
current node as a (first) argument of
\Tvn{height}, \Tvn{leftindex} and \Tvn{viachild}.

\begin{description}
\item[\normalfont\Tvn{choice}:]
Return 0 if $D_r.\Tvn{isempty}=1$.
Otherwise,
starting at $r$ and as long as the current node
$u$ is not a leaf, step from $u$ to its $i$th child,
where $i$ is obtained with a call of $D_u.\Tvn{choice}$.
When a leaf $v$ is reached, return $\Tvn{leftindex}(v)$.
\item[\normalfont\Tvn{contains}$(\ell)$:]
Starting at $r$ and as long as the height in
$T$ of the current node $u$
is at least~2, let $i=\Tvn{viachild}(u,\ell)$
and, if $D_u.\Tvn{contains}(i)=1$,
step from $u$ to its $i$th child;
otherwise return~0.
If and when a node $u$ of height 1 is reached,
return $D_u.\Tvn{contains}(\Tvn{viachild}(u,\ell))$.
\item[\normalfont\Tvn{insert}$(\ell)$:]
Starting at $r$ and as long as
the height in $T$ of
the current node $u$ is at least~2,
let $i=\Tvn{viachild}(u,\ell)$ and let
$v$ be the $i$th child of~$u$.
If $D_u.\Tvn{contains}(i)=0$,
initialize $D_v$ (possibly not for the first time)
for universe size $d$, where $d$ is the degree of~$v$,
and execute $D_u.\Tvn{insert}(i)$.
Subsequently step from $u$ to~$v$.
When a node $u$ of height 1 is reached,
execute $D_u.\Tvn{insert}(\Tvn{viachild}(u,\ell))$.
\item[\normalfont\Tvn{delete}$(\ell)$:]
Starting at $r$ and as long as
the height in $T$ of the current node $u$ is at least~2,
let $i=\Tvn{viachild}(u,\ell)$.
If $D_u.\Tvn{contains}(i)=1$,
step from $u$ to its $i$th child;
otherwise abandon the deletion ($\ell\not\in S$).
If and when a node $v$ of height 1 is reached,
execute $D_v.\Tvn{delete}(viachild(v,\ell))$.
Then, as long as the current node $v$ is not $r$
and $D_v.\Tvn{isempty}=1$,
step from $v$ to its parent~$u$
and execute $D_u.\Tvn{delete}(\Tvn{viachild}(u,\ell))$.
\end{description}

A tree data structure that supports the operations
\Tvn{movetoroot}, etc., in constant time is simple
to design if $T$ is sufficiently regular.
For given $n\ge 2$, let $(p_1,p_2,\ldots)$ be
a finite or
infinite
\emph{degree sequence}
of positive integers whose product
is at least $n$
and, for $j=0,1,\ldots,$
take $P_j=\prod_{i=1}^j p_i$.
Let $h$ be the smallest positive
integer with $P_h\ge n$.
Then we can let $T$ be an ordered tree on the leaf
set $\{1,\ldots,n\}$ in which all leaves have
depth~$h$ and every node of height~$j$, except
possibly the rightmost one, has degree exactly $p_j$,
for $j=1,\ldots,h$.
Suppose that we represent the current node $u$
through the triple $(j,k,P_j)$, where
$j=\Tvn{height}(u)$ and $k=\Tvn{leftindex}(u)$.
Then we can navigate in $T$ through the following
simple observations:
The parent of $u$ is
$(j+1,\Tceil{{k/{p_{j+1}}}},P_j\cdot p_{j+1})$,
its $i$th child is
$(j-1,(k-1)p_j+i,{{P_j}/{p_j}})$,
and $\Tvn{viachild}(u,\ell)=
 \Tceil{p_j({\ell/{P_j}}-k+1)}$.
The evaluation of the operations
$\Tvn{height}$ and $\Tvn{leftindex}$ is trivial,
and the root of $T$ is (represented by) $(h,1,P_h)$.
As for accessing the choice dictionary
of the current node, i.e., evaluating \Tvn{data},
suppose that $f_1,\ldots,f_h$ are given nonnegative
integers and that each choice dictionary of a
node of height $j$ in $T$ can be accommodated in
a block of memory of $f_j$ bits, for $j=1,\ldots,h$.
For $j=0,\ldots,h$, let
$F_j=\sum_{i=1}^j f_j\Tceil{{n/{P_j}}}$ be
the total number of bits needed for the
blocks of nodes in $T$ of height at most~$j$.
Then, when the current node is $(j,k,P_j)$ and $j\ge 1$,
$\Tvn{data}$ can return $F_{j-1}+(k-1)f_j$
plus the starting address of a global segment of
$F_h$ bits reserved for all blocks of nodes in~$T$.
If we also maintain $F_j$, i.e., if we extend the
triple $(j,k,P_j)$ by $F_j$ as a fourth component,
\Tvn{data} can be executed in constant time as well.

If the choice dictionaries of all nodes in $T$
support \Tvn{iterate}, the overall choice
dictionary can also support \Tvn{iterate}
with the following procedure, which is explained
below:

\begin{description}
\item[\normalfont\Tvn{iterate}.\Tvn{init}:]
Execute $D_r.\Tvn{iterate}.\Tvn{init}$
and initialize an integer $\ell$ to~0.
\item[\normalfont\Tvn{iterate}.\Tvn{more}:]
If $\ell=0$, return $D_r.\Tvn{iterate}.\Tvn{more}$.
Otherwise,
starting at $r$ and as long as the current node
$u$ is not a leaf and no value was returned,
return 1 if $D_u.\Tvn{iterate}.\Tvn{more}=1$.
Otherwise step to the $i$th child of~$u$,
where $i=\Tvn{viachild}(u,\ell)$.
If and when a leaf is reached, return~0.
\item[\normalfont\Tvn{iterate}.\Tvn{next}:]
If $\Tvn{iterate}.\Tvn{more}=0$, return~0.
Otherwise proceed as follows:

If $\ell=0$, start at $r$ and, as long as the current node
$u$ is not of height~1,
step to the $i$th child $v$ of $u$,
where $i=D_u.\Tvn{iterate}.\Tvn{next}$,
and execute $D_v.\Tvn{iterate}.\Tvn{init}$.

If $\ell>0$, instead start at $r$ and, as long as the current node $u$
is not of height~1, step to the $i$th child of $u$,
where $i=\Tvn{viachild}(u,\ell)$.
Then, as long as $D_u.\Tvn{iterate}.\Tvn{more}=0$,
where $u$ is the current node, step to
the parent of~$u$.
Subsequently, as long as the current node $u$
is not of height~1,
step to the $i$th child $v$ of $u$,
where $i=D_u.\Tvn{iterate}.\Tvn{next}$,
and execute $D_v.\Tvn{iterate}.\Tvn{init}$.

Whether or not $\ell=0$,
when a node $u$ of height~1 is reached,
let $v$ be its $i$th child,
where $i=D_u.\Tvn{iterate}.\Tvn{next}$,
set $\ell:=\Tvn{leftindex}(v)$ and return~$\ell$.
\end{description}

\noindent
If we say that the choice dictionary of
a node $u$ in $T$ is \emph{activated}
through a call of $D_u.\Tvn{iterate}.\Tvn{init}$,
becomes \emph{exhausted} when
$D_u.\Tvn{iterate}.\Tvn{more}$ first evaluates to~0,
and is \emph{active} between the two events,
the iteration procedure above maintains a single
root-to-leaf path of active choice dictionaries,
which it remembers in the integer~$\ell$, with
$\ell=0$ denoting an initial situation in which
such an \emph{active path} has not yet been
established.
The global call $\Tvn{iterate}.\Tvn{next}$
finds a first active path (if $\ell=0$) or
(if $\ell>0$) exhausts the choice dictionaries of 
the current active path in a bottom-up fashion
until reaching a node $u$ with
$D_u.\Tvn{iterate}.\Tvn{more}=1$,
then steps to the ``next'' child $v$ of $u$ and
changes the last part of the current path to
be the path from $v$ to its ``first'' leaf descendant.

If the choice dictionaries of some nodes in $T$
support \Tvn{successor} (or \Tvn{predecessor})
instead of \Tvn{iterate}, the overall choice dictionary
can still support iteration through the reduction
of \Tvn{iterate} to
\Tvn{successor} (or \Tvn{predecessor})
described in Section~\ref{sec:preliminaries}.
This needs additional space for a set of
``state variables'' that record the active path,
but it is easy to see that
the single variable $\ell$ can
represent these compactly,
so that the overall space cost of an iteration
is $\Tceil{\log(n+1)}$ bits.

If the nodes of height~1 in $T$ have
$c$-color choice dictionaries, for some $c\ge 2$,
the overall choice dictionary can also support
$c$ colors and therefore maintain a client vector
$(S_0,\ldots,S_{c-1})$.
In this case we equip every node $u$ in $T$
of height $\ge 2$ with $c$
choice dictionaries, each with universe size
equal to the degree of~$u$ and associated
with a different color in $\{0,\ldots,c-1\}$.
Conceptually, the choice dictionaries
associated with each color $j\in\{0,\ldots,c-1\}$
form an \emph{upper tree} $T_j$ that realizes a 
choice dictionary $D_j$, called the
choice dictionary of $T_j$, with
universe size $n_1=\Tceil{n/{p_1}}$
and with client set
$\{i\in\TbbbN\mid 1\le i\le n_1$ and
$S_j\cap\{(i-1)p_1+1,\ldots,i p_1\}\not=\emptyset\}$.
For $j\in\{1,\ldots,c-1\}$ the choice dictionaries
in $T_j$ and $D_j$ are all colorless.
Because $S_0=U$ initially, the choice dictionaries
in $T_0$ and $D_0$ must instead allow
two colors and use the elements of color~0
as their ``client set''.
For $i=1,\ldots,n_1$,
the $i$th leaves of all of $T_0,\ldots,T_{c-1}$
are associated with the same $i$th
\emph{lower tree}, the tree induced
by the $i$th node of height~1 in $T$ and the children of that node.
An additional colorless dictionary $D^*$
with universe size $n_1+c$ is used to keep
track of which choice dictionaries of upper and lower trees
have been initialized.
The realization of the ``colored'' choice-dictionary
operations in terms of ``colorless'' operations on
upper trees and ``colored'' operations on lower
trees is easy.
For instance,
to execute $\Tvn{choice}(j)$,
call $D_j.\Tvn{choice}$ to find a lower tree $\widetilde{T}$
in which the color $j$ is ``represented'' and
call $\Tvn{choice}(j)$ in the choice dictionary
of (the root of) $\widetilde{T}$ to determine an element of $S_j$.
To execute $\Tvn{color}(\ell)$, consult the appropriate
lower tree $\widetilde{T}$.
In all cases, before operating on the
dictionary $D$ of an upper or
lower tree, use $D^*$ to initialize $D$
if this has not been done before.
The remaining details are left to the reader.
When putting together a choice dictionary as described
in this section, we will say that we apply the
\emph{trie-combination} method.

\section{Systematic and Related Choice Dictionaries}
\label{sec:systematic}%

\subsection{Upper Bounds}

One is frequently faced with the problem of
maintaining a permutation $\pi$ of $\{1,\ldots,n\}$
initialized to the identity permutation of that
set, say, under inspection of function values
and updates of $\pi$ of some kind.
Allowing an initialization time of $\Theta(n)$,
the problem is trivial.
Assume that we want the initialization time to be constant.
Proceeding as described after Lemma~\ref{lem:2.12},
we can maintain $\pi$ using around $n\log n$ bits
for the values of $\pi$ itself and $2 n\log n$ bits 
for its ``initialization on the fly'' component.
If the inverse permutation $\pi^{-1}$ is also
maintained in the same manner, the space
requirements grow to approximately
$6 n\log n$ bits.
In the following lemma we demonstrate how to maintain
both $\pi$ and $\pi^{-1}$ using only about
a third of this space.
Our data structure shows some
similarity to an algorithm of Brassard and Kannan
for computing random permutations
``on the fly''~\cite{BraK88}.

The data structure of
Lemma~\ref{lem:permutation}
must be employed with a little care because the
user acquires full ``control'' over $\pi$ only
gradually in the course of
$n$ calls of an operation \Tvn{consolidate}.
More precisely, when $r\le n$ calls of \Tvn{consolidate}
have been executed, the value of $\pi$ after
an update, which is supposed to ``rotate'' 
the function values within a
given subset of $\{1,\ldots,n\}$,
is in fact known only on the $r$ largest elements
of $\{1,\ldots,n\}$.
One way of coping with the associated uncertainty
is illustrated in the
proof of Theorem~\ref{thm:nlogn}.

The space savings by a factor of~3 discussed
above plays no role in our development after
Theorem~\ref{thm:nlogn}, but we consider
Lemma~\ref{lem:permutation} to be
of independent interest.

\begin{lemma}
\label{lem:permutation}%
There is a data structure with the following properties:
First, for every $n\in\TbbbN$, it can be
initialized for universe size $n$ and
subsequently maintains a pair $(\pi,\mu)$
composed of a permutation $\pi$ of $U=\{1,\ldots,n\}$,
initially the identity permutation
$\Tvn{id}_n$ of $U$, and an integer $\mu$, initially $n$,
under evaluation of $\pi$ and $\pi^{-1}$
and the following operations:

\begin{description}
\item[{\normalfont \Tvn{consolidate}:}]
Replaces $\mu$ by $\max\{\mu-1,0\}$.
\item[{\normalfont $\Tvn{rotate}(j_1,\ldots,j_k)$}]
($k\in\TbbbN$ and $j_1,\ldots,j_k$ are distinct
elements of $U$):
Replaces $\pi$ by a permutation of $U$ that agrees
on $\{\mu+1,\mu+2,\ldots,n\}$ with the permutation
$\pi'$ of $U$ with $\pi'(j_i)=\pi(j_{i+1})$ for
$i=1,\ldots,k-1$, $\pi'(j_k)=\pi(j_1)$,
and $\pi'(\ell)=\pi(\ell)$ for all
$\ell\in U\setminus\{j_1,\ldots,j_k\}$.
\end{description}

\noindent
Second, for known $n$, the data structure uses at most
$(2 n+1)\Tceil{\log n}$ bits,
can be initialized in constant time, executes queries
and calls of $\Tvn{consolidate}$ in constant time and
executes $k$-argument calls of $\Tvn{rotate}$
in $O(k)$ time, for all $k\in\TbbbN$.
\end{lemma}

\begin{proof}
The permutation $\pi$ is represented through two
arrays $P[1\Ttwodots n]$ and $P^{-1}[1\Ttwodots n]$,
each of whose entries can hold an arbitrary element of~$U$.
For $\ell\in U$, say that $\ell$ is \emph{proper}
in $P$ if $P[\ell]\in U$, $P^{-1}[P[\ell]]=\ell$, and
$\max\{\ell,P[\ell]\}>\mu$.
Correspondingly, $\ell$ is proper in $P^{-1}$
if $P^{-1}[\ell]\in U$, $P[P^{-1}[\ell]]=\ell$, and
$\max\{\ell,P^{-1}[\ell]\}>\mu$.
If some $\ell\in U$ is not proper in $P$ or $P^{-1}$,
we say that $\ell$ is \emph{improper} in that array.
Observe that if $\ell\in U$ is proper in~$P$,
then $P[\ell]$ is proper in $P^{-1}$.
When $\ell$ is improper in $P$, say, $P[\ell]$
may contain an arbitrary value
(``be uninitialized'').
The following invariant will hold at all times
between operations:
For all $\ell\in U$, $\ell$ is proper in $P$
if and only if $\ell$ is proper in $P^{-1}$;
for $\ell=\mu+1,\ldots,n$, $\ell$ is proper in both $P$ and $P^{-1}$.
When saying simply that $\ell$ is proper,
we will mean that $\ell$ is proper in both $P$ and $P^{-1}$.
The arrays $P$ and $P^{-1}$ represent a
permutation $\pi$ of $U$ in the following manner:
For $\ell\in U$, if $\ell$ is proper,
then $\pi(\ell)=P[\ell]$; if not, $\pi(\ell)=\ell$.
To see that this really defines $\pi$ as a
permutation of $U$, let $W=\{\ell\in U\mid \ell$ is proper$\}$
and observe that $\pi$ is a function from $U$ to $U$ that
maps $W$ to $W$ and is injective both on $W$
(because $P^{-1}[\pi(\ell)]=\ell$ for each $\ell\in W$)
and on $U\setminus W$.
It is easy to see that
$\pi$ and $\pi^{-1}$ can be evaluated in
constant time on arbitrary arguments in~$U$.
Informally, $\ell$ is proper in $P$
and $P[\ell]=\pi(\ell)$ if $P[\ell]$ is a ``plausible''
value for $\pi(\ell)$ (i.e., $P[\ell]\in U$)
and that value is confirmed by $P^{-1}$
(i.e., $P^{-1}[P[\ell]]=\ell$).
However, only values of $P[\ell]$ and $P^{-1}[\ell]$
with $\ell>\mu$ are considered trustworthy,
and if both $\ell$ and $P[\ell]$ are $\le \mu$,
$\ell$ is improper and $P[\ell]$ is ignored.
Initially, the invariant is
satisfied, and the permutation $\pi$
represented through $P$ and $P^{-1}$ is
the identity permutation $\Tvn{id}_n$.

To execute \Tvn{consolidate} when $\mu>0$,
store $\mu$ in both $P[\mu]$ and $P^{-1}[\mu]$
if $\mu$ is improper.
Then, whether or not $\mu$ is proper, decrement $\mu$.
It can be seen that neither step
invalidates the invariant or changes~$\pi$.

The implementation of
$\Tvn{rotate}$ is illustrated in
Fig~\ref{fig:rotation}.
To execute $\Tvn{rotate}(j_1,\ldots,j_k)$
in the situation of
Fig.~\ref{fig:rotation}(a),
let $J=\{j_1,\ldots,j_k\}$ and begin by
setting $P[j]:=j$ for each improper $j\in J$
(Fig.~\ref{fig:rotation}(b)).
Then change $P$ in a way that
reflects the permutation $\pi'$ in the
definition of \Tvn{rotate}:
Save $P[j_1]$ in a temporary variable,
then, for $i=1,\ldots,k-1$,
execute $P[j_i]:=P[j_{i+1}]$, and
next store the original value of $P[j_1]$ in $P[j_k]$.
Subsequently change $P^{-1}$ accordingly
by setting $P^{-1}[P[j]]:=j$ for all $j\in J$.
At this point $P[j]=\pi'(j)$ for all $j\in J$,
but the invariant may be violated
(Fig.~\ref{fig:rotation}(c)).

\begin{figure}
\begin{center}
\epsffile{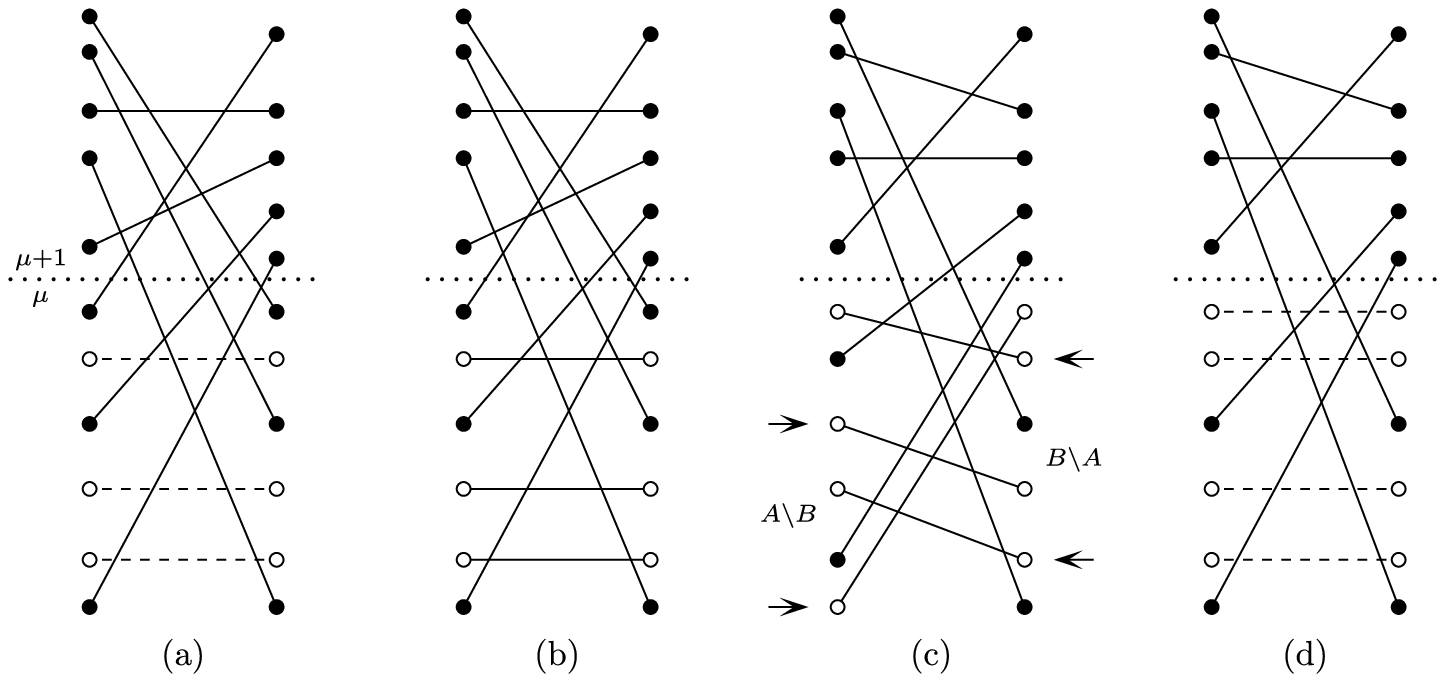}
\end{center}
\caption{The execution of $\Tvn{rotate}(j_1,\ldots,j_k)$, step
by step, starting from an example permutation~$\pi$.
The example is chosen to have $j_1>\cdots>j_k$.
Each of parts (a)--(d) shows $J=\{j_1,\ldots,j_k\}$ on the left
and $\pi(J)$ on the right.
The sets $\{1,\ldots,\mu\}$ and $\{\mu+1,\ldots,n\}$
are separated by a dotted line.
(a): The initial situation.
For each $j\in J$, $j$ and $\pi(j)$ are connected by a fully
drawn line if $j$ is proper (then $P[j]=\pi(j)$) and by a dashed line if
$j$ is improper (then $P[j]$ may be arbitrary).
(b): After the execution of $P[j]:=j$ for each improper $j\in J$.
Each $j\in J$ is connected to $P[j]$.
(c): After the actual rotation.
Now $P[j]=\pi'(j)$ for all $j\in J$,
where $\pi'$ is as in Lemma~\ref{lem:permutation}.
The elements of $A\setminus B$ and $B\setminus A$
are indicated by arrows.
(d): After the restoration of the invariant.
The final permutation is shown
with conventions as in part~(a).}
\label{fig:rotation}
\end{figure}

Let us say that the invariant is violated
at some $\ell\in U$ if $\ell\le \mu$ and $\ell$ is
improper in exactly one of $P$ and $P^{-1}$ or $\ell>\mu$ and $\ell$ is
improper in at least one of $P$ and $P^{-1}$.
Let $J'=J\cap\{1,\ldots,\mu\}$, 
$A=\{j\in J'\mid P[j]\le \mu\}$ and $B=\{P[j]\mid j\in A\}$.
Obviously $|A|=|B|$.
It can be seen that the invariant is not
violated at any element outside of $A\cup B$.
Moreover, for $j\in A\cup B$, $j$ is improper in $P$
exactly if $j\in A$, whereas $j$ is improper in $P^{-1}$
exactly if $j\in B$.
Therefore the invariant is violated exactly at each $j$
in the symmetric difference of $A$ and~$B$.
Observe that $|A\setminus B|=|B\setminus A|$ and finally
restore the invariant by changing the values of $P$
on $A\setminus B$ to make $P$
map $A\setminus B$ injectively to $\pi'(B\setminus A)$
and then setting $P^{-1}[P[j]]:=j$ for all $j\in A\setminus B$.
This simultaneously makes the elements
of $A\setminus B$ proper in $P$ and makes the elements
of $B\setminus A$ improper in $P$
(Fig.~\ref{fig:rotation}(d)).

The data structure uses slightly
more space than claimed because $\mu$ can take
arbitrary values in $\{0,\ldots,n\}$ and so needs
$\Tceil{\log(n+1)}$ bits for its storage.
To lower this to $\Tceil{\log n}$ bits, execute
\Tvn{consolidate} one first time already as part
of the initialization, so that $\mu$ never has
the value~$n$.
\end{proof}

Recall from Section~\ref{sec:quick} that
our main result about systematic choice dictionaries
is
obtained by the combination of three simple ingredients:
A choice dictionary that is wasteful in terms
of space (Theorem~\ref{thm:nlogn}), a choice
dictionary for very small universes
(Lemma~\ref{lem:atomic-c}), and the
trie-combination method of
Section~\ref{sec:trie}.

\begin{theorem}
\label{thm:nlogn}%
There is a choice dictionary that,
for arbitrary $n,c\in\TbbbN$,
can be initialized for universe size~$n$
and $c$ colors in constant time
and subsequently
occupies at most
$(2 n+4 c)\Tceil{\log(n+1)}+n\Tceil{\log(2 c)}
 +O(c\log c)=O((n+c)\log(n+c))$
bits and supports \Tvn{color}, \Tvn{p-rank}, \Tvn{p-select}
(and hence \Tvn{choice} and \Tvn{uniform-choice})
and robust iteration in constant time 
and \Tvn{setcolor} in $O(c)$ time.
A more precise time bound for \Tvn{setcolor} is
that the execution of a call $\Tvn{setcolor}(j',\ell)$,
for all $j'\in\{0,\ldots,c-1\}$ and $\ell\in\{1,\ldots,n\}$,
takes $O(|j'-j|+1)$ time, where $j$ is the color of $\ell$
immediately before the call.
\end{theorem}

\begin{proof}
Denote the client vector by $(S_0,\ldots,S_{c-1})$.
The choice dictionary maintains a semipartition
$(R_0,\ldots,R_{2 c-1})$
of $U=\{1,\ldots,n\}$, whose sets
will be called \emph{segments}.
The intended meaning of the segments is that
for $j=0,\ldots,c-1$, $R_{2 j}$ and $R_{2 j+1}$
are the sets of those elements of $S_j$ that are (still)
to be enumerated and are not to be enumerated, respectively,
in the current iteration over $S_j$, if any;
thus at all times $R_{2 j}\cup R_{2 j+1}=S_j$.
For brevity, let us say that the elements in $R_k$
are of \emph{hue} $k$, for $k=0,\ldots,2 c-1$, and
denote the hue of each $\ell\in U$ by $\Tvn{hue}(\ell)$.
The segments are realized via $2 c$ integers
$m_0,\ldots,m_{2 c-1}$ that store $|R_0|,\ldots,|R_{2 c-1}|$,
respectively, and a pair $(\pi,\mu)$,
where $\pi$ is a permutation of~$U$ and
$\mu\in\{0,\ldots,n\}$,
together with
the convention that
$\pi$ sorts the elements of $U$ by hue, i.e.,
$\Tvn{hue}(\pi(1))\le\cdots\le\Tvn{hue}(\pi(n))$.
We also maintain the prefix sums
$s_k=\sum_{i=0}^k m_i$, for $k=-1,\ldots,2 c-1$,
and the hue of each element of $U\setminus R_0$
explicitly in two arrays, so that $\Tvn{hue}(\ell)$
can be determined in constant time for each $\ell\in U$.
The invariant $\mu\le m_0$ will hold at all times.

The pair $(\pi,\mu)$ is maintained in an instance
$D$ of the data structure of Lemma~\ref{lem:permutation}.
$D$'s \Tvn{rotate} operation can be used
to move elements from one segment to another.
E.g., to move an element $\ell$ from $R_j$ to
$R_{j'}$, where $j<j'$,
execute $D.\Tvn{rotate}(j_1,\ldots,j_k)$,
where $(j_1,\ldots,j_k)$ is the sequence obtained
from $(\pi^{-1}(\ell),s_j,s_{j+1},\ldots,s_{j'-1})$
by eliminating duplicates, i.e., by removing every
element equal to an earlier element,
and subsequently decrement $m_j$
and each of $s_j,\ldots,s_{j'-1}$ and increment $m_{j'}$.
This takes $O(j'-j)$ time.
Note how the condition $\mu\le m_0$ prevents
unintended transfers of elements from one segment to another
by the \Tvn{rotate} operation.
The operations of the choice dictionary
are implemented as follows:

\begin{description}
\item[\normalfont $\Tvn{color}(\ell)$:]
Return $\Tfloor{{{\Tvn{hue}(\ell)}/2}}$.
\item[\normalfont $\Tvn{setcolor}_j(\ell)$:]
If $\Tvn{color}(\ell)\not=j$,
then execute $D.\Tvn{consolidate}$ and subsequently move $\ell$
from its current segment to $R_{2 j+1}$
and record $\Tvn{hue}(\ell)=2 j+1$.
\item[\normalfont $\Tvn{p-rank}(\ell)$:]
Return $\pi^{-1}(\ell)-s_{2 j-1}$, where $j=\Tvn{color}(\ell)$.
\item[\normalfont $\Tvn{p-select}_j(k)$:]
Return $\pi(s_{2 j-1}+k)$ if
$1\le k\le |S_j|=m_{2 j}+m_{2 j+1}$, and 0 otherwise.
\item[\normalfont $\Tvn{iterate}_j.\Tvn{init}$:]
Merge $R_{2 j+1}$ into $R_{2 j}$, i.e.,
execute first $m_{2 j}:=m_{2 j}+m_{2 j+1}$
and $s_{2 j}:=s_{2 j+1}$ and then $m_{2 j+1}:=0$.
\item[\normalfont $\Tvn{iterate}_j.\Tvn{more}$:]
Return 1 if $R_{2 j}\not=\emptyset$, i.e., if $m_{2 j}>0$,
and 0 otherwise.
\item[\normalfont $\Tvn{iterate}_j.\Tvn{next}$:]
Return $0$ if $\Tvn{iterate}.\Tvn{more}=0$.
Otherwise execute $D.\Tvn{consolidate}$ and subsequently
move the boundary between $S_{2 j}$ and $S_{2 j+1}$ backward
by one element and return the element that crosses the boundary.
In other words, decrement $m_{2 j}$ and $s_{2 j}$, increment
$m_{2 j+1}$ and return $\pi(s_{2 j}+1)$.
\end{description}

The initialization
sets $m_0:=n$,
$m_j:=0$ for $j=1,\ldots,2 c-1$,
$s_{-1}:=0$ and $s_j:=n$ for $j=0,\ldots,2 c-1$
and initializes $D$.
To achieve a constant
initialization time, use Lemma~\ref{lem:2.12}.
After the initialization $R_0=U$,
$R_1=\cdots=R_{2 c-1}=\emptyset$,
$\mu=n$ and $\pi$ is the identity permutation $\Tvn{id}_n$,
so the client vector represented is
$(U,\emptyset,\ldots,\emptyset)$,
as required, and the invariant is satisfied.
The only operations that may decrease $m_0$
are \Tvn{setcolor} and $\Tvn{iterate}.\Tvn{next}$,
and the decrease is only by~1.
Both operations call
$D.\Tvn{consolidate}$ before they carry out any
other change, so
the invariant $\mu\le m_0$ is always satisfied.
Only the operation $\Tvn{setcolor}$
calls $D.\Tvn{rotate}$, and therefore the elements
returned by calls of $\Tvn{p-rank}$ and
\Tvn{p-select} are consistent
with bijections that do not change between
calls of \Tvn{setcolor}.
Storing elements that are moved to $S_j$
in $R_{2 j+1}$ rather than in $R_{2 j}$
prevents the elements from being enumerated
more than once during an iteration over~$S_j$.
Therefore the iterations over $S_j$ are robust.

An accurate count of the size of the data
structures introduced above yields an upper bound of
$(2 n+4 c+2)\Tceil{\log(n+1)}+n\Tceil{\log(2 c)}
 +O(c\log c)$ bits.
To this should be added a number of bits needed
to store the parameters $n$ and $c$.
On the other hand, we can omit every second
prefix sum $s_i$, so the space bound stated
in the theorem is easily achievable.
\end{proof}

\begin{lemma}
\label{lem:atomic-c}%
There is a choice dictionary that, for arbitrary
$n,c\in\TbbbN$, can be initialized for universe size~$n$
and $c$ colors in $O(1+{{(n\log c)}/w})$ time
and subsequently occupies $n\Tceil{\log c}$
bits and executes \Tvn{color} and \Tvn{setcolor}
in constant time and \Tvn{successor} and
\Tvn{predecessor} (and hence also \Tvn{choice})
in $O(1+{{(n\log c)}/w})$ time.
\end{lemma}

\begin{proof}
Store only the $n$ color values, each in a
field of $\Tceil{\log c}$ bits.
The realization of $\Tvn{color}(\ell)$ and
$\Tvn{setcolor}_j(\ell)$ is obvious---read and overwrite
the contents of the $\ell$th field,
respectively.
To execute $\Tvn{successor}_j(\ell)$ for
$j\in\{0,\ldots,c-1\}$ and $\ell\in\{0,\ldots,n\}$,
remove the $\ell$ leftmost fields in a copy, replace
the value in
every remaining field by its bitwise
\textsc{xor} with $j$, and use an algorithm
of Lemma~\ref{lem:word}(b).
The implementation of $\Tvn{predecessor}$ is
analogous.
\end{proof}

\Fpasteinsection{\systematic2}

\begin{proof}
Take $k=t w$ and assume without loss of generality that $k\ge 2$.
Compute $q\in\TbbbN$ so that
$q\ge n^{1/t}$, but $q=O(n^{1/t})$.
We compose the choice dictionaries of
Theorem \ref{thm:nlogn}
and Lemma~\ref{lem:atomic-c}, both
initialized for 2 colors,
with the trie-combination method of
Section~\ref{sec:trie} and with the degree
sequence $(p_1,p_2,\ldots)$,
where $p_1=p_2=k$,
$p_3=\Theta(\log n)$, and
$p_j=q$ for $j\ge 4$.
Every inner node of height at most~3 in the resulting trie $T$
is equipped with an instance of the choice dictionary
of Lemma~\ref{lem:atomic-c},
while every node in $T$
of height at least~4
has an instance of
the choice dictionary of Theorem~\ref{thm:nlogn}.
Every operation
on the overall choice dictionary spends $O(t)$
time on each of the three bottom levels of $T$ above the leaves
and constant time on every other level.
Since the height of $T$ is $O(t)$, this
sums to~$O(t)$.
The choice dictionaries of the nodes in $T$
of height $1$ need a total of exactly $n$ bits,
and the most natural layout ensures that the
overall dictionary is systematic.
The height-2 and height-3 choice dictionaries,
if present, need
$\Tceil{n/k}$ bits and $O({n/{k^2}})$ bits,
respectively.
The number of nodes in $T$ of height $4$
is $O(({{{n/{k^2}})}/{\log n}})$, so
the number of bits required for all instances
of the dictionary of Theorem~\ref{thm:nlogn}
is $O({n/{k^2}})$.
\end{proof}

If we allow 
the dictionary not to be systematic,
we can generalize to several colors and
obtain an additional space bound that depends
on the maximum size of the client set.
In order to support $c$ simultaneous iterations,
one for each color, the theorem below requires
$O(c\log n)$ additional bits.
In general, with enough additional space to keep
track of their states, a smaller or larger number
of simultaneous iterations can be supported,
here and in data structures described later.

\begin{theorem}
\label{thm:m-c}%
There is a choice dictionary
that, for arbitrary $n,c,t,k\in\TbbbN$
with $k\log c=O(t w)$,
can be initialized
for universe size $n$, $c$ colors
and tradeoff parameters $t$ and $k$
in constant time and subsequently uses
$n\Tceil{\log_2 c}+{{c n}/k}+O({{c n}/{k^2}}+\log n)$
bits of memory
and supports \Tvn{color}, \Tvn{setcolor},
\Tvn{choice} and,
given $O(c\log n)$ additional bits, robust
iteration in $O(t)$ time.
Moreover, as long as the number of elements
of nonzero color
remains bounded by $m\in\TbbbN$,
the number of bits of memory used is
$O(((t+c)n^{1/t}+c k^2)m\log n)$.
In particular, for every fixed $\epsilon>0$, there is
a choice dictionary that executes all
operations in constant time and uses
$O(c m n^\epsilon+1)$ bits to store
semipartitions that never have more than~$m$
elements of nonzero color.
\end{theorem}

\begin{proof}
If $n<8 k$, the result follows
from Lemma~\ref{lem:atomic-c}.
Assume therefore that $n\ge 8 k$.
We use largely the same construction as in the previous proof
and with $p_3=\Theta(\log n)$ and $q$ chosen as there, but 
now for general values of $k$ and with $p_1=p_2=4 k$
instead of $p_1=p_2=k$.
There are two additional changes:

First, the choice
dictionaries of nodes of height~1
in the trie $T$ are initialized for $c$ rather
than~2 colors and, as detailed in
Section~\ref{sec:trie}, each choice dictionary
of a node of height 2 or more in $T$
is replaced by $c$ independent 2-color
choice dictionaries, one for each color.
As also discussed in Section~\ref{sec:trie},
this change makes it necessary to keep track
explicitly of the initialization of
upper and lower trees.
Instead of using a single dictionary $D^*$
of universe size $\Tceil{n/{p_1}}+c$ as
suggested in Section~\ref{sec:trie},
we handle the initialization of the
$c$ upper trees in a separate choice
dictionary with universe size~$c$
(realized according to Theorem~\ref{thm:systematic-2}, say)
and equip each node $u$ of height~2 with
a colorless instance of the choice dictionary 
of Lemma~\ref{lem:atomic-c}
that records the initialization
of the choice dictionaries at $u$'s children.
The total number of bits needed for the
dictionaries that take the place of $D^*$
can be bounded by
$\Tceil{n/{p_1}}+2 c$.

Second, rather than reserving space
permanently for every choice dictionary,
we allocate space to the $c$ choice dictionaries
of a node $u$ of height $\ge 3$ in $T$
only when one of them acquires its first element
($u$ becomes \emph{nonempty}) and reclaim that
space if and when $u$ returns to being empty.
When space for the choice dictionaries of a node $u$ of
height~$\le 3$ is allocated, we also allocate space for the
choice dictionaries of all children of~$u$
(and, recursively, for those of their children).

If $n\le (4 k)^2$, the height of $T$ is bounded by~2,
and its choice dictionaries can be accommodated
in a total of 
$n\Tceil{\log c}+(c+1)\Tceil{n/{p_1}}+2 c
\le n\Tceil{\log c}+2 c({n/{(4 k)}}+1)+2 c
=n\Tceil{\log c}+c({n/{(2 k)}}+4)
\le n\Tceil{\log c}+{{c n}/k}$ bits,
a bound easily seen to be covered
by those of the theorem (recall that $k^2=\Omega(n)$).
In the rest of the proof assume that $n>(4 k)^2$,
so that $T$ is of height at least~3.
 
The total number of bits needed by the choice dictionaries
of the descendants of a node of height~3
is $s\Tsub L=p_1 p_2 p_3\Tceil{\log c}+(c+1)p_2 p_3+c p_3$,
and these choice dictionaries are accommodated in
a \emph{leaf chunk} of $s\Tsub L=O(c k^2\log n)$ bits.
An exception concerns the descendants of
the rightmost node of height~3,
whose choice dictionaries may need less space;
exactly the required number of bits
is set aside statically for these dictionaries.
In the interest of simplicity, let us ignore
this exception for most of the following
discussion and return to it briefly at the end of the proof.
The $c$ choice dictionaries of a node $u$ of height $\ge 4$
occupy $O(c q\log q)$ bits.
Because the neighbors of~$u$ in $T$ are no longer
stored in fixed places in memory, the representation
of $u$ must be augmented by $q+1$ explicit pointers
of $O(\log n)$ bits each that allow
navigation in~$T$.
Altogether, $u$ and its choice dictionaries can be
accommodated in an \emph{inner chunk} of
$s\Tsub I=O(q\log n+c q\log q)=
O((1+{c/t})q\log n)$ bits.

The total number of nodes of height~3 in $T$
is $n\Tsub L=\Tceil{n/{(p_1 p_2 p_3)}}$, and
the total number $n\Tsub I$ of nodes of height
$\ge 4$ can be computed in $O(t)$ time.
Accordingly, the available memory is conceptually
partitioned into $n\Tsub L$ \emph{leaf slots}
of $s\Tsub L$ bits each
and $n\Tsub I$ \emph{inner slots}
of $s\Tsub I$ bits each.
When space for a chunk is needed, a free slot of
the right size is allocated to it, and returned
slots are kept in one of two \emph{free lists},
one for each chunk size, that can
easily be maintained in the free slots themselves.
When a free slot is requested, it is taken from
the relevant free list unless the latter is empty.
If the relevant free list is empty, the first
slot of the right size and
unused so far is put into service; two
simple variables suffice to keep track of the
borders between slots that were allocated at least
once and new slots.

A leaf slot is exactly as large as the choice dictionaries
that may be stored in the slot.
An inner slot is larger by the $O(q\log n)$ bits
for pointers to other slots, but since the number
of inner slots is $O({n/{(q k^2\log n)}})$,
the total number of
additional bits is $O({n/{k^2}})$.
Therefore the number of bits used by the entire
data structure never exceeds
$n\Tceil{\log c}+(c+1)\Tceil{n/{(4 k)}}+O({{c n}/{k^2}}+\log n)
=n\Tceil{\log c}+{{c n}/k}+O({{c n}/{k^2}}+\log n)$.

As long as the number of elements with
nonzero colors remains bounded by~$m\in\TbbbN$,
the data structure allocates at most the first
$t m$ inner slots and the first $m$ leaf slots.
The total number of bits in these slots is
$t m s\Tsub I+m s\Tsub L=
O(t m(1+{c/t})q\log n+c m k^2\log n)=
O(((t+c)n^{1/t}+c k^2)m\log n)$.
The slots cannot be packed tightly because they
are allocated from two different pools, but we can
still ensure that they come from a block of
memory of $O(((t+c)n^{1/t}+c k^2)m\log n)$ bits by
laying out the slots in memory according to
the following pattern:
First a leaf slot, then $t$ inner slots, then
again a leaf slot, and so on.
The space bound easily admits the few choice
dictionaries that were allocated statically above.
\end{proof}

\subsection{A Lower Bound for Systematic Choice Dictionaries}
\label{subsec:lower}%

In this subsection we show that the systematic choice
dictionary of Theorem~\ref{thm:systematic-2} is optimal,
up to a constant factor, in the tradeoff that it offers
among redundancy, execution time and word length.

For all integers $h,r,s$, let an
\emph{$(h,r,s)$-language} be a
language $L$ over
$\Sigma=\{\mathtt{a}_0,\mathtt{a}_1,\mathtt{b}_0,\mathtt{b}_1\}$
that does not contain two words of the
form $u a v_1$ and $u b v_2$ with
$u,v_1,v_2\in\Sigma^*$, $a\in\{\mathtt{a}_0,\mathtt{a}_1\}$
and $b\in\{\mathtt{b}_0,\mathtt{b}_1\}$ and
for which each
$u\in L$ satisfies $|u|=h$, $|u|_{\mathtt{a}_1}\le r$
and $|u|_{\mathtt{b}_0}+|u|_{\mathtt{b}_1}\le s$.
Here $|u|_{\mathtt{a}_1}$, e.g., denotes the number
of occurrences of the character $\mathtt{a}_1$ in~$u$.
Let $\TbbbN_0=\TbbbN\cup\{0\}$.

\begin{lemma}
\label{lem:counting}%
For all integers $h,r,s$, the cardinality
of every $(h,r,s)$-language 
is bounded by $2^{s'}{{h-s'+r}\choose r}$,
where $s'=\min\{h,s\}$.
\end{lemma}

\noindent\textbf{Proof.}
For all integers $h,r,s$, let
$N(h,r,s)$ be the maximum cardinality of
an $(h,r,s)$-language.
The bound of the lemma can be shown by induction
on $h$ using the recurrence
\[
N(h,r,s)=
\begin{cases}
0,&\mbox{if $h<0$ or $r<0$ or $s<0$};\cr
1,&\mbox{if $h=0$ and $r,s\ge 0$};\cr
\eqalign{\max\{2&N(h-1,r,s-1),\cr
\noalign{\vglue -1pt}
&N(h-1,r,s)+N(h-1,r-1,s)\}\cr}&
 \mbox{if $h>0$ and $r,s\ge 0$}.\cr
\end{cases}
\eqno{\lower 8mm\hbox{$\Box$}}
\]

\begin{theorem}
\label{thm:lower}%
Let $n,s,t\in\TbbbN_0$, let $n\ge 2$ and assume that
some systematic data structure $D$
can represent every subset of $U=\{1,\ldots,n\}$
in a sequence $B$ of $n+s$ bits.
Assume further that for each $r\in\{1,\ldots,n\}$,
it is possible to distinguish
among the $n\choose r$ subsets $S$ of $U$
of size $r$ with at most $r t$ bit probes to an
arbitrary sequence $B(S)$ that represents each set~$S$
according to $D$'s conventions.
Then $(s+{1/{\ln 2}})t\ge {n/{(e\ln 2)}}$
and, if $s>0$, $s t\ge{n/2}$.
\end{theorem}

\begin{proof}
The second assumption of the theorem
cannot hold for $t=0$.
We can therefore assume without loss
of generality that $t\ge 1$ and that $s\le n$.
For an $r\in\{1,\ldots,n\}$ to be chosen later,
we associate a word $u_S$ over
$\Sigma$ with each
$S\in\mathcal{S}=\{S\subseteq\{1,\ldots,n\}:|S|=r\}$.
Let $\mathcal{A}$ be an algorithm that can
distinguish among the sets in $\mathcal{S}$
with at most $r t$ bit probes.
Without loss of generality, $\mathcal{A}$
probes no bit more than once.
For each $S\in\mathcal{S}$, we apply $\mathcal{A}$
to a bit sequence $B(S)$ used by $D$ to represent~$S$
and chosen to be \emph{minimal} in the sense
that no sequence $B\not=B(S)$ of $n+s$ bits
that also represents $S$ satisfies $B\le B(S)$,
where $\le$ denotes the conjunction of
bitwise $\le$ in all bit positions.
Informally, the minimality of $B(S)$ implies
that every uninitialized bit in $B(S)$
has the value~0.
Without loss of generality, assume that the first
$n$ bits of $B(S)$ are the bits $b_1,\ldots,b_n$
referred to in the definition of a systematic
data structure (informally, the bit-vector
representation of~$S$).
For each $S\in\mathcal{S}$, $u_S$ is obtained as follows:
Initialize $u_S$ to be the empty word
and append a character to $u_S$
at each probe carried out by $\mathcal{A}$
on input $B(S)$, choosing the character as $c_i$,
where $i\in\{0,1\}$ is the value of the bit
probed, $c=\mathtt{a}$ if the bit probed is
among the first $n$ bits of $B(S)$, and 
$c=\mathtt{b}$ if the bit probed is
among the last $s$ bits of $B(S)$.
At this point, since $\mathcal{A}$ uses at
most $h=r t$ probes, $|u_S|\le h$.
Finally increase $|u_S|$ to exactly $h$ 
by appending $h-|u_S|$ occurrences of $\mathtt{a}_0$ to $u_S$.

For each
$S\in\mathcal{S}$, with $u=u_S$,
$\mathcal{A}$ probes each bit at most once, and so
$|u|_{\mathtt{a}_1}\le r$
since $B(S)$ is minimal and $|S|=r$, and
$|u|_{\mathtt{b}_0}+|u|_{\mathtt{b}_1}\le s$
since there are only $s$ bits in addition
to $b_1,\ldots,b_n$.
$L=\{u_S\mid S\in\mathcal{S}\}$ is therefore an
$(h,r,s)$-language.
For $S_1,S_2\in\mathcal{S}$ with $S_1\not=S_2$,
we cannot have $u_{S_1}=u_{S_2}$, so
Lemma~\ref{lem:counting} shows that
\[
{n\choose r}=|\mathcal{S}|=|L|\le 2^{s'}{{h-s'+r}\choose r},
\eqno{(*)}
\]
where $s'=\min\{h,s\}$.
If $s=0$, choose $r=1$, which turns the inequality $(*)$ into
$n\le t+1$ or $t\ge n-1$.
Adding $t\ge 1$, we obtain $t\ge {n/2}$, which implies the inequality
$(s\ln 2+1)t\ge {n/e}$
of the theorem.
In the following assume that $s\ge 1$.

We will make sure to choose
$r\le s$ and therefore $r\le s'$,
so that $h-s'+r\le h$.
Then, since $h=r t\le s t$, we may assume
without loss of generality that $h\le n$.
Now
\[
2^s\ge 2^{s'}\ge{{n\choose r}\Biggm/{h\choose r}}
={{n(n-1)\cdots(n-r+1)}\over
 {h(h-1)\cdots(h-r+1)}}
\ge\left({n\over h}\right)^r
=\left({n\over{r t}}\right)^r.
\eqno{(**)}
\]
If we choose $r=s$, the requirement $r\le s$
is certainly satisfied, $(**)$ becomes
\[
2^s\ge \left({n\over{s t}}\right)^s,
\]
and the inequalities $2\ge {n/{(s t)}}$
and $s t\ge{n/2}$ follow.
If $t\le{n/e}$ and we instead choose
$r=\Tfloor{{n/{(e t)}}}\ge 1$,
the inequality $s\ge{n/{(2 t)}}$ that was just
established shows that the requirement
$r\le s$ is again satisfied.
With this choice of $r$,
${n/{(r t)}}\ge e$ and therefore
$2^s\ge e^r$ and $s\ln 2+1\ge r+1\ge {n/{(e t)}}$.
Thus $(s\ln 2+1)t\ge{n/e}$,
a relation that also holds if $t>{n/e}$.
The theorem follows.
\end{proof}

\begin{corollary}
Let $n\in\TbbbN$ and $s\in\TbbbN_0$ and let $D$ be a systematic
choice dictionary with universe size $n$ that
never occupies more than $n+s$ bits.
Let $\Ttvn{delete}$ and $\Ttvn{choice}$
be upper bounds on the number of bits read from
memory during an execution of $D$'s operations
\Tvn{delete} and \Tvn{choice}, respectively
(the two quantities may depend on~$n$).
Then
$(s+{1/{\ln 2}})(\Ttvn{delete}+\Ttvn{choice})\ge {n/{(e\ln 2)}}$
and, if $s>0$,
$s(\Ttvn{delete}+\Ttvn{choice})\ge{n/2}$.
\end{corollary}

\begin{proof}
The return value of \Tvn{choice}
cannot be independent of $D$'s
client set $S$, so we must have
$\Ttvn{choice}\ge 1$.
We can therefore assume without loss
of generality that $n\ge 2$.

Given knowledge of
$r=|S|$,
we can output $S$ with $r$ iterations of
a loop in which an element of $S$ is first obtained
with a call of \Tvn{choice} and subsequently
output and removed from $S$ with a call of \Tvn{delete}.
The procedure reads at most $r(\Ttvn{delete}+\Ttvn{choice})$
bits of $D$'s representation of $S$, i.e.,
Theorem~\ref{thm:lower} can be applied with
$t=\Ttvn{delete}+\Ttvn{choice}$.
\end{proof}

With a very similar argument we can obtain a lower
bound on the amortized complexity of
\Tvn{insert}, \Tvn{delete} and \Tvn{choice}.

\begin{corollary}
Let $n\in\TbbbN$ and $s\in\TbbbN_0$ and let $D$ be a systematic
choice dictionary with universe size $n$ that
never occupies more than $n+s$ bits.
Fix an arbitrary potential function for~$D$
and assume that every representation of the
empty client set has the same potential.
Let $\Ttvn{insert}$, $\Ttvn{delete}$
and $\Ttvn{choice}$ be upper bounds
on the worst-case amortized
number of bits read from memory during $D$'s execution of
\Tvn{insert}, \Tvn{delete} and \Tvn{choice}, respectively
(the three quantities may depend on~$n$).
Then
$(s+{1/{\ln 2}})
 (\Ttvn{insert}+\Ttvn{delete}+\Ttvn{choice})\ge {n/{(e\ln 2)}}$
and, if $s>0$,
$s(\Ttvn{insert}+\Ttvn{delete}+\Ttvn{choice})\ge{n/2}$.
\end{corollary}

\begin{proof}
As above, assume without loss of generality that $n\ge 2$.
For every $r\in\{1,\ldots,n\}$ and every
$S\subseteq\{1,\ldots,n\}$ with $|S|=r$, we can take $D$
from its initial state with empty client set
via a state in which its client set is $S$
and back to a state with empty client set
using exactly $r$ calls of each of
\Tvn{insert}, \Tvn{delete} and \Tvn{choice}.
By assumption, the final potential is the
same as the initial potential, so the total
amortized number of bits read is the same as the
total actual number of bits read.
Therefore Theorem~\ref{thm:lower} is applicable with
$t=\Ttvn{insert}+\Ttvn{delete}+\Ttvn{choice}$.
\end{proof}

With $n$, $s$ and $t$ as in Theorem~\ref{thm:lower},
the theorem states that $(s+O(1))t\ge\alpha n$,
where $\alpha={1/{(e\ln 2)}}\approx0.53$.
We complement this result by showing that for every
sufficiently easily computable
function $t:\TbbbN\to\TbbbN$ with
$t(n)=\omega(\log n)$ but $t(n)=o(n)$, there is a 2-color
systematic choice dictionary $D$ that, when
initialized for universe size $n\in\TbbbN$,
reads $t(n)+o(t(n))$ bits of its internal
representation during the execution of each
operation and has a redundancy $s(n)$
for which $s(n)t(n)=n+o(n)$.
$D$ is simple.
It is a three-level trie constructed as described
in Section~\ref{sec:trie}
with the degree sequence $(p_1,p_2,p_3)$,
where $p_1=t(n)$, $p_2=\Theta(\sqrt{t(n)\log n})$
and $p_3=n$.
The two bottom levels of choice dictionaries are realized
according to Lemma~\ref{lem:atomic-c} with $c=2$,
whereas the choice dictionary of the root
is an instance of that of Theorem~\ref{thm:nlogn},
again with $c=2$.
The redundancy of the overall choice dictionary is
$\Tceil{{n/{p_1}}}+O(1+({n/{(p_1 p_2)}})\log n)=
({n/{t(n)}})(1+o(1))$,
and it is easy to see that the number of bits read
during the execution of an operation is bounded by
$p_1+p_2+O(\log n)=t(n)(1+o(1))$.
The product of the two, indeed, is
$n+o(n)$.

\section{Restricted and Extended Choice Dictionaries}
\label{sec:other}

\subsection{A Data Structure with \Tvn{insert} and
 \Tvn{extract-choice}}
\label{subsec:mlog}%

The main result of this subsection (Theorem~\ref{thm:mlog})
is a choice-like dictionary with universe size~$n$
that stores a client set $S$ of size $m$
in fewer than $n$ bits even when $m$ is not much
smaller than~$n$.
More precisely, the number of bits used is
$O(m\log(2+{n/{(m+1)}})+1)$.
Since $\log{n\choose m}\ge\log(({n/m})^m)
=m\log({n/m})$ for $m>0$, the space used by our
data structure is within a constant factor of the
information-theoretic lower bound for most
combinations of $n$ and~$m$.
What we show is that this tight space bound still
admits certain dynamic operations.
More precisely, we can support insertion and the
operation \Tvn{extract-choice} that returns and
deletes an (arbitrary) element of~$S$,
but neither unrestricted deletion nor queries
about specific elements such as \Tvn{contains}.

There is an apparent conflict between fast insertion
and a very space-efficient representation in
fewer than $n$ bits.
To represent the client set $S$ using little space, we
can store it in difference form, i.e., as the sequence of
differences between successive elements of~$S$
(in sorted order).
With this representation, however, insertion is
easily seen to be prohibitively expensive.
On the other hand, insertion is easy if we store the
elements of $S$ in no particular order, but then we
need about $\log n$ bits per element of $S$,
which is excessive.
Our solution is to store $S$ permanently in
difference form, but to insert new elements into an
unsorted buffer.
The buffer is wasteful of space, and so has to be
sorted and merged into the rest of $S$ before it
becomes too large.
Because every operation is supposed to work in constant time,
this entails a certain technical complexity.
As a warm-up before tackling this, we illustrate the
use of the difference form by developing a data
structure of possible independent interest, a
space-efficient stack that requires its elements
to occur in sorted order at all times.

\begin{definition}
A \emph{bounded-universe sorted stack} is a data structure that,
for every $n\in\TbbbN$, can be
initialized for universe size~$n$ and
subsequently maintains an initially empty
sequence $(x_1,\ldots,x_m)$ with
$1\le x_1\le\cdots\le x_m\le n$ while supporting
the following operations:

\begin{description}
\item[\normalfont$\Tvn{sorted-push}(x)$]
($x\in\{1,\ldots,n\}$):
Replaces $(x_1,\ldots,x_m)$ by $(x_1,\ldots,x_m,x)$
if $x_m\le x$ and does nothing otherwise.
\item[\normalfont$\Tvn{pop}$:]
Replaces $(x_1,\ldots,x_m)$ by $(x_1,\ldots,x_{m-1})$
and returns $x_m$ if $m>0$; returns 0 and
does nothing else if $m=0$.
\end{description}
\end{definition}

\begin{lemma}
\label{lem:stack}%
There is a bounded-universe
sorted stack that, for arbitrary $n\in\TbbbN$,
can be initialized for universe size $n$ in constant time,
subsequently
supports \Tvn{sorted-push} and \Tvn{pop} in constant
time and, when it currently holds a sequence
$(x_1,\ldots,x_m)$ of $m\in\TbbbN_0$ elements, uses at most
$m\log(q+1)+O(m\log\log(q+4)+\log n+1)$
bits, where $q={{(n-1)}/{(m+1)}}$. 
\end{lemma}

\begin{proof}
We encode integers using a scheme quite similar
to Elias' $\delta$ representation~\cite{Eli75}:
Every nonnegative integer $d$ can be represented in binary
in $\Tfloor{\lambda(d)}$ bits, where
$\lambda(d)=\log(d+1)+1$ for all $d\in\TbbbN_0$,
but this presupposes knowledge of the length of
the representation of $d$, i.e., of $\Tfloor{\lambda(d)}$.
To add this information, we append to the
representation of $d$ a sequence of
$\ell=\Tfloor{\lambda(\Tfloor{\lambda(d)})}$
bits that encode $\Tfloor{\lambda(d)}$, followed by another
$\ell+1$ bits that encode $\ell$ suitably in unary.
Altogether, this encodes $d$ in a string of
at most $\lambda(d)+3\lambda(\lambda(d))$ bits that,
read backwards, can be decoded in constant time without prior
knowledge of the length of the string.

With $x_0=1$ and $x_{m+1}=n$,
we store a bit string $B$ that encodes
$(d_1,\ldots,d_{m+1})$, where
$d_i=x_i-x_{i-1}$, for $i=1,\ldots,m+1$,
is encoded as described above
and the encodings of $d_1,\ldots,d_{m+1}$
are simply concatenated.
Before storing $B$ itself, we store the
encoding of its length
$|B|$, so that we can find the end of $B$ in
constant time.
By the implicit assumption $n,m<2^w$,
$O(w)$ bits suffice for this purpose.
Accordingly, we always store $|B|$ in a field of $\kappa w$ bits
for some suitable constant $\kappa\in\TbbbN$.
In order not to waste the last part of the field,
however, we move to there a suffix of $B$ of
the appropriate length.

It is easy to see that the operations
\Tvn{sorted-push} and \Tvn{pop} can
be supported in constant time.
$|B|$ is bounded by
$\sum_{i=1}^{m+1}(\lambda(d_i)+3\lambda(\lambda(d_i)))$.
Because $\sum_{i=1}^{m+1}d_i=n-1$ and
$\lambda$ and $\lambda\circ\lambda$ are concave
on the set of nonnegative real numbers,
$|B|\le(m+1)(\lambda(q)+3\lambda(\lambda(q)))$
by Jensen's inequality.
Since the number of bits used is
$|B|+O(\log|B|)$, the lemma follows.
\end{proof}

\begin{theorem}
\label{thm:mlog}%
There is a data structure $D$ with the following
properties:
First, for every $n\in\TbbbN$, $D$ can
be initialized for universe size~$n$
in constant time and subsequently maintains
an initially empty multiset $S$ with elements drawn from
$U=\{1,\ldots,n\}$ and executes the following
operations in constant time:

\begin{description}
\item[\normalfont$\Tvn{insert}(\ell)$]
($\ell\in U$): Inserts (another copy of) $\ell$ in $S$.
\item[\normalfont$\Tvn{extract-choice}$:]
Removes an arbitrary (copy of an) element of $S$
and returns it.
\end{description}

\noindent
Second, $D$ allows robust iteration over $S$
in constant time per element enumerated,
and when the current size of $S$
(counting each element with its multiplicity) is $m\in\TbbbN_0$,
the number of bits used by $D$ is
$O(m\log(2+{n/{(m+1)}})+1)$.
\end{theorem}

\begin{proof}
A single bit indicates whether $S=\emptyset$.
If this is not the case, $S$ is realized as the disjoint union
of four multisets,
$S_1$, $S_2$, $S^*_1$ and $S^*_2$.
$S_1$ and $S_2$ are called \emph{reservoirs}.
The elements of each reservoir are stored
in sorted order in a data structure that is
exactly as the bounded-universe sorted stack
of Lemma~\ref{lem:stack}, except that the
binary representation of each
difference between consecutive elements
is not only followed, but also
preceded by an encoding of its length.
Thus a reservoir can be read both in the
forward and in the backward direction,
and it respects the space bound even if
it contains all $m$ elements of~$S$.
$S^*_1$ and $S^*_2$ are called \emph{buffers}.
For $i=1,2$, the elements of $S^*_i$
are stored in no
particular order in an array of $|S^*_i|$ cells
of $\Tceil{\log n}$ bits each.
We also store $|S^*_1|$ and $|S^*_2|$, and the representations
of $S_1$, $S_2$, $S^*_1$, $S^*_2$,
$|S^*_1|$ and $|S^*_2|$ are interleaved in
memory so that the space occupied by them is within
a constant factor of the size of a largest among them.

At all times, exactly one buffer is called \emph{active}.
New elements are always inserted in the active buffer.
Consider a particular point in time and
let $k\in\{1,2\}$ be the index of the active
buffer at that time.
As soon as $S^*_k$ both contains at least
$2\sqrt{n}$ elements and occupies at least as many
bits as $S_k$, it stops
being active, and the other buffer, $S^*_{3-k}$,
which is empty at this time, takes over as the active buffer.
From this point on, in a background process carried out
piecemeal and interleaved with the execution of
calls of \Tvn{insert} and \Tvn{extract-choice},
$S^*_k$ is sorted in linear time with
2-phase radix sort and merged with $S_k$.
The resulting sequence is stored in $S_{3-k}$,
after which $S^*_k$ and $S_k$ are emptied.
While $S_{3-k}$ is under construction, its size
is artificially taken to be $\infty$, in the sense
that $S^*_{3-k}$ is kept active at least until
$S_{3-k}$ is finished.
The elements to be removed and returned by calls
of \Tvn{extract-choice} are taken from $S_k$
during the sorting of $S^*_k$ and from
$S_{3-k}$ during the merge of $S^*_k$ and $S_k$.
The background process should be fast enough
to keep $S_k$ nonempty during the sorting of $S^*_k$,
to keep $S_{3-k}$ nonempty once the
merge has started, and to complete before
$S^*_{3-k}$ occupies as many bits as $S_{3-k}$
will when it is finished.

In order for the background process to require only
constant time per operation executed by~$D$,
the merge of $S_k^*$ with $S_k$ to obtain
$S_{3-k}$ must happen in $O(|S_k^*|)$ time, which
generally means sublinear time with
respect to the number of elements in the reservoir $S_k$.
Using the compactness of the representation
of $S_k$, we achieve this by carrying out the merge using table lookup.
Recall that $S_k$ and $S_{3-k}$ are stored in
difference form.
During the merge, we therefore
keep track of
the element most recently removed from $S_k$ and the 
element most recently written to $S_{3-k}$,
so that future elements can be decoded or
encoded correctly.
As detailed above,
the elements of $S_k$ are represented through
variable-length bit \emph{segments}, each
complete with its length encoding.
At any given time during the merge, initial parts of
$S_k^*$ and $S_k$ have been merged, while
the remaining elements are still to be processed.
To continue the merge, we use table lookup to
find the last element, $x$, of $S_k$, if any, whose segment
is fully contained in the next $\Tfloor{({1/5)}\log n}$
bits of the representation of $S_k$.
If $x$ exists and is no larger than the next
element $x^*$ of $S_k^*$
($\infty$ if $S_k^*$ is exhausted), the elements of $S_k$
up to and including $x$ can be moved from $S_k$
to the output sequence $S_{3-k}$.
If $x$ exists but $x^*<x$, all remaining elements of $S_k$
no larger than $x^*$ can be identified
with a second table lookup
(applied to the next $\Tfloor{({1/5})\log n}$ bits of
the representation of $S_k$
and the difference between the first elements
of $S_k$ and $S_k^*$, a total of at most
$({2/5})\log n$ bits)
and moved to $S_{3-k}$.
In both cases, subsequently
the first remaining element of $S_k$
($\infty$ is $S_k$ is exhausted) can be
compared with $x^*$, and the smaller of the two
can be moved to $S_{3-k}$.
All of this takes constant time, and it consumes
$\Omega(\log n)$ bits of the
representation of either $S_k$ or $S_k^*$.
Therefore the total time needed to sort
$S^*_k$ and to merge it with $S_k$ is $O(|S^*_k|)$.
As anticipated above, this shows that executing
a constant number of steps of the
background process for each call of
\Tvn{insert} and \Tvn{extract-choice}
suffices to let the process finish in time.
Thus $D$ executes every operation in
constant time.

The tables needed for the merge occupy
$O(\sqrt{n})$ bits and can be constructed in
$O(\sqrt{n})$ time.
This happens
in another background process that is
advanced by a constant number of steps
(if it has not completed) whenever
the active buffer contains at least
$\sqrt{n}$ elements and that finishes before
the active buffer reaches $2\sqrt{n}$ elements.
Whenever the number of elements in the
active buffer drops below $\sqrt{n}$---in particular,
when an empty buffer becomes active---the tables
are discarded, so that the
space that they occupy is always within a
constant factor of that taken up by the active buffer.

In order to support robust iteration over~$S$,
we add more structural information to the data structure.
Call an element of $S$ \emph{live} if it was
present in $S$ at the time of the most recent preceding
call of $\Tvn{iterate}.\Tvn{init}$, if any,
and has not since been enumerated.
If an element of $S$ is not live, it is \emph{dead}.
An initial part of each reservoir or buffer
that was not emptied since the most recent preceding call
of $\Tvn{iterate}.\Tvn{init}$
(informally, that existed at that time)
is marked off as its \emph{live area}.
If one or more live areas are nonempty just before
a call of $\Tvn{iterate}.\Tvn{next}$, the next
element enumerated is the last element of some
live area.
If the last element of a live area is enumerated
or deleted in a call of \Tvn{extract-choice},
the live area shrinks by one element.
A call of $\Tvn{iterate}.\Tvn{init}$ sets the live
area of each reservoir and of each buffer to be the
entire reservoir or buffer.
This is the only occasion on which a live area expands.

When a buffer is sorted and merged with a reservoir
to create a new reservoir, the live elements can
no longer be kept in a contiguous area.
On the other hand, constant-time iteration
is possible only if the next live element
to be enumerated can be found in constant time.
For this reason reservoirs and buffers
created since the most recent preceding call of
$\Tvn{iterate}.\Tvn{init}$ have empty live areas,
and the live elements in each such
reservoir are kept in a linked list in the
opposite of the order in which they occur in the reservoir.
Each live element in the list is labeled by its distance,
measured in bits, to the next live element in the list.
It is easy to see that the labels do not add to
the total space requirements by more than a
constant factor (small labels could even
be represented in unary).
Any list labels present in the live area of
a reservoir are ignored---all elements in
the live area are alive.

Before a buffer is sorted in preparation for being
merged with a reservoir, each buffer element is
given a bit of satellite data that indicates
whether the element is alive or dead.
The bit is carried along with the element during
the sorting and used in the merge to equip the
resulting reservoir with the linked list of its
live elements described above.
The table that drives the merge,
instead of merely indicating a number of bits
that can be copied from the old to the new
reservoir, must be changed to supply a
``piece of reservoir'' complete with list labels.
The at most one label that points outside of
the current piece, namely to an element in
the part of the new reservoir that has already
been constructed
(the \emph{nonlocal pointer}),
must be filled in separately
and outside of the table lookup.
Correspondingly, it is convenient to split
the table lookup into two:
The first table lookup yields the piece of
reservoir that precedes the nonlocal pointer
(thus the piece specifies
only dead elements), and the second table lookup
yields the piece of reservoir that follows it
(if there is no nonlocal pointer, let the
first table lookup
provide the entire piece of reservoir).
Even with these additional computational
steps, the time and
space bounds established above remain valid.

Since an iteration enumerates only elements that
were present when the iteration started and were
not deleted before their enumeration, iterations
are
robust.
\end{proof}

As is easy to see, the data structure of
Theorem~\ref{thm:mlog} also supports constant-time \Tvn{choice}.
This operation,
however, is not likely to
be very useful except in the combination \Tvn{extract-choice}.

\subsection{A Data Structure with \Tvn{color} and \Tvn{setcolor}}

In this subsection we reconsider a data structure
of Dodis et al.~\cite{DodPT10} and extend it to support
constant-time initialization.
The data structure basically emulates
$c$-ary memory, for arbitrary $c\ge 2$,
on standard binary memory almost without
losing space.
This is essential
to the choice dictionaries
of Subsection~\ref{subsec:random}
and Section~\ref{sec:nonsystematic} in the case
where the number of colors is not a power of~2.

For $b=\Tfloor{\log n}$ the following lemma
essentially coincides with the result
of Dodis et al.
Only in extreme cases can it be useful
to choose smaller values for~$b$, but we need the
present form of Lemma~\ref{lem:succincter-t}
to prove Lemma~\ref{lem:unsystematic-tc}
in full generality.
In Lemma~\ref{lem:succincter-t} and Theorem~\ref{thm:succincter}
the sequence $(c_1,\ldots,c_n)$ is assumed to be represented
in a way that allows it to be communicated to the
initialization routine in constant time,
most naturally as a list of constant length
of (value, multiplicity) pairs.

\begin{lemma}
\label{lem:succincter-t}%
There is a data structure
that, for all given $n,b\in\TbbbN$
and for every given sequence
$(c_1,\ldots,c_n)$ of $n$ positive integers
with $|\{\ell\in\TbbbN:1\le\ell<n$ and
$c_\ell\not= c_{\ell+1}\}|=O(1)$,
can be initialized in constant time and subsequently
occupies $\sum_{\ell=1}^n\log c_\ell+O({n/{2^b}}+\log n+1)$ bits,
needs access to a table of $O(b^2)$ bits
that can be computed in $O(b)$ time and
depends only on $n$, $b$
and $(c_1,\ldots,c_n)$,
and maintains a sequence
$(a_1,\ldots,a_n)$ 
drawn from
$\{0,\ldots,c_1-1\}\times\cdots\times\{0,\ldots,c_n-1\}$
under constant-time reading and writing
of individual elements in the sequence.
The data structure does not initialize the sequence.
\end{lemma}

\begin{proof}
Dodis et al.~\cite{DodPT10}
considered the fundamental problem of
realizing an array of $n$ entries,
each drawn from a set of the
form $\{0,\ldots,c-1\}$
for integer $c\ge 2$.
They
described a beautifully simple
construction that allows individual array entries
(called \emph{small digits} in what follows) to
be read and written in constant time, yet needs
just $n\log c+O(1)$ bits
(the authors even argue that
$\Tceil{n\log c}$ bits suffice).
A few details not considered by
Dodis et al.\ can easily be handled:
(1) The authors actually group $r$ small
digits into one
\emph{big digit} drawn from
$\{0,\ldots,c^r-1\}$, where
$r\in\TbbbN$ is chosen to make
$c^r=\Theta(n^2)$, and accordingly store
approximately $n/r$ big digits
instead of $n$ small digits.
The proof in fact is correct for arbitrary $r$
as long as $c^r$
is $\Omega(n^2)$
(but still $2^{O(w)}$, so that big
digits can be manipulated in constant time);
this allows us to handle
values of $c$ larger than $\Theta(n^2)$.
(2) The big digits are associated with the nodes
of a binary tree~$T$.
The proof remains correct if the big digit associated
with the root of $T$ in fact is drawn from a smaller domain
than the other big digits;
this allows us to deal with the fact that $r$
may not divide~$n$.
(3) The construction needs a
table $Y$ of $O((\log n)^2)$ bits
of precomputed numbers that depend on $n$ and~$c$.
The authors do not mention the time needed
to obtain
$Y$ from $n$ and $c$,
but it is easy to see that it can be done in
$O(\log n)$ time, namely constant time per
level in~$T$.
In particular,
a part of~$Y$ indicates
the size and memory layout of the data structure.
The same time allows us to compute the
powers $c^i$, for $i=2,\ldots,r$, which are needed
to extract small digits from big digits and
update small digits within big digits in
constant time.

Lemma~\ref{lem:succincter-t} follows by
splitting
the sequence
$(a_1,\ldots,a_n)$ 
into
$O(1+{n/{2^b}})$ subsequences
$(a_i,\ldots,a_j)$ with $c_i=\cdots=c_j$
and $j-i=O(2^b)$
and applying the
construction of Dodis et al.\ independently
to each such subsequence.
\end{proof}

For $n\in\TbbbN$ and $\epsilon>0$, let us call a sequence
$(c_1,\ldots,c_n)$ of $n$ positive integers 
\emph{$\epsilon$-balanced} if
$\epsilon (\log n)^2\sum_{\ell\in S}\log c_\ell
\le\sum_{\ell=1}^n\log c_\ell$
for all $S\subseteq\{1,\ldots,n\}$ with $|S|\le(\log n)^3$.
The theorem below requires $(c_1,\ldots,c_n)$ to
be $\epsilon$-balanced for some fixed~$\epsilon>0$.
While this requirement is necessary for our proof,
it is a technicality of scant interest.
In order for a sequence $(c_1,\ldots,c_n)$ not to be
$\epsilon$-balanced for any fixed $\epsilon>0$, at least some of
its elements must be very large relative to~$n$.
Indeed, provided that $c_i\ge 2$ for $i=1,\ldots,n$,
our implicit convention
$\log c_i=O(w)$ for $i=1,\ldots,n$ implies
that the condition of $\epsilon$-balance is
automatically satisfied for some fixed $\epsilon>0$ if
$w=O({n/{(\log n)^5}})$.

\begin{theorem}
\label{thm:succincter}%
For all fixed $\epsilon>0$,
there is a data structure
that, for all given $n\in\TbbbN$
and for every given
$\epsilon$-balanced sequence
$(c_1,\ldots,c_n)$ of $n$ integers
with $c_\ell\ge 2$ for $\ell=1,\ldots,n$
and $|\{\ell\in\TbbbN:1\le\ell<n$ and
$c_\ell\not= c_{\ell+1}\}|=O(1)$,
can be initialized in constant time and subsequently
occupies $\sum_{\ell=1}^n\log_2 c_\ell+O((\log n)^2+1)$ bits
and maintains a sequence
$(a_1,\ldots,a_n)$ of $n$ integers,
drawn from
$\{0,\ldots,c_1-1\}\times\cdots\times\{0,\ldots,c_n-1\}$,
under constant-time reading and writing
of individual elements in the sequence.
The data structure does not initialize the sequence.
For $\ell=1,\ldots,n$, the parameter $c_\ell$ may be presented
to the data structure in the form of a pair $(x_\ell,y_\ell)$
of positive integers with $c_\ell=x_\ell^{y_\ell}$ and
$y_\ell=n^{O(1)}$.
\end{theorem}

\begin{proof}
Let $Y$ be the table
used by the data structure
of Lemma~\ref{lem:succincter-t}.
Aside from the question of $(c_1,\ldots,c_n)$
being $\epsilon$-balanced,
the only essential difference between 
Theorem~\ref{thm:succincter} and
Lemma~\ref{lem:succincter-t} is that the
theorem does not assume $Y$ to be externally
available.
The theorem provides space for storing $Y$,
but no time for computing it before the first
operation must be served.
Assume without loss of generality that $n\ge 2$
and let $U=\{1,\ldots,n\}$ and $b=\Tfloor{\log n}$.
Recall that $Y$
can be computed from $(c_1,\ldots,c_n)$ in $O(b)$ time---if
necessary, this time bound also allows for the calculation
of $c_\ell$ from $x_\ell$ and $y_\ell$
for $\ell=1,\ldots,n$ via repeated squaring.

Consider first the special case $c_1=\cdots=c_n=c$.
We allocate first $\Theta(b^2)$ bits for $Y$,
then a block of $\Theta(b\log b)$ additional bits whose use
will be explained later, and finally
space for a data structure~$D\Tsup T$.
$D\Tsup T$ is a trie of constant depth with $n$ potential leaves,
the $\ell$th of which, for $\ell=1,\ldots,n$,
holds $a_\ell$ in $\Tceil{\log c}$ bits
if $a_\ell$ has changed from its
initial value of~0.
Disregarding the question of space,
$D\Tsup T$ can provide the functionality promised in
the theorem.
Choose the (constant) depth of $D\Tsup T$ to
be at least $4$.
We will use $D\Tsup T$ for at most $2 b$ operations.
Thus, even without paying particular attention
to economy of space, we can easily ensure
that the number of bits needed for $D\Tsup T$ is
$O(2 b (n^{{1/4}}\log n+\log c))$ and
therefore, for $n$ larger than a constant, at most
$({1/{(2 b+2)}})n\log c$.

We construct $Y$ in a background process
interleaved with the execution of the first
$b$ operations (the \emph{first phase}),
using $D\Tsup T$ to serve these $b$ operations.
At the end of the first phase,
when $Y$ is ready, we start running an
instance $D$ of the data structure of
Lemma~\ref{lem:succincter-t} in parallel
with $D\Tsup T$, making sure that every
$\ell\in U$ that
is written to
after the first phase
has the correct associated value
in $D$ and, if $\ell$ is still present in $D\Tsup T$,
in $D\Tsup T$.
Interleaved with the execution of the second
group of $b$ operations (the \emph{second phase}), we empty $D\Tsup T$
element by element, making
sure that every element deleted from $D\Tsup T$
has the correct value in $D$---informally, we transfer the element
from $D\Tsup T$ to $D$.
To determine $a_\ell$
for some $\ell\in\{1,\ldots,n\}$ in the second phase,
we first query $D\Tsup T$.
If $\ell$ is present in~$D\Tsup T$,
the associated value is returned.
Otherwise we return the value
associated with $\ell$ in~$D$.
All operations can
still be executed
in constant time
during the parallel operation of
$D$ and $D\Tsup T$.
After the second phase, $D\Tsup T$
is empty and can be considered to have disappeared.

A problem that was ignored until now is
that $D\Tsup T$ and $D$ must use the same memory area
during the second phase.
Partition $U$ into $2 b+2$ \emph{sectors}
$U_1,\ldots,U_{2 b+2}$ of consecutive elements,
all of the same size $n'$,
except that $U_{2 b+2}$ may be smaller.
We call $U_1$ the \emph{forbidden} sector
and $U_{2 b+2}$ the \emph{incomplete} sector.
$D$ is in fact split into three instances of
the data structure of Lemma~\ref{lem:succincter-t}:
$D_1$, whose universe is the forbidden sector $U_1$,
and two instances $D_*$ and $D_{2 b+2}$
whose universes are
$\bigcup_{i=2}^{2 b+1} U_i$ and $U_{2 b+2}$
(both translated suitably to begin at~1),
respectively.
At the end of the first phase, aided by $Y$,
we can compute the exact number of bits required
for each of $D_1$, $D_*$ and $D_{2 b+2}$ in constant time.
Overlapping the memory space used for $D\Tsup T$,
we allocate space first for $D_1$ and then for
$D_*$ and $D_{2 b+2}$.
Since $D_1$ occupies at least $n'\log c$ bits
and $D\Tsup T$ occupies at most $n'\log c$ bits,
$D_1$ is the only component of $D$ whose memory
area overlaps that of $D\Tsup T$.

During the transfer of elements from $D\Tsup T$ to $D$
in the second phase,
it is easy to handle
elements outside of the forbidden sector.
As concerns elements destined for $D_1$, however,
a problem arises because
$D_1$ overlaps $D\Tsup T$, which is still in use.
In order to circumvent this problem, we will
ensure that during the execution of the
first $2 b$ operations, no element
located in the forbidden sector
acquires a nonzero value.
We achieve this by storing the elements not directly, but
according to a rudimentary hash function that
is data-dependent and defined in a lazy manner.
The hash function takes the form of a bijection
$g$ from $U$ to itself.
Every element in the incomplete sector is mapped to
itself by $g$, and the rest of $g$ is induced
by a permutation $\pi$ of $I=\{1,\ldots,2 b+1\}$
in the following way:
For each $i\in I$, the $\ell$th element
of $U_i$ is mapped by $g$ to the $\ell$th element
of $U_{\pi(i)}$, for $\ell=1,\ldots,n'$.
Arguments in $U$ of operations to be executed
are mapped under $g$ before being passed on
to one of the three components of $D$, and
answers obtained from the components of $D$ are
mapped under $g^{-1}$ to obtain the
appropriate return values.
The permutation $\pi$ and its inverse must
be stored in the data structure, the necessary
space being furnished by the block of
$\Theta(b\log b)$ bits allocated but not used so far.

It remains to describe the permutation $\pi$ of~$I$.
When an element outside of the incomplete sector is
first given a nonzero value,
suppose that it belongs to $U_{i_1}$.
Then $\pi(i_1)$ is defined to be an arbitrary
element of $I\setminus\{1\}$, such as its minimum.
Similarly, if $U_{i_2}$ is the second sector
other than $U_{2 b+2}$ to receive
a nonzero value, $\pi(i_2)$ is defined to
be an arbitrary element of
$I\setminus\{1,\pi(i_1)\}$.
Continuing in the same manner, we can avoid
$1$ as a value of $\pi$ for the duration
of at least $|I|-1=2 b$ operations, which
was our goal.

Let us now turn to the general case.
While $(c_1,\ldots,c_n)$ may contain 
distinct values, by assumption there is a constant
$q\in\TbbbN$ for which $U$ can be partitioned
into $q$ \emph{segments} $V_1,\ldots,V_q$
of consecutive elements such that for
$i=1,\ldots,q$, $c_{\ell}=c_{\ell'}$
for all $\ell,\ell'\in V_i$.
Informally, our plan is to run the procedure described for
a single segment in parallel for all segments,
with a shared temporary data structure $D\Tsup T$
``hiding'' in the memory area that holds
the union of the forbidden sectors of all segments.

For $\ell\in U$,
if and when $a_\ell$ becomes nonzero
during the execution of the first $2 b$ operations,
$D\Tsup T$ now stores $a_\ell$
in $\Tceil{\log c_\ell}$ bits.
The total number of bits, $s\Tsup T$, required for $D\Tsup T$
is therefore $O(2 b n^{1/4}\log n)$ plus
$\sum_{\ell\in S}\Tceil{\log c_\ell}$ for a
set $S\subseteq U$ with $|S|\le 2 b$.
Let $Z=\sum_{\ell\in U}\log c_\ell$.
Since $(c_1,\ldots,c_n)$ is $\epsilon$-balanced
and $Z\ge n$,
$s\Tsup T\le({1/{(\epsilon (\log n)^2)}}+2 {b/n}
 +O(2 b n^{{-3}/4}\log n))Z$,
which, for $n$ larger than a constant, is at most
${Z/{(4 b+4)}}$.

Call a segment \emph{small} if it contains at
most $(2 b+2)^2$ elements, and
\emph{large} otherwise.
The reason for distinguishing between small and
large segments is that a
large segment can always be divided into $2 b+2$
sectors, all of the same size, except that
one segment may be smaller.
Because this is
not necessarily the case for a small segment,
the small segments do not contribute
forbidden sectors.
We must show that, even so, the forbidden sectors
together require enough space for the values
associated with their elements to ``cover''
the temporary data structure $D\Tsup T$.
If this is so, the parallel procedure will
work as intended:
During the first $2 b$ operations, a shared
permutation $\pi$ is defined so that the at
most $2 b$ elements that receive nonzero values
in the first and second phases avoid all
forbidden sectors, and in the second phase
the elements stored in $D\Tsup T$ are
transferred to at most $3 q$ instances
of the data structure of Lemma~\ref{lem:succincter-t}.

If $S$ is the set of elements in small segments,
then $|S|\le q (2 b+2)^2$, which, for $n$ larger
than a constant, is at most $(\log n)^3$.
Since $(c_1,\ldots,c_n)$ is $\epsilon$-balanced,
we may conclude, still for $n$ larger than a
constant, that
$\sum_{\ell\in S}\log c_\ell
\le{Z/{(\epsilon (\log n)^2)}}$.
If this relation holds and $S\Tsub F$ is the set of
elements in forbidden sectors, it is easy to see that
$\sum_{\ell\in S_{\mathrm{F}}}\log c_\ell
\ge({{{1-1/{(\epsilon (\log n)^2)}})Z}/{(2 b+2)}}.$
Thus, for $n$ larger than a constant,
$\sum_{\ell\in S_{\mathrm{F}}}\log c_\ell
\ge{Z/{(4 b+4)}}\ge s\Tsup T$,
from which the desired conclusion follows.
\end{proof}

\subsection{Choice Dictionaries with \Tvn{p-rank} and
 \Tvn{p-select}}
\label{subsec:random}

In this subsection we describe an extended choice
dictionary that supports the
additional operations
\Tvn{p-rank} and \Tvn{p-select}.
Before delving into the technical details, we
provide a brief overview of the main ideas
involved in the special case of constant
operation times and for the special application
of uniform random generation.

First, using methods that are fairly standard,
at least if suitable tables are assumed to be available,
the operations \Tvn{p-rank} and \Tvn{p-select} and
even \Tvn{rank} and \Tvn{select} can be supported
within \emph{segments} of
polylogarithmic size (Lemma~\ref{lem:p-t}):
One maintains a summation tree $T$ of constant depth
and almost-logarithmic degree, with each node storing the
number of elements of the client set $S$ below each
of its children, and the procedures of accumulating
prefix sums along a root-to-leaf path in $T$
(for \Tvn{rank}) and of searching within the
prefix sums along such a path (for \Tvn{select})
are carried out with table lookup.

At this point the uniform generation boils down
to choosing a random segment.
The choice should not be uniform, however.
Instead a segment should have a probability
or \emph{weight}
proportional to the number of elements of $S$
in the segment.
This partitions the segments
dynamically into a polylogarithmic
number of \emph{weight classes}, and the uniform
generation can proceed by first picking a 
random weight class, according to a suitable probability
distribution, and subsequently picking a segment
uniformly at random within the chosen weight class.
The latter task can be solved with a data structure
that we already have, namely that of Theorem~\ref{thm:nlogn}.
It is wasteful in terms of space, but
because it is applied to segments of
polylogarithmic size and not directly
to elements of the original universe, the space
requirements can be made sufficiently small.

As for picking a random weight class, the relevant
probability distribution changes dynamically and
can be almost arbitrary,
which renders the problem difficult.
What makes it manageable nonetheless is the fact
that the number of weight classes is only
polylogarithmic.
We solve the problem using a data structure of
P\v atra\c scu and Demaine~\cite{PatD06},
slightly modified to suit our needs
(Lemma~\ref{lem:prefSum}).
Conveniently, although it is more powerful
than what is required, the same data structure can also
be used within segments.
This ends the overview of this subsection.

If the operations \Tvn{rank} and \Tvn{select}
or \Tvn{p-rank} and \Tvn{p-select} are to be realized
with the trie-combination method of
Section~\ref{sec:trie}, the inner nodes
in the trie must be generalized.
As usual, identify the leaves of the trie, in
the order from left to right, with the integers
$1,\ldots,n$ and consider the colorless case
and \Tvn{rank} and \Tvn{select}.
An inner node $u$ with $d$ children, rather than maintaining
a subset of $\{1,\ldots,d\}$ or, what amounts to
the same, a bit vector of length~$d$, must maintain
a sequence $A$ of $d$ nonnegative integers, the $i$th
of which, for $i=1,\ldots,d$, is the number of
leaf descendants of the $i$th child of~$u$
that belong to the client set, and \Tvn{rank}
and \Tvn{select} must be generalized to the
functions \Tvn{sum} and \Tvn{search}
defined below.

Given a sequence $A$ of $m$ integers, where
$m\in\TbbbN$, let us denote its $j$th entry,
for $j=1,\ldots,m$, by $A[j]$.
Moreover, let $\sigma(A)$ be the
sequence of prefix sums of~$A$, i.e., the
sequence of length $m$ with
$\sigma(A)[j]=\sum_{i=1}^j A[i]$ for $j=1,\ldots,m$.
Note that $\sigma$ is a linear operator.
Say that an \emph{atomic ranking structure}
for $A$ is a data structure that can return
$\Tvn{rank}(x,A)$ in constant time for
arbitrary integer~$x$.
Let a \emph{searchable prefix-sums structure}
be a data structure that, for arbitrary
$n,b,\delta\in\TbbbN$ with $\delta\le b\le w$,
can be initialized for parameters
$(n,b,\delta)$ and subsequently
maintains a sequence $A$
of $n$ integers, initially $(0,\ldots,0)$,
under the following operations:

\begin{description}
\item[\normalfont$\Tvn{sum}(j)$]
($j\in\{0,\ldots,n\}$):
Returns $\sigma(A)[j]$ if $j>0$ and 0 if $j=0$.
\item[\normalfont$\Tvn{search}(x)$]
($x\in\{1,\ldots,2^b-1\}$):
Returns $\Tvn{rank}(x-1,\sigma(A))+1$.
\item[\normalfont$\Tvn{update}(j,\Delta)$]
($j\in\{1,\ldots,n\}$ and
$\Delta\in\{-(2^\delta-1),\ldots,2^\delta-1\}\cap
\{-A[j],\ldots,2^b-1-\sigma(A)[n]\}$):
Replaces $A[j]$ by $A[j]+\Delta$.
\end{description}

\noindent
The complicated precondition of \Tvn{update}
simply stipulates that $\Delta$ be a signed $\delta$-bit
quantity whose addition to $A[j]$ neither causes
$A[j]$ to become negative nor causes
some entry in $\sigma(A)$
to exceed $2^b-1$.
Negative values for $\Delta$ are assumed
to be represented suitably.
Following the initialization for parameters $(n,b,\delta)$,
we call $n$ the \emph{universe size},
$b$ the \emph{sum bit length} and
$\delta$ the \emph{update bit length} of a searchable
prefix-sums structure.
Although the definition does not list an operation
\Tvn{value} such that $\Tvn{value}(j)$ returns $A[j]$,
for $j=1,\ldots,n$, it can easily be derived as
$\Tvn{value}(j)=\Tvn{sum}(j)-\Tvn{sum}(j-1)$.

Assume that each inner node $u$ of a trie $T$ constructed
as described in Section~\ref{sec:trie} is equipped
with a searchable prefix-sums structure $D_u$ with
universe size $d$, where $d$ is the degree of~$u$,
sum bit length at least
$\Tceil{\log(n+1)}$ and update bit length~1.
Informally, the $i$th integer maintained by~$D_u$
will be the sum of the values ``below'' $u$'s
$i$th child.
With notation as in
Section~\ref{sec:trie}, \Tvn{rank} and \Tvn{select}
for the uncolored case
can be realized as follows:

\begin{description}
\item[\normalfont$\Tvn{rank}(\ell)$:]
Initialize a variable $s$ to 0.
Then, starting at the root $r$ of $T$ and as long as the
height of the current node $u$ is at least~2,
let $i=\Tvn{viachild}(u,\ell)$,
add $D_u.\Tvn{sum}(i-1)$ to $s$ and step
to the $i$th child of~$u$.
When a node $u$ of height~1 is reached,
return $s+D_u.\Tvn{sum}(\Tvn{viachild}(u,\ell))$.
\item[\normalfont$\Tvn{select}(k)$:]
Starting at $r$ and as long as the current node
$u$ is not a leaf,
let $i=D_u.\Tvn{search}(k)$,
subtract $D_u.\Tvn{sum}(i-1)$ from $k$
and step to the $i$th child of~$u$.
When a leaf $v$ is reached, return $\Tvn{leftindex}(v)$.
\end{description}

The operations \Tvn{insert} and \Tvn{delete} given
in Section~\ref{sec:trie} have to
be modified in minor ways.
E.g., the test $D_u.\Tvn{contains}(i)=1$
should be replaced by $D_u.\Tvn{value}(i)\ge 1$.
The details are left to the reader.
The remaining operations will not be needed.

A suitable searchable prefix-sums structure for our
purposes is the slight generalization of a data
structure due to
P\v atra\c scu and Demaine~\cite[Section~8]{PatD06}
expressed in the following lemma.
Our result differs from that of
\cite{PatD06} in that we allow
entries in the array $A$ to be zero and
distinguish between the sum bit length $b$
(which bounds the values in~$A$)
and the word size $w$
(which determines the computational power
of the RAM).
We provide a proof that is somewhat
simpler and more explicit than
that of~\cite{PatD06}.

\begin{lemma}
\label{lem:prefSum}%
There is a searchable prefix-sums structure that,
for arbitrary $n,b,\delta\in\TbbbN$ with
$\delta\le b\le w$, can be initialized
for parameters $(n,b,\delta)$
in constant time and subsequently
occupies $O(n b)$ bits and supports
\Tvn{sum}, \Tvn{search} and \Tvn{update}
in $O(1+{{\log n}/{\log(1+{w/\delta})}})$ time.
\end{lemma}

\begin{proof}
For the time being ignore the claim about constant-time initialization.
Choose $m$ as an integer with $m\ge 2$ and
$m=\Theta(\min\{\sqrt{{w/\delta}},n\})$.
As noted by P\v atra\c scu and Demaine,
it suffices to prove the lemma for
universe size at most~$m$.
This is because,
similarly as in the trie-combination method,
the overall data structure can
be organized as a tree $T$ of height
$O(1+{{\log n}/{\log m}})=
O(1+{{\log n}/{\log(1+{w/\delta})}})$, each of whose
$O({n/m})$ nodes contains a data structure for the
same problem, but for a sequence of
length at most~$m$.
Assume therefore that
the universe size is bounded by $m$
and, in fact, that it is exactly~$m$.

Let $h=2^{\delta+1}m$
and choose $f\in\TbbbN$ with
$f=O(\delta+\log m)=O(m\delta)=O({w/m})$
such that for each $a\in\TbbbN$ with $a\le 6 m h$, an integer
$x\in\{-a,\ldots,a\}$
can be encoded through the binary
representation of (the nonnegative
integer) $x+a$ in a field
of $f$ bits.
A sequence of $m$ integers encoded in this way,
called a \emph{small vector} with \emph{offset} $a$, fits in
$O(w)$ bits and can be manipulated in constant time.
In particular, provided that no overflow occurs,
$\sigma$ can be applied to a small vector
in constant time
through multiplication by $1_{m,f}$ and use of
the relation $\sigma(x)=\sigma(x+a)-\sigma(a)$.
Let a \emph{big vector} be a sequence of $m$
integers drawn from $\{0,\ldots,2^b-1\}$.
In the following, when an integer $a$ is used
in a context that requires a vector,
$a$ is shorthand for the vector $(a,\ldots,a)$.
For $a\in\TbbbN$ and
integer $x$, let
$\Trho_a(x)$ be the number in $\{-a,\ldots,a\}$
closest to $x$, i.e., $\Trho_a(x)=\min\{\max\{x,-a\},a\}$.

For simplicity, assume that the number of
calls of \Tvn{update} is infinite.
For $t=1,2,\ldots,$ if the $t$th update
is $\Tvn{update}(j,\Delta)$, briefly define $E_t$
as the sequence of length $m$ with
$E_t[j]=\Delta$ and $E_t[i]=0$ for
$i\in\{1,\ldots,m\}\setminus\{j\}$.
For arbitrary integers $s$ and $t$,
let $A_t=\sum_{i=1}^t E_i$ if $t>0$,
$A_t=0$ ($=(0,\ldots,0)$) if $t\le 0$,
$A_{s\Ttwodots t}=\sum_{i=s+1}^t E_i$ if
$0\le s<t$, and $A_{s\Ttwodots t}=0$ otherwise.
Let \emph{phase}~0 be the period of time from the
initialization until and including the
execution of the $m$th update and,
for $k=1,2,\ldots,$ let phase~$k$
be the time from the end of phase~$k-1$
until and including the execution of
the $((k+1)m)$th update.

We pretend to keep track of the number $t$
of updates executed so far;
it will be easy to see that it suffices
to know $t\bmod(2 m)$.
During phase $k$, for $k=0,1,\ldots,$
we store $A_t$ and $\sigma(A_{(k-1)m})$
as big vectors and
$A_{(k-1)m\Ttwodots k m}$,
$A_{k m\Ttwodots t}$ and
$\Trho_{2 h}(A_t)$ as small vectors with offset $2 h$.
Moreover, we have an atomic ranking dictionary
$D_{(k-1)m}$ for $\sigma(A_{(k-1)m})$.
Throughout the phase and piecemeal, interleaved with
the execution of updates, we add
$\sigma(A_{(k-1)m})$ and
$\sigma(A_{(k-1)m\Ttwodots k m})$
componentwise to obtain $\sigma(A_{k m})$
and compute an atomic ranking dictionary $D_{k m}$ for
$\sigma(A_{k m})$.
Since the latter can also be done in $O(m)$
time \cite[Corollary~8]{Hag98}, it suffices
to spend constant time per update on
this background process in order for
$\sigma(A_{k m})$ and $D_{k m}$
to be ready at the beginning of the next phase,
which is when they are needed.

To execute $\Tvn{update}(j,\Delta)$
in phase~$k$, for some $k\ge 0$,
add $\Delta$ to the $j$th components of
$A_t$ and $A_{k m\Ttwodots t}$ to obtain $A_{t+1}$
and $A_{k m\Ttwodots t+1}$, respectively,
replace the $j$th component of
$\Trho_{2 h}(A_t)$ by $\Trho_{2 h}(A_{t+1}[j])$
to obtain $\Trho_{2 h}(A_{t+1})$, and increment~$t$.
If subsequently $t\bmod m=0$ and hence $t=(k+1)m$,
prepare for phase~$k+1$ by initializing to zero
an integer variable that held $A_{(k-1)m\Ttwodots k m}$
in the phase that ends and will hold
$A_{(k+1)m\Ttwodots t}$ in the phase
that begins.

To execute $\Tvn{sum}(j)$ in phase~$k$, for some $k\ge 0$,
return the sum of the $j$th components
of $\sigma(A_{(k-1)m})$,
$\sigma(A_{(k-1)m\Ttwodots k m})$ and
$\sigma(A_{k m\Ttwodots t})$.

To support $\Tvn{search}$, it suffices to be able
to compute $\Tvn{rank}(x,\sigma(A_t))$ for arbitrary given
$x\in\{0,\ldots,2^b-1\}$.
To solve this problem
in phase~$k$, for some $k\ge 0$, let $s=(k-1)m$ and
use $D_s$ to identify a $j^*\in\{1,\ldots,m\}$
with $|x-\sigma(A_s)[j^*]|=\min\{|x-\sigma(A_s)[j]|:1\le j\le m\}$.
Then compute $x_0=\sigma(A_s)[j^*]$ and
$\widetilde{B}=
 \sigma(\Trho_{2 h}(A_t))
 -(\sigma(\Trho_{2 h}(A_t))-\sigma(A_{s\Ttwodots t}))[j^*]$
(of course, $A_{s\Ttwodots t}$ can be obtained as
$A_{s\Ttwodots k m}+A_{k m\Ttwodots t}$)
and return $\Tvn{rank}(\Trho_h(x-x_0),\widetilde{B})$.
By Lemma~\ref{lem:word}(d),
this can be done in constant time.

To see that the computation of
$\Tvn{rank}(x,\sigma(A_t))$ is correct, let
$B=\sigma(A_t)-(\sigma(A_t)-\sigma(A_{s\Ttwodots t}))[j^*]$
$=\sigma(A_t)-x_0$.
Informally, $B$ is the current state
(i.e., after $t$ updates)
of the sequence $\sigma(A)$ of prefix sums, but
``normalized'' through the subtraction of $x_0$
to have the value 0 at time $s$ and at the index~$j^*$.
Of course, instead of computing
$\Tvn{rank}(x,\sigma(A_t))$, we can just as well determine
$\Tvn{rank}(x-x_0,\sigma(A_t)-x_0)=\Tvn{rank}(x-x_0,B)$.
We cannot compute ranks in $B$ with
Lemma~\ref{lem:word}(d), however,
because $B$ may contain large values, and $x-x_0$
can also be large.
Instead of finding $\Tvn{rank}(x-x_0,B)$ directly,
we therefore compute and return $\Tvn{rank}(y,\widetilde{B})$,
where $y=R_h(x-x_0)$ and $\widetilde{B}$ can be
viewed as approximations of $x-x_0$ and $B$,
respectively, that contain only small values.
What remains to be shown is that the differences
between $B$ and $\widetilde{B}$ and between
$x-x_0$ and $y$ do not influence the result.

$B[j^*]$ and $\widetilde{B}[j^*]$ coincide and are small.
Indeed,
$|B[j^*]|=|\widetilde{B}[j^*]|=|\sigma(A_{s\Ttwodots t})[j^*]|
\le (t-s)\cdot 2^\delta<2^{\delta+1}m=h$.
$\widetilde{B}$ is defined similarly as $B$, but
where $B$ is $\sigma(A_t)$ plus a constant (vector),
$\widetilde{B}$ is $\sigma(R_{2 h}(A_t))$ plus a constant.
If we define the \emph{jump} in $B$ at $j\in\{1,\ldots,m-1\}$
as $B[j+1]-B[j]$ and the jump in $\widetilde{B}$ at~$j$
analogously, it follows that for $j=1,\ldots,m-1$,
either $B$ and $\widetilde{B}$ have the same jump
at $j$, or the jump in $\widetilde{B}$ at~$j$ is $2 h$.
We may conclude that for $j=1,\ldots,m$, if
$\widetilde{B}[j]\not=B[j]$, then $|\widetilde{B}[j]|>2 h-h=h$.
Because $|y|\le h$, we have
$\Tvn{rank}(y,\widetilde{B})=\Tvn{rank}(y,B)$.
In other words, for our purposes $\widetilde{B}$ is
a sufficiently good approximation of~$B$.

To finish the argument, we must demonstrate that
$\Tvn{rank}(y,B)=\Tvn{rank}(x-x_0,B)$.
The basic reason why this is so is
that if $x$ is far from $x_0$,
it is also far from all components of $\sigma(A_t)$,
so that even the very bad approximation $y$
of $x-x_0$ has the same rank in $B$ as $x-x_0$.
It suffices to show that if $y\not=x-x_0$, then
$|(x-x_0)-y|<|(x-x_0)-B[j]|$ for $j=1,\ldots,m$.
But if $y\not=x-x_0$
and $j\in\{1,\ldots,m\}$, then, by the choice of $j^*$,
\[
\eqalign{
|(x-x_0)-y|&=|x-x_0|-h
\le|x-\sigma(A_s)[j]|-h\cr
&\le|x-(B[j]+x_0)|+|(B[j]+x_0)-\sigma(A_s)[j]|-h\cr
&=|(x-x_0)-B[j]|+|\sigma(A_t)[j]-\sigma(A_s)[j]|-h\cr
&\le|(x-x_0)-B[j]|+(t-s)\cdot 2^\delta-h
<|(x-x_0)-B[j]|.\cr}
\]

For each node $u$ in the tree $T$, let $D_u$ be the
searchable prefix-sums structure at~$u$.
The leaves of $T$ can be identified with
the $n$ positions in the sequence of integers
maintained by the overall data structure, and when
$u$ is a node in $T$ and $v$ is a child of~$u$,
the value recorded for $v$ in $D_u$
is the sum $s_v$ of the values stored in the leaf
descendants of~$v$.
In order to achieve a constant initialization
time, we (re-)initialize $D_v$ only when $s_v$
changes from~0 to some other value.
Initializing $D_v$ involves initializing a
constant number of simple variables and
small vectors, which can certainly happen in
constant time, and initializing two big vectors
of $O(m b)$ bits each, which can be done
with the method of
Lemma~\ref{lem:2.12} using another $O(m b)$ bits.
\end{proof}

\begin{lemma}
\label{lem:p-t}%
There is a choice dictionary that, for arbitrary
$N,c,t,r\in\TbbbN$ with $r\log c=O(w)$,
can be initialized for universe
size $N$, $c$ colors and tradeoff parameters $t$ and $r$
in constant time
and subsequently occupies
$N\log c+O({{c N\log N}/{(r t)}}+(\log N)^2+1)$ bits and,
if $r=1$ or if given access to tables of $O(c^r)$ bits
that can be computed in $O(c^r)$ time
and depend only on $N$, $c$, $t$ and $r$, executes
\Tvn{color} in constant time and \Tvn{setcolor}, 
\Tvn{rank} and \Tvn{select}
in $O(t+{{\log N}/{\log w}})$ time.
\end{lemma}

\begin{proof}
Without loss of generality assume that $N\ge 2$.
View each of the $N$ color values to be maintained
as a \emph{small digit} in the range $\{0,\ldots,c-1\}$
and take $\rbar=\Tceil{\sqrt{{r/2}}}$ and
$r'=\rbar^2$ ($\rbar$ is introduced only for
the sake of the proof of
Theorem~\ref{thm:p}).
Partition the $N$ small digits into
$N'=\Tceil{{N/{r'}}}$ groups of $r'$
consecutive small digits each, except that the
last group may be smaller, and represent the
small digits in each group through a \emph{big digit}
in the range $\{0,\ldots,C-1\}$,
where $C=c^{r'}$, except that the last big digit
may come from a smaller range.
A natural scheme represents small digits
$a_0,\ldots,a_{r'-1}$ through the integer
$\sum_{i=0}^{r'-1}a_i c^i$, but from the
point of view of correctness, the
\emph{encoding function} can be an arbitrary
bijection from $\{0,\ldots,c-1\}^{r'}$ to
$\{0,\ldots,C-1\}$.
We realize the encoding function and its
inverse through tables $Y\Tsub E$ and $Y^{-1}\Tsub E$.
In more detail, the encoding table $Y\Tsub E$ maps
$(r'\Tceil{\log c})$-bit concatenations of
the binary representations of $r'$ small digits,
called a \emph{loose representation}
of the $r'$ small digits,
to the corresponding big digit, and the
decoding table $Y^{-1}\Tsub E$ realizes the
exact inverse mapping.
If the last group of small digits contains fewer
than $r'$ small digits, it needs separate
encoding and decoding tables.
This is easy to handle and will be
ignored in the following.

We maintain the sequence of $N'$ big
digits in an instance $D$ of the data structure of
Theorem~\ref{thm:succincter},
whose space requirements are
$N\log c+O((\log N)^2)$ bits.
Forming $N''=\Tceil{{{N'}/t}}=\Tceil{N/{(r' t)}}=O({N/{(r t)}}+1)$
\emph{ranges} $R_1,\ldots,R_{N''}$
of ${t}$ consecutive big digits each,
except that the last range may be smaller,
we initialize all big digits in a range
exactly when for the first time a
digit in the range acquires a nonzero value.
We keep track of the ranges that have been initialized
using an instance of the choice dictionary
of Theorem~\ref{thm:systematic-2} with
universe size $N''$ and therefore
negligible space requirements.

In order to support \Tvn{rank} and \Tvn{select},
we maintain for each $j\in\{0,\ldots,c-1\}$
in an instance $D_j$  of the searchable prefix-sums structure
of Lemma~\ref{lem:prefSum}, initialized with
sum bit length
$\Tceil{\log(N+1)}$ and update bit length $1$, a sequence
of $N''$ integers, the $i$th of which is the
number of occurrences of the color $j$ in the range $R_i$,
for $i=1,\ldots,N''$.
An exception concerns the color~0:
Instead of storing the number $n_{i,0}$ of
occurrences of 0 in $R_i$,
we store the complementary number $s_i-n_{i,0}$,
where $s_i$ is the number of small digits
in $R_i$ (usually $r' t$).
The reason is that $s_i-n_{i,0}$ is initially zero,
which matches the initial value provided by~$D_j$.
In the following we assume that $D_0$ is
modified to replace counts of occurrences
communicated to and from a caller by their
complementary numbers.
The number of bits needed for $D_0,\ldots,D_{c-1}$
is $O({{c N\log N}/{(r t)}}+\log N)$.

To execute \Tvn{color}, we obtain the
relevant big digit from $D$ and use $Y^{-1}\Tsub E$
to convert it to the corresponding loose
representation, after which answering
the query is trivial.
The realization of \Tvn{setcolor} is similar,
except that a call $\Tvn{setcolor}_j(\ell)$
must additionally call \Tvn{update} in
$D_{j_0}$ and $D_j$, where $j_0$ is the color
of $\ell$ just before the call under consideration.
For $r$ larger than a constant---the only case
in which $Y\Tsub E$ and $Y^{-1}\Tsub E$ are
actually needed---the tables occupy
$O(2^{\lceil\log c\rceil r'}r\log c)=O(c^r)$ bits
and can, if realized according to
the natural scheme discussed above, be computed in $O(c^r)$ time.

To execute $\Tvn{rank}(\ell)$,
where
$\ell\in\{1,\ldots,N\}$,
we first compute $i=\Tceil{{\ell/{(r't)}}}$ and
$m=\ell-(i-1)r't$ and find $j=\Tvn{color}(\ell)$.
The value to be returned is
$D_j.\Tvn{sum}(i-1)$ plus the number
of occurrences
of the color $j$ among the $m$ first small digits
of $R_i$.
We compute the latter quantity by obtaining
the at most $t$ big digits of $R_i$ from $D$ one
by one
and processing each in constant time as follows while
accumulating a count of the number of
relevant occurrences of $j$ seen:
After obtaining the loose representation of the
big digit at hand with the aid of $Y^{-1}\Tsub E$,
we use the algorithm of
Lemma~\ref{lem:word}(c) to reduce the problem
of counting the number of relevant occurrences
of~$j$ in the loose representation
to one of counting the total number of
1s in at most $r'$ fields,
each of which occupies $\Tceil{\log c}$ bits
and holds a value of either 0 or~1.
Finally the latter problem is solved
by lookup in a table $Y\Tsub R$.
For $r$ larger than a constant,
$Y\Tsub R$ occupies
$O(2^{\lceil\log c\rceil r'}r\log r)=O(c^r)$ bits
and can be computed in $O(c^r)$ time.

To execute $\Tvn{select}(j,k)$, where $j\in\{0,\ldots,c-1\}$ and
$k\in\{1,\ldots,N\}$,
we first compute $i=D_j.\Tvn{search}(k)$.
If $i\le N''$,
the $k$th occurrence of $j$ is in $R_i$,
and the value to be returned is the position
of the $m$th occurrence of $j$ in $R_i$,
where $m=k-D_j.\Tvn{sum}(i-1)$.
If $i=N''+1$, $k$ is larger than the total number
of occurrences of $j$, and we return~0.
To find the $m$th occurrence of $j$
in $R_i$ if $i\le N''$, we proceed similarly
as in the case of $\Tvn{rank}$ and
obtain the big digits of $R_i$ one by one
from~$D$.
Using $Y^{-1}\Tsub E$ and $Y\Tsub R$ and spending
$O(t)$ time, it is easy
to identify the big digit that contains the
$m$th occurrence of $j$ in $R_i$ and the number of
occurrences of $j$ before that big digit.
Again using $Y^{-1}\Tsub E$
and the algorithm of Lemma~\ref{lem:word}(c)
to replace occurrences of $j$ by occurrences
of~1 in fields of $\Tceil{\log c}$ bits,
we finish the computation by consulting an
appropriate table $Y\Tsub S$ that,
for $r$ larger than a constant,
also occupies $O(2^{\lceil\log c\rceil r'}r\log r)=O(c^r)$ bits
and can be computed in $O(c^r)$ time.

For $j\in\{0,\ldots,c-1\}$,
each operation on $D_j$ runs in
$O(1+{{\log N''}/{\log w}})=O(1+{{\log N}/{\log w}})$ time, 
and every consultation of $D$
takes constant time.
Therefore \Tvn{color} runs in constant time
and every other operation
runs in $O(t+{{\log N}/{\log w}})$ time.
\end{proof}

\begin{theorem}
\label{thm:p}%
For all fixed $\epsilon>0$,
there is a choice dictionary that, for arbitrary
$n,c,t\in\TbbbN$,
can be initialized for universe
size $n$, $c$ colors and tradeoff parameter $t$
in constant time
and subsequently occupies
$n\log_2 c+O({{c n\log c\log(1+t\log n)}/{(1+t\log n)}}+n^\epsilon)$
bits and executes \Tvn{color} in
constant time and \Tvn{setcolor},
\Tvn{p-rank} and \Tvn{p-select} (and hence \Tvn{choice}
and \Tvn{uniform-choice})
and, given $O(c\log n)$ additional bits,
robust iteration in $O(t)$ time.
\end{theorem}

\begin{proof}
For the time being ignore the claim about
constant-time initialization and assume
the tables $Y\Tsub E$, $Y^{-1}\Tsub E$,
$Y\Tsub R$ and $Y\Tsub S$
of the previous proof
to be available.
Let $K$ be an integer constant with $K\ge{1/\epsilon}$.
If $c^K\ge n$ choose $r=1$.
Otherwise let $q={{(\log n)}/{(K\log c)}}\ge 1$
and choose $r$ as a positive integer with
$r\le q$, but $r=\Omega(q)$.
For $t\ge n^{1/3}$, the claim follows from the previous
lemma, used with $N=n$. 
For $t<n^{1/3}$
we use the construction 
shown
in
Fig.~\ref{fig:select}.

\begin{figure}
\begin{center}
\epsffile{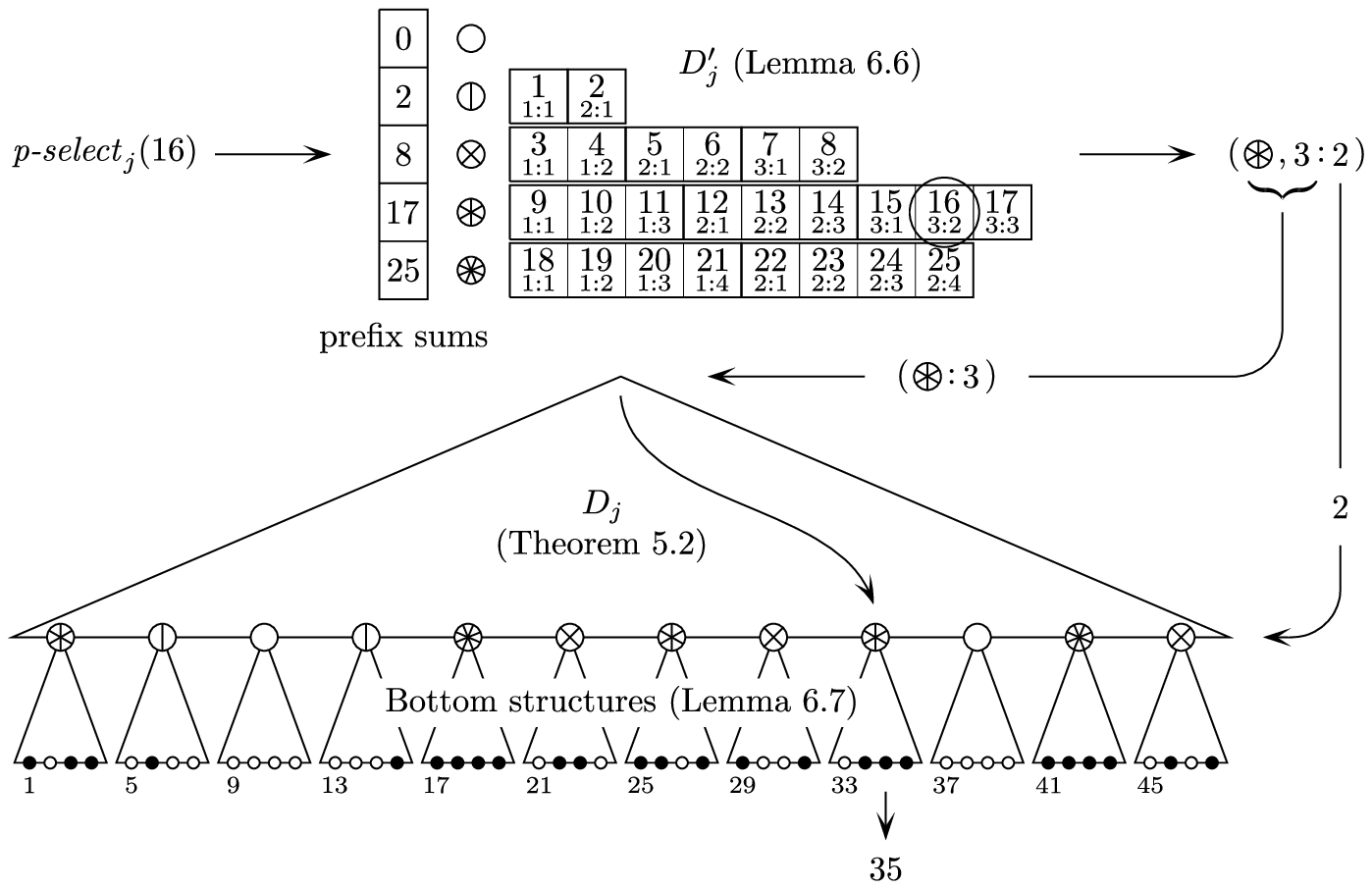}
\end{center}
\caption{An example evaluation of $\Tvn{p-select}_j$
for some color $j$ with
the combined data structure of Theorem~\ref{thm:p}.
First the argument of $\Tvn{p-select}_j$, 16,
is translated by $D'_j$
to the triple $(3,3,2)$ 
(shown as \Tmyw)
consisting of the relevant
weight, $\Tmys=3$, the index, 3, of the relevant segment of
weight 3 among all segments of weight 3, and the
index, 2, of the relevant element within that segment.
Then the relevant segment is identified with
the aid of $D_j$, and finally the relevant
element in that segment is located
with the corresponding bottom structure,
expressed in the form of
its global index, 35, and returned.
In the interest of clarity,
the figure assumes that the $\Tvn{p-select}_3$
function of $D_j$ in fact coincides
with $\Tvn{select}_3$.}
\label{fig:select}
\end{figure}

Computing $N$ as a positive integer with
$N=\Theta(t(\log n)^2)$, we partition the $n$
color values to be maintained into
$m=\Tceil{{n/N}}$ \emph{segments} of $N$
color values each, except that the last segment
may be smaller, and maintain each segment in an
instance of the data structure of
Lemma~\ref{lem:p-t} called a
\emph{bottom structure}.
The total number of bits occupied by all
bottom structures is
$m(N\log c+O({{c N\log N}/{(r t)}}+(\log N)^2+1))
=n\log c+O({{c n\log(t\log n)}/{(r t)}})
=n\log c+O({{c n\log c\log(t\log n)}/{(t\log n)}})$,
plus $O(c^r)=O(n^\epsilon)$ bits
for shared tables.

For each color $j\in\{0,\ldots,c-1\}$,
define the \emph{$j$-weight} of each segment
as the number of elements of color~$j$
in the segment.
We maintain the $j$-weights of all segments in an
instance $D_j$
of the data structure of Theorem~\ref{thm:nlogn},
with the $j$-weights
playing the role of the colors in~$D_j$.
Thus $D_j$ is initialized for universe size $m$
and $N+1$ colors.
With $a_i$ equal to the number of
elements of color~$j$ in segments of $j$-weight~$i$,
for $i=0,\ldots,N$, we also maintain the sequence
$(a_0,\ldots,a_N)$ in an instance $D'_j$ of the
searchable prefix-sum structure of Lemma~\ref{lem:prefSum},
initialized with sum bit length
$b=\Tceil{\log(n+1)}$
and update bit length
$\delta=\Tceil{\log(N+1)}$.
As in the previous proof, $D_0$ and $D'_0$ must
be treated slightly differently.
The total number of
bits occupied by
$D_0,D'_0,\ldots,D_{c-1},D'_{c-1}$ is
$O(c(m+N)\log(m+N)+c N\log n)=O({cn/{(t\log n)}})$.
As described at the end of Section~\ref{sec:trie},
we use an additional choice dictionary $D^*$ with
universe size $m+c$ and therefore
negligible space requirements to keep track
of and carry out the initialization of
the bottom structures and
$D_0,D'_0,\ldots,D_{c-1},D'_{c-1}$
as appropriate.

A \Tvn{color} query can be answered in
constant time by the relevant bottom structure.
Because $t\ge\log n$ or $\log N=O(\log w)$,
every operation of a bottom structure
executes in $O(t)$ time.
When a call of \Tvn{setcolor} changes the
color of an element, from $j_0$ to $j$, say,
this can be recorded in the relevant
bottom structure in $O(t)$ time, after which
the update must be reflected in $D_{j_0}$,
$D_j$, $D'_{j_0}$ and $D'_j$.
Since each of the two weight changes is
by 1 or $-1$, 
each of the updates of $D_{j_0}$ and $D_j$
can happen in constant time---from the
perspective of $D_{j_0}$ and $D_j$,
a color changes into a neighboring color.
Similarly, in each of $D'_{j_0}$ and $D'_j$,
the update changes two values in the sequence
maintained, each by at most~$N$.
A change of this magnitude is covered by the
update bit length of $D'_{j_0}$ and $D'_j$,
and the update can be executed
in $O(1+{{\log N}/{\log(1+{w/\delta})}})=
O(t)$ time.

Consider a call $\Tvn{p-select}_j(k)$
with $j\in\{0,\ldots,c-1\}$ and
$k\in\{1,\ldots,n\}$
(Fig.~\ref{fig:select}).
In the sequence $(a_0,\ldots,a_N)$ maintained
by $D'_j$, each element $a_i$ can be thought of
as representing $a_i$ elements of the top-level
universe $U=\{1,\ldots,n\}$, namely precisely those
that have color $j$ and
are located in segments of $j$-weight $i$.
In particular, $a_i$ is always a multiple of~$i$.
Provided that $k\le\sum_{i=0}^N a_i$,
$k$ designates a particular element $\ell\in U$
of color $j$ in a natural way:
First $i=D'_j.\Tvn{search}(k)$ selects
a particular $j$-weight, $i$, as the weight of~$\ell$.
Then $p=k-D'_j.\Tvn{sum}(i-1)$ is the index of
$\ell$ in the sequence of all elements of $U$
of color $j$ in segments of $j$-weight $i$, and finally
$q=\Tceil{{p/i}}$ is the index of the segment that contains~$\ell$,
among those of $j$-weight $i$,
and $p-(q-1)i$ is the index of $\ell$
within that segment.
Here ``index'' is to be understood as relative
to the orders imposed by the operation
$\Tvn{p-select}_i$ in $D_j$ and the
operation $\Tvn{select}_j$ in the
relevant bottom structure.
Altogether, the top-level call
$\Tvn{p-select}_j(\ell)$ reduces to one call of
each of \Tvn{search} and \Tvn{sum} in $D'_j$,
one call of \Tvn{p-select} in $D_j$, and one
call of \Tvn{select} in a bottom structure.
It can therefore be executed in $O(t)$ time.

To execute $\Tvn{p-rank}(\ell)$ for
$\ell\in U$, we first consult the relevant
bottom structure to find the color $j$
of $\ell$ and the index $k$ of $\ell$
among the elements of color $j$ in its segment $R$.
Then $D_j$ is queried for the $j$-weight $i$ of $R$ 
and the index $q$ of $R$ among the segments
of $j$-weight~$i$.
Finally the return value is obtained as
$D'_j.\Tvn{sum}(i-1)+(q-1)i+k$.
The procedure works in $O(t)$ time.

To equip the data structure with robust iteration,
we ``plant'' $c$ additional
instances of the choice dictionary of
Theorem~\ref{thm:nlogn}, one for each
color, on top of the bottom structures and
appeal to the general
trie-combination method of Section~\ref{sec:trie}.

Let us now drop the assumption that the
tables $Y\Tsub E$, $Y^{-1}\Tsub E$, $Y\Tsub R$
and $Y\Tsub S$ are available for free.
As in the proof of Lemma~\ref{lem:p-t},
define $\rbar=\Tceil{\sqrt{{r/2}}}$ and $r'=\rbar^2$.

Recall that
the task of $Y\Tsub E$ is to
map loose representations of sequences of
$r'$ small digits to big digits in an arbitrary
bijective manner and that $Y^{-1}\Tsub E$
should realize the inverse mapping.
We compute $Y\Tsub E$ and $Y^{-1}\Tsub E$
in a lazy fashion that combines techniques
used already in the proofs of
Lemmas \ref{lem:2.12} and~\ref{lem:permutation}.
We begin by setting $Y\Tsub E[0]:=0$
and $Y^{-1}\Tsub E[Y\Tsub E[0]]:=0$ and initializing
an integer $\mu$ to 0.
Subsequent loose representations are mapped to
the big digits $1,2,\ldots$ in the order in which
they present themselves to the encoding table.
More precisely, in order to compute
the big digit corresponding
to a loose representation $a$, we first check
whether $a$ was mapped previously in the same manner.
This is the case if $0\le Y\Tsub E[a]\le\mu$
and $Y^{-1}\Tsub E[Y\Tsub E[a]]=a$.
If so, the big digit corresponding to $a$
is simply $Y\Tsub E[a]$.
Otherwise $\mu$ is incremented, and the new
value of $\mu$ becomes the big digit corresponding
to~$a$, a fact recorded by executing
$Y\Tsub E[a]:=\mu$ and
$Y^{-1}\Tsub E[Y\Tsub E[a]]:=a$.
It is easy to see that whenever an entry in
$Y^{-1}\Tsub E$ is inspected by the data
structure of Lemma~\ref{lem:p-t}, it
has already been computed
(only encoded values are decoded).

$Y\Tsub R$ and $Y\Tsub S$ are also
provided in a lazy fashion, but present minor
additional technical difficulties.
We in fact realize $Y\Tsub R$ and $Y\Tsub S$
not as tables, but as constant-time functions
that carry out two table lookups each.

Recall that $Y\Tsub R$ operates on ``binarized'' loose
representations of big digits, ones in which
all occurrences of a color $j$ of interest have been
replaced by 1 and all occurrences of colors other than~$j$
have been replaced by~0, with each such value stored
in a field of $\Tceil{\log c}$ bits.
Correspondingly, define a \emph{big vector} to be
a sequence of $r'$ fields, each of $\Tceil{\log c}$
bits and containing a value drawn from $\{0,1\}$,
and view a big vector as composed of $\rbar$
\emph{blocks} of $\rbar$ fields each.
After a slight redefinition, the task of $Y\Tsub R$
is to map each pair $(a,p)$, where $a$ is a big
vector and $p\in\{1,\ldots,r'\}$, to the
sum of the $p$ first fields in~$a$.
We divide this task into two subtasks:
sum the fields in the first $i-1$ blocks in $a$,
where $i=\Tceil{{p/{\rbar}}}$;
and sum the first $p-(i-1)\rbar$
fields in the $i$th block in~$a$.

The first subtask is solved with a table $Y^{(1)}\Tsub R$:
For each big vector $a$, $Y^{(1)}\Tsub R[a]$ is the
sequence $(n_1,\ldots,n_{\rbar})$, where $n_i$ is the
sum of the fields in the $i$ first blocks in~$a$,
for $i=1,\ldots,\rbar$.
Thus $Y^{(1)}\Tsub R$ is a table of sequences of prefix sums.
Note that each sequence is of
$O(\rbar\log r')=O(w)$ bits, so that it can be handled
in constant time
(this is the reason for introducing $\rbar$).
Each use of the table needs only a single prefix sum that
must be picked out from the full sequence.
This organization of the table ensures that it can be
computed in a lazy fashion:
Each color change leads to at most two new big vectors,
the entry in $Y^{(1)}\Tsub R$ of each of which can be
computed in constant time using word parallelism
from an old entry.
The second subtask is handled in a very similar way
using a second table $Y^{(2)}\Tsub R$.

The task of $Y\Tsub S$ is to map each pair $(a,k)$,
where $a$ is a big vector and
$k\in\{1,\ldots,r'\}$, to the position of
the $k$th 1 in $a$, if any.
Again the task is divided into two subtasks,
each of which is handled in constant time. 
For the first subtask, we find the number $i$
of the block in $a$ that contains the $k$th 1---assume
for simplicity that there is such a block---by
computing
$i=\Tvn{rank}(k-1,Y^{(1)}\Tsub R[a])+1$ with the algorithm of
Lemma~\ref{lem:word}(d).
Let $n_{i-1}$ be the $(i-1)$th number in the
sequence $Y^{(1)}\Tsub R[a]$ (0 if $i=1$).
For the second subtask, we have to
locate the $(k-n_{i-1})$th 1 in the $i$th
block of~$a$.
This can be done in a similar way
using $Y^{(2)}\Tsub R$ in place of $Y^{(1)}\Tsub R$.

One may remark that the $O(c\log n)$ bits required to
carry out robust iteration are already contained in
the bound of the theorem except in the extreme case
$t=\Omega({n/{\log n}})$.
\end{proof}

If only \Tvn{p-select} and not \Tvn{p-rank} is
to be supported (e.g., if the only goal is to
realize the operation \Tvn{uniform-choice}),
it is possible to avoid the use of Lemma~\ref{lem:prefSum}
for $t=(\log n)^{O(1)}$.
In the context of Theorem~\ref{thm:p} and
with $N$ and $\delta$ defined as in its proof,
the ``bottom'' instances of the data structure of
Lemma~\ref{lem:prefSum}
(those incorporated, via Lemma~\ref{lem:p-t}, in
the bottom structures in the proof of
Theorem~\ref{thm:p} and in Fig.~\ref{fig:select})
can easily be replaced
by tries of constant height
of data structures that maintain the prefix sums
directly, realize \Tvn{update}
via a multiplication by
$1_{N,\delta}$, two shifts and an addition,
and execute \Tvn{search} with the algorithm
of Lemma~\ref{lem:word}(d).
In slightly greater generality, this method
yields a constant-time searchable prefix-sums structure
that maintains a sequence of $N+1$ integers,
each drawn from $\{0,\ldots,N\}$, under
arbitrary updates of single sequence elements.
Let us call such a structure an
\emph{$N$-structure}.

For the ``top'' instances $D'_j$ in the proof of
Theorem~\ref{thm:p} and in Fig.~\ref{fig:select},
avoiding Lemma~\ref{lem:prefSum} is more involved.
We sketch the construction.
The essential task of a top instance
$D'_j$ can be viewed as that of
maintaining a set of $s\le n$ indistinguishable
items, each with a weight in $\{0,\ldots,N\}$
(put differently, a multiset of weights),
under insertion and deletion of some (arbitrary) item
with a given weight and
an operation \Tvn{p-select} that maps each argument
$k\in\{1,\ldots,s\}$ to the pair $(i,p)$,
where $i$ is the weight of the
$k$th item and $p$ is its index 
within the set of items of weight~$i$,
for some ordering of the items.
A first solution to this problems stores the
items of weight $i$ in a doubly-linked list $L_i$,
for $i=0,\ldots,N$, and marks each list item
with its weight and
its distance to the end of its list.
The $N+1$ lists are stored compactly together in an
array $A$ of $s$ cells, each of $\Theta(\log n)$ bits,
and the positions in $A$ of the first items in each list
are recorded in a second array.
A new item of weight $i$
is stored in the first free cell in~$A$
and inserted at the beginning of $L_i$
and computes its distance-to-end value
as one more than that
of the formerly first item in~$L_i$.
To delete an item of weight $i\in\{0,\ldots,N\}$,
we first swap the first item in $L_i$ with the
item stored in the last used cell in~$A$
and then delete it,
which does not upset
the distance-to-end value of any other item.
To execute $\Tvn{p-select}(\ell)$, simply return
the pair of
the weight and one more than the distance-to-end value of the
item in $A[\ell]$.

Assume $n\ge 2$.
In order to reduce the space requirements
per color
from $O(n\log n)$ to $O({n/{(t\log n)}})$,
we aggregate the $s$ items into \emph{superitems}
of $N$ items of a common weight each, with up to $N-1$ items
left over for each weight in $\{0,\ldots,N\}$.
We store the numbers of left-over items
for each weight in an
$N$-structure $D\Tsub L$.
The superitems are maintained in $N+1$ lists
as described above.
Since their number is $O({n/N})$, the total
number of bits used is indeed $O({n/{(t\log n)}})$.
Consider an execution of $\Tvn{p-select}_j(k)$,
where $j$ is the color under consideration,
and let $n\Tsub L=D\Tsub L.\Tvn{sum}(N)$
be the total number of left-over items.
If $k\le n\Tsub L$,
compute $i=D\Tsub L.\Tvn{search}(k)$
and return the pair
$(i,k-D\Tsub L.\Tvn{sum}(i-1))$ as for the original
top-level structure $D'_j$.
Otherwise, with $k'=k-n\Tsub L$,
let $(i,p)$ be the
pair returned by the list-based structure,
called with argument $\Tceil{{{k'}/N}}$,
and return the pair $(i,(p-\Tceil{{k'}/N})N+k')$.
Thus the left-over items are numbered before
the items in superitems, and the global number of
an item in a superitem is $n\Tsub L$ plus $N$ times the
number of superitems before
its own superitem plus its number
within the superitem.

\section{Nonsystematic Choice Dictionaries}
\label{sec:nonsystematic}%

In this section we describe our most space-efficient
but also most complicated choice dictionaries.
We first consider the (somewhat easier) case
in which the number $c$ of colors is a power of~2---until
and including Theorem~\ref{thm:unsystematic-f}---and
subsequently detail the changes necessary to cope
with general values of~$c$.

As the reader may recall from the introduction,
the game is basically one of squeezing navigational
information into the leaves of a tree.
Lemma~\ref{lem:j-free} below describes a leaf that
can be in either the \emph{standard representation},
which offers no potential for storing
additional information,
or the \emph{$\jj$-free representation} for some
color $\jj$ that happens not to be represented
at the leaf.
In the latter case, information pertaining to the
tree path that ends at the leaf can be stored
in the leaf together with the usual information
kept there.
The proof of the central Lemma~\ref{lem:smalltree} describes how
to combine many such leaves to obtain a
tree that supports the operations \Tvn{color},
\Tvn{setcolor} and \Tvn{successor}.
We first address the overall data organization of the tree
and then discuss how to navigate in the tree,
after which the implementation of the query operations
\Tvn{color} and \Tvn{successor} is fairly
straightforward.
The final part of the proof of
Lemma~\ref{lem:smalltree} describes how to
re-establish the data-representation invariants
of the tree
after a call of the update operation \Tvn{setcolor}.
Lemma~\ref{lem:unsystematic-tf} essentially shows
how to put many such trees next to each other
to cover a larger universe, and
Theorem~\ref{thm:unsystematic-f} finally obviates
the need for precomputed tables.

In the following, let $f$ and $t$ be given
positive integers, take $w'={w/f}$ and $c=2^f$, assume
that $d={w/{(2 c^2 f t)}}$ is an integer and at
least~2 and let $N=d^t$.

\begin{lemma}
\label{lem:j-free}%
There is a choice dictionary $D$
with universe size $w'$ and for $c$ colors
that can be initialized in constant
time and subsequently occupies $w$ bits and
executes \Tvn{color} and \Tvn{setcolor}
in $O(f)$ time and
\Tvn{successor} in $O(c)$ time.
Moreover, during periods in which
$S_\jj=\emptyset$, where
$(S_0,\ldots,S_{c-1})$ is $D$'s client vector and
$\jj\in\{0,\ldots,c-1\}$,
$D$ supports two additional operations
that execute in $O(c)$ time:
Conversion from
the (initial) \emph{standard representation}
to the \emph{$\jj$-free representation}
and conversion back to
the standard representation.
When $D$ is in the $\jj$-free representation, $\jj$
must be supplied as an additional argument in
calls of \Tvn{color}, \Tvn{setcolor} and \Tvn{successor}
and in requests for conversion to the standard representation
(we will, however, suppress this in our notation).
In return, the $c d t$ bits
of the $\jj$-free representation of $D$
whose positions are multiples of $2 c f$ are unused,
i.e., free to hold unrelated information.

Alternatively, for arbitrary fixed $\epsilon>0$,
if given access to tables of
$O(c^{\epsilon c^2})$ bits that can be computed in
$O(c^{\epsilon c^2})$ time and depend
only on~$c$, $D$ can execute
\Tvn{color}, \Tvn{setcolor} and \Tvn{successor}
in constant time.
\end{lemma}

\begin{proof}
We can view $D$'s task as that of maintaining a
sequence of $w'$ \emph{digits} to base $c$.
The standard representation is simply the concatenation,
in order, of the $f$-bit binary representations of the 
$w'$ digits.
With this representation,
the operations can be carried out as for the
data structure of Lemma~\ref{lem:atomic-c}.

For each $\jj\in\{0,\ldots,c-1\}$, the
$\jj$-free representation partitions the $w'$ digits
into \emph{big groups}
of $2 c^2$ consecutive digits each and stores
each big group in $(2 c f-1)c$ rather
than $2 c^2 f$ bits, leaving free every bit whose
position is a multiple of $2 c f$,
as promised in the lemma.
First, using the increasing bijection
$\Tvn{skip}_\jj$ from $\{0,\ldots,c-1\}\setminus\{\jj\}$
to $\{0,\ldots,c-2\}$, the $2 c^2$ digits
to base $c$ of each big group are
transformed into $2 c^2$ digits to base $c-1$.
Call this the \emph{$\jj$-intermediate representation}.
Then the $2 c^2$ transformed digits are partitioned
into $2 c$ \emph{small groups} of $c$ consecutive digits each,
and each small group is viewed as a $c$-digit integer
written to base $c-1$ and is represented in binary
in $c f-1$ bits, which is possible because
$c\log(c-1)=c\log c+c\log(1-{1/c})
\le c f+c\ln(1-{1/c})\le c f-1$.
At this point, within each big group, the $2 c$ bits whose
positions are multiples of $c f$ are unused.
For $j=0,\ldots,c-2$, we store in the $(2 j+1)$st
such bit a \emph{summary bit} equal to 1 exactly if
the color $\Tvn{skip}^{-1}_{\jj}(j)$
occurs as a transformed digit in the big group.
The summary bits are redundant, but help us to
execute \Tvn{successor} in constant time.
One bit is wasted, and
the other half of the
$2 c$ bits are the promised free bits.

In order to convert $D$ from the standard to the
$\jj$-free representation, for some $\jj\in\{0,\ldots,c-1\}$
with $S_\jj=\emptyset$, first the function
$\Tvn{skip}_\jj$ is applied independently
to each digit.
Say that the $f$ consecutive bits in which
a digit is stored form a \emph{field}.
By Lemma~\ref{lem:word}(c), we can compute an
integer $z$, each of whose fields stores~1
if the corresponding digit is $\le\jj$
and 0 otherwise.
The function $\Tvn{skip}_\jj$ can now be applied in parallel
to all fields by a subtraction of $1_{w',f}-z$.
Subsequently, within each small group of $c$
digits, say $a_0,\ldots,a_{c-1}$, we must
convert $\sum_{i=0}^{c-1}a_i c^i$ to
$\sum_{i=0}^{c-1}a_i(c-1)^i$.
Since the digits $a_i$ are readily available
as the values of $f$-bit fields,
this can be done in $O(c)$ time for all small groups
using a word-parallel version of Horner's scheme
in a straightforward manner.
Finally, for
each 
$j\in\{0,\ldots,c-2\}$,
within each big group a summary bit
must be computed and
stored in the appropriate position
within the big group.
To this end, first apply the algorithm of
Lemma~\ref{lem:word}(c) at most twice, with $k=j$ and,
if $j>0$, with $k=j-1$, followed by bitwise
Boolean operations, to obtain an integer in
which the most significant digit of each
big group is zero, while the remaining bits
of the big group are also zero
if and only if the digit $j$ does not occur
in the big group.
A subtraction from $1_{{{w'}/{(2c^2)}},2c^2 f}\ll(2c^2 f-1)$
followed by the computation of \textsc{and}
and \textsc{xor} with the same number and
a suitable shift finishes the computation.

For the conversion in the other direction,
i.e., the conversion from
$\sum_{i=0}^{c-1}a_i(c-1)^i$ to
$\sum_{i=0}^{c-1}a_i c^i$ within each small group,
after clearing the bits whose positions are
multiples of $c f$ (those that held
summary and extraneous bits),
we compute the digits $a_0,\ldots,a_{c-1}$
by repeatedly obtaining the remainder modulo~$c-1$,
which yields the next digit, and keeping
only the integer part of the quotient with $c-1$.
Except for the division by $c-1$ with
truncation, the necessary steps are easily carried out in
constant time per digit and $O(c)$ time altogether.
Since division is not readily amenable to
word parallelism, we replace division by $c-1$
by multiplication by its approximate inverse.
More precisely, we carry out the division in
constant time using the relation
$\Tfloor{a/{(c-1)}}
=\Tfloor{{{a\cdot\Tceil{{{c^{2 c}}/{(c-1)}}}}/{c^{2 c}}}}$.
To see the validity of the relation for
all integers $a$ with $0\le a<c^c$,
simply observe that
${a/{(c-1)}}\le {{a\cdot\Tceil{{{c^{2 c}}/{(c-1)}}}}/{c^{2 c}}}
<{a/{(c-1)}}+{1/{c^c}}<{{(a+1)}/{(c-1)}}$
and note that there is no integer
strictly between ${a/{(c-1)}}$ and ${{(a+1)}/{(c-1)}}$.
The product $a\cdot\Tceil{{{c^{2 c}}/{(c-1)}}}$
may have more than $c f$ bits.
It has no more than $3 c f$ bits, however, so it can
be computed using ``triple precision'', which
we simulate by handling the small groups
in three rounds, each round operating only on
every third group.
Truncated division by $c^{2 c}$ is, of course, realized
as a right shift by $2 c f$ bit positions followed by
a ``masking away'' of the unwanted bits.
At the very end, to get from the
$\jj$-intermediate to the standard representation,
$\Tvn{skip}^{-1}_\jj$ must be
applied independently to each field.
This can be done similarly as described above
for $\Tvn{skip}_\jj$.

As detailed above, the conversion between
the standard and the $\jj$-intermediate representations
depends on $\jj$, but takes only constant time.
In contrast, the conversion between the
$\jj$-intermediate and the $\jj$-free representations
takes $\Theta(c)$ time, but is independent of~$\jj$.
This observation is important to
the proof of Theorem~\ref{thm:unsystematic-f}.

Assume that $D$ is in the $\jj$-free representation, for
some $\jj\in\{0,\ldots,c-1\}$, and that a
call $\Tvn{successor}(j,\ell)$ is to be
executed for some
$j\in\{0,\ldots,c-1\}\setminus\{\jj\}$
and $\ell\in\{0,\ldots,w'-1\}$.
Suppose, for ease of discussion, that the return value
$\ell'$ is nonzero, and let $G$ and $G'$ be the big
groups that contain the $(\ell+1)$st and the $(\ell')$th
digit, respectively.
Applying to $G$ a computation that, informally,
converts the single big group $G$ to the standard
representation, we can test whether
the $(\ell')$th digit belongs to $G$ and,
if so, find and return~$\ell'$.
Otherwise we locate $G'$ by applying an algorithm
of Lemma~\ref{lem:word}(a) to a suitable suffix of
those summary bits that pertain
to the color $j$, with all other bits cleared,
after which $\ell'$ can be found by converting
$G'$ to the standard representation
as done previously for~$G$.
The computation runs in $O(c)$ time, its
bottleneck being the conversions to the
standard representation.
It is easy to see that \Tvn{color} and \Tvn{setcolor}
can be executed in $O(f)$ time by computing
the relevant power of $c$ via repeated squaring.

Alternatively, the conversion of single big groups
from the $\jj$-intermediate
to the $\jj$-free representation
and vice versa can be carried out
by table lookup.
A table for each direction
of the conversion maps each sequence of
$c^2$ possible digits
to the corresponding other
representation and therefore has $O(c^{c^2})$ entries
of $O(c^2 f)$ bits each.
For fixed $\epsilon>0$ and
for $c$ larger than a suitable constant,
we can instead
use repeated table lookup,
mapping at most ${{\epsilon c}/2}$ small groups of
$c$ digits each at a time.
This reduces the number of bits in the tables
and the time needed to compute them to
$O(c^{3+{{\epsilon c^2}/2}})=O(c^{\epsilon c^2})$.
In the case of the conversion to the
$\jj$-free representation, each table entry
for at most ${{\epsilon c}/2}$ small groups
must provide suitable summary bits for the
small groups concerned, and the composition 
of such entries includes forming the bitwise
\textsc{or} of the partial summaries.
\end{proof}

\begin{lemma}
\label{lem:smalltree}%
There is a choice dictionary $D$ with
universe size $N w'$ and for $c$ colors
that can be initialized in constant time,
uses $N w+2$ bits and
supports \Tvn{color} in $O(t+f)$ time and
\Tvn{setcolor} 
and \Tvn{successor}
in $O(t+c)$ time.

Alternatively, for arbitrary fixed $\epsilon>0$,
if given access to tables of
$O(c^{\epsilon c^2})$ bits that can be computed in
$O(c^{\epsilon c^2})$ time and depend
only on~$c$, $D$ can execute
\Tvn{color} and \Tvn{successor}
in $O(t)$ time.
\end{lemma}

\begin{proof}
Let $T=(V,E)$ be a complete $d$-ary outtree of
depth~$t$, whose leaves, in the order from
left to right, we will identify with
the integers $1,\ldots,N$.
Let $r$ be the root of $T$ and, for all $u\in V$,
take $T_u$ to be the maximal subtree of $T$
rooted at~$u$.

Let the client vector of $D$ be
$(S_0,\ldots,S_{c-1})$.
We divide the universe
$U=\{1,\ldots,N w'\}$ into $N$
\emph{segments} 
$U_1,\ldots,U_N$ of $w'$ consecutive integers each.
For each $u\in V$, call $u$ \emph{full}
if $S_j\cap U_i\not=\emptyset$ for all
$j\in\{0,\ldots,c-1\}$ and all leaf
descendants $i$ of $u$, and
\emph{deficient} otherwise.
Informally, a deficient node is one that has a
leaf descendant with a missing color.
For each $u\in V$, let the \emph{spectrum} of $u$
be the string $b_0\cdots b_{c-1}$ of $c$
bits defined as follows:
If $u$ is deficient,
then for $j=0,\ldots,c-1$, $b_j=1$
exactly if $S_j\cap U_i\not=\emptyset$
for some leaf descendant $i$ of~$u$
(informally, if the color~$j$ is represented in $T_u$).
If $u$ is full, as a special convention,
$b_0\cdots b_{c-1}=0\cdots 0$, a
bit combination that cannot otherwise occur.
If $b_0\cdots b_{c-1}=100\cdots 0$
(only the color~0 is represented in $T_u$),
we say that $u$ is \emph{empty};
this is initially the case for every node~$u$.
If a node in $T$ is deficient but not empty,
we call it \emph{light}.
For each inner node $u$ in $T$, define the
\emph{navigation vector} of $u$ to be the concatenation
$\gamma_1\cdots \gamma_d$, where $\gamma_1,\ldots,\gamma_d$ are the
spectra of the $d$ children of $u$ in the
order from left to right.

For $i=1,\ldots,N$, let $\mathcal{P}'_i$
be the semipartition
$(S_0\cap U_i,\ldots,S_{c-1}\cap U_i)$ of $U_i$
and let $\mathcal{P}_i$ be the semipartition
of $\{1,\ldots,w'\}$ obtained from $\mathcal{P}'_i$
by subtracting $(i-1)w'$ from every element
in each of its sets.
We do not store~$T$.
Instead, for $i=1,\ldots,N$, $\mathcal{P}_i$
is stored in an instance $D_i$ of the
choice dictionary of Lemma~\ref{lem:j-free}
called a \emph{leaf dictionary},
and $D_1,\ldots,D_N$ are in turn stored
in $N$ words $H_1,\ldots,H_N$ of $w$ bits each.
Two additional \emph{root bits} indicate whether the root $r$
of $T$ is full and whether it is empty.
It is helpful to think of $H_i$ as ``normally''
storing $D_i$,
for $i=1,\ldots,N$.
If this were always the case and all navigation vectors
were available, a call of $\Tvn{successor}(j,\ell)$
could essentially use navigation vectors to find
a path from $r$ to the leftmost leaf $i$ in $T$
such that $S_j\cap U_i$ contains an 
element larger than $\ell$, if any,
and the smallest such element
could be obtained through a call of
$D_i.\Tvn{successor}$.
Moreover, \Tvn{setcolor}
could update navigation vectors as appropriate.
However, we have no space left to store
navigation vectors, and so have to proceed differently.

The parent of every light node in $T$ other than $r$
is also light, and
every deficient inner node in $T$ has at least
one deficient child.
Let the \emph{preferred child} of a deficient
inner node be its leftmost light child if
it has at least one light child, and its leftmost
empty child otherwise.
Let $Q=(V_Q,E_Q)$ be the subgraph of $T$ induced
by the edge set $E_Q$ obtained as follows:
First include in $E_Q$ all edges
from a light inner node to its preferred child.
Then, for every empty node $v$
that has an incoming edge in $E_Q$,
include in $E_Q$ the edges on the path
from $v$ to its leftmost leaf descendant.
$Q$ is a collection of node-disjoint paths
called \emph{light paths}, each of which ends
at a leaf in~$T$.
When $P$ is a light path that starts at a (light) node
$u$ and ends at a (deficient) leaf $v$, we call $u$ the
\emph{top node}, $v$ the \emph{proxy}
and the leftmost leaf descendant of $u$
(that may coincide with~$v$)
the \emph{historian} of $P$ and of every node on~$P$.
A light node in $T$ that is neither the root nor a leaf
is a top node exactly if it is not the preferred child
of its parent, i.e., if it has at least one
light left sibling.
No proper ancestor of a top node $u$ can have
a descendant of $u$ as its leftmost leaf descendant,
so a leaf is the historian of at
most one light path.
If $h$ is the historian of a light path $P$, the
top node and the proxy of $P$ are also said to be the
top node and the proxy, respectively, of $h$.
These concepts are illustrated in Fig.~\ref{fig:lightpaths}.
A leaf $i$ cannot be the historian
of one light path and the proxy of another,
since otherwise the two corresponding
top nodes would both be ancestors of $i$
and the path between them would contain
only gray nodes and be
part of a light path, an impossibility.
A similar argument shows that in the
left-to-right order of the leaves of~$T$,
no historian or proxy lies
strictly between a historian and its proxy.

\begin{figure}
\begin{center}
\epsffile{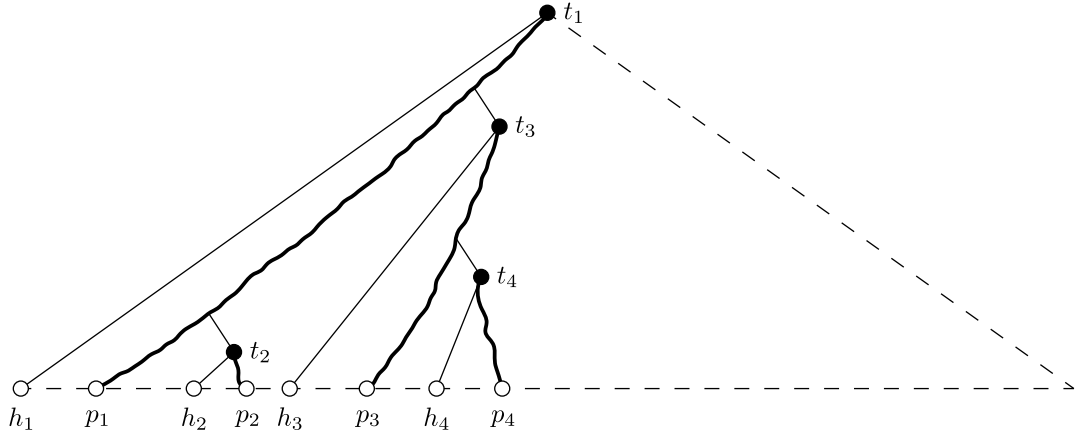}
\end{center}
\caption{Example light paths (drawn thicker).
Top nodes, historians and proxies are labeled
``$t$'', ``$h$'' and ``$p$'', respectively,
and a subscript identifies the associated
light path.}
\label{fig:lightpaths}
\end{figure}

Suppose that the nodes on a light path $P$
are $u_1,\ldots,u_k$, in that order.
Then the \emph{history} of
$P$ and of each of $u_1,\ldots,u_k$
is the concatenation of the
navigation vectors of $u_1,\ldots,u_{k-1}$,
in that order ($u_k$, as a leaf, has
no navigation vector).
An important fact to note is that if
a leaf dictionary
is in the $\jj$-free representation for some~$\jj\in\{0,\ldots,c-1\}$,
then it allows the history of a light path to be stored
in its $c d t$ free bits.
In order for this actually to be possible, we assume
that histories (and, by extension,
navigation vectors and spectra) are represented with 
gaps of $2 c f-1$ bit positions between consecutive bits,
so that a history spreads over up to an entire
$w$-bit word.
To fill up the word,
we store a history of $c d k$ bits
prefixed by $c d(t-k)$ arbitrary bits,
so that the positions in the word of the bits that make up
the navigation vector of a node $u$
depend only on the height of~$u$.

We represent $(S_0,\ldots,S_{c-1})$ using the following
storage scheme:
Let $\pi$ be the permutation of $\{1,\ldots,N\}$
that maps each $i\in\{1,\ldots,N\}$ to itself,
except that $\pi(p)=h$ and $\pi(h)=p$ for each
pair of a proxy $p$ and its historian~$h$.
Then the following holds for $i=1,\ldots,N$:
$D_i$ is stored in $H_{\pi(i)}$, and
\begin{itemize}
\item
if $i$ is a proxy (and therefore deficient),
$D_i$ is in the $\jj$-free representation,
where $\jj=\min\{j\in\TbbbN_0\mid 0\le j\le c-1$ and
$S_j\cap U_i=\emptyset\}$---we
will say that
$D_i$ is in the
\emph{compact representation}---and
$H_{\pi(i)}$ stores not only $D_i$, but also the history of $i$;
\item
if $i$ is empty but not a proxy,
$D_i$ may not have been initialized;
equivalently, $H_{\pi(i)}$ may hold an arbitrary value;
\item
in all remaining cases, i.e., if $i$ is neither a
proxy nor empty,
$D_i$
is in the standard representation.
\end{itemize}

\noindent
This paragraph tries to motivate the
not-so-natural storage scheme.
The usefulness of navigation vectors was
already observed above.
All nontrivial navigation vectors
are contained in the histories of the proxies.
As we have seen, if $p$ is a proxy and
therefore deficient, $\mathcal{P}_p$ can
be represented sufficiently compactly to allow
the history of $p$ to be stored with it.
However, when the history of a proxy $p$
is needed, $p$ is not known, so $H_p$ cannot
be located.
The top node of $p$ and hence also the
historian $h$ of $p$ are known, however, so
we store $D_p$ and the history of $p$
in $H_h$ rather than in $H_p$.
In return (unless $h=p$), $H_p$ must hold
(the standard representation of) $D_h$.
The terms ``historian'' and ``proxy'' serve as reminders
that a historian (more precisely, the
corresponding storage word) holds a history
(namely that of its proxy), whereas a proxy
holds the data of what may be
another leaf (namely of its historian).
The convention that $H_i$ may be arbitrary if
$i$ is an empty leaf that is not in use as a proxy
is necessary to guarantee
a constant initialization time.

If we represent a current node in $T$ as suggested
in Section~\ref{sec:trie}, i.e., through
the triple $(j,k,d^j)$, where $j$ is the height
in $T$ of the current node and $k$ is one more
than the number of nodes of height $j$ to its left,
we can navigate in $T$ as described
in Section~\ref{sec:trie}.
One operation that was not considered there and that
we need now is computing the leafmost leaf
descendant of the current node.
This is easy:
If the current node is (represented through)
$(j,k,d^j)$, its leftmost leaf descendant in $T$ is
$(0,(k-1)\cdot d^j+1,1)$.

\begin{proposition}
\label{prop:navigation}%
Let $u$ and $v$ be inner nodes in $T$ such that
$v$ is a child of $u$ and assume that we know
the navigation vector of~$u$, whether $u\in V_Q$
(i.e., whether $u$ belongs to a light path)
and, if $u$ is light, its history.
Then, in constant time, we can compute the spectrum
and the navigation vector of~$v$,
decide whether $v\in V_Q$
and whether $v$ is a top node and, if
$v$ is light, compute its history.
\end{proposition}

\begin{proof}
The spectrum of $v$ can be read off the
navigation vector of~$u$.
If $v$ is full or empty, its navigation vector
is trivial, namely the concatenation of $d$
copies of either $0\cdots 0$ or $100\cdots 0$,
$v$ is not a top node, and $v\in V_Q$
exactly if $u\in V_Q$ and $v$ is
$u$'s preferred child.
The latter condition can be tested in constant
time by inspection of
the navigation vector of $u$.
Assume now that $v$ is light, so that $v\in V_Q$.
Then $v$ is a top node exactly
if it has at least one light
left sibling.
If this is the case, the history of~$v$
is stored at $v$'s leftmost leaf descendant $h$
(more precisely, in $H_h$),
from where it can be retrieved in constant time.
If $v$ is not a top node, it belongs to the
same light path as $u$,
whose history is known by assumption.
The navigation vector of $v$ can be extracted
from $v$'s history in constant time.
\end{proof}

Proposition~\ref{prop:navigation} provides the general
step in an inductive argument to show that we can
traverse a root-to-leaf path in $T$ in constant
time per edge, always---until a
leaf is reached---knowing
the navigation vector of the current node,
whether it is a top node,
whether it belongs to a light path and, except
in the case of the root $r$ of $T$,
its spectrum.
As for the inductive basis, $r$
is a top node and belongs to a light
path if and only if $r$ is light,
and whether this is the case is
indicated by the two root bits.
If $r$ is a top node, its history
is available
in $H_1$ and, as above, the navigation vector
of $r$ can be extracted from its history in constant time.
If $r$ is not a top node, it is full or empty,
and its navigation vector is trivial, as above.

Our traversals of root-to-leaf paths in $T$
(called ``descents'') will be carried out by
starting at $r$ and repeatedly applying a
\emph{selection rule} at the current node~$u$
until a leaf is reached.
The selection rule indicates the child of $u$
at which the descent is to be continued.
We use three different selection rules that
we name for easier reference:

\begin{description}
\item[{\normalfont``leaf-seeking''$(i)$}]
($i$ is a leaf descendant of the current node~$u$):
Step to that child of $u$ that is an
ancestor of $i$.
\item[{\normalfont``proxy-seeking''}]
(the current node $u$ belongs to a light path):
Step from $u$ to its
preferred child.
\item[{\normalfont``color-seeking''$(j)$}]
($j\in\{0,\ldots,c-1\}$ and
the color $j$ is represented in $T_u$,
where $u$ is the current node):
Step from $u$ to the leftmost child of $u$
in whose spectrum the $(j+1)$st bit is set.
\end{description}

Using algorithms of Lemma~\ref{lem:word} for the rules
``proxy-seeking'' and ``color-seeking'',
we can apply each of the selection rules
above in constant time.

The permutation $\pi$ is not explicitly recorded.
As the following proposition shows, however,
we can compute $\pi(i)$ for arbitrary given
$i\in\{1,\ldots,N\}$.
What the proposition actually says is that we
can compute both $\pi(i)$ and the information
necessary to make sense of the 
contents of
$H_{\pi(i)}$.

\begin{proposition}
\label{prop:semi}%
Given $i\in\{1,\ldots,N\}$, in $O(t)$ time
we can compute $\pi(i)$, the spectrum of~$i$,
and whether $i$ is a proxy.
\end{proposition}

\begin{proof}
Use a first descent in $T$ with the selection rule
``leaf-seeking''$(i)$ to 
determine if $i\in V_Q$ and to
compute the spectrum of~$i$,
which will be known when the leaf $i$ is reached.
Now $i$ is a proxy exactly if $i\in V_Q$.
In a second descent in $T$,
initially again use the selection rule
``leaf-seeking''$(i)$.
Starting at the time when the current node
is first a top node, if ever, always remember
the historian $h$ of the
most recently visited top node
(the history stored in $H_h$ is used for navigational
purposes anyway).
If and when the current node becomes a top node
with $i$ as its leftmost leaf descendant,
(this will happen at some point exactly if
$i$ is a historian),
permanently change the
selection rule to ``proxy-seeking'' and
continue the descent.
Let $k$ be the leaf reached.
If the selection rule is still ``leaf-seeking'' at this time 
and $k=i\in V_Q$, 
$i$ is a proxy with historian $h$ and
$\pi(i)=h$;
otherwise $\pi(i)=k$.
\end{proof}

\noindent
$D$'s operations are implemented as follows:

\medskip\noindent
$\Tvn{color}$:
To execute $\Tvn{color}(\ell)$
for $\ell\in\{1,\ldots,N w'\}$, take
$i=\Tceil{{\ell/{w'}}}$ and $m=\ell-(i-1)w'$.
Thus $\ell$ is the $m$th element of the
$i$th segment $U_i$.
Use the algorithm of
Proposition~\ref{prop:semi} to compute
$\pi(i)$ and the related information.
If $i$ is empty, return~0.
Otherwise determine whether $D_i$
(stored in $H_{\pi(i)}$) is
in the $\jj$-free representation for some
$\jj$ and, if so, for which $\jj$
(Lemma~\ref{lem:word}(b)).
Using the information just computed to
consult $D_i$,
return $D_i.\Tvn{color}(m)$.

\medskip\noindent
$\Tvn{successor}$:
To execute $\Tvn{successor}(j,\ell)$
for $j\in\{0,\ldots,c-1\}$ and $\ell\in\{1,\ldots,N w'\}$,
initially proceed similarly as in
the case of \Tvn{color}:
Take $i=\Tceil{{\ell/{w'}}}$ and $m=\ell-(i-1)w'$
and use the algorithm of
Proposition~\ref{prop:semi} to compute
$\pi(i)$ and the related information.
If $i$ is empty, $m<w'$ and $j=0$,
return $\ell+1$.
If $i$ is nonempty, 
use $D_i$ (stored in $H_{\pi(i)}$) to
compute $k=D_i.\Tvn{successor}(j,m)$
and, if $k\not=0$, return $(i-1)w'+k$.

If no value was returned until this point
(the answer could not be established locally
in the $i$th segment),
again traverse the path $P$ in $T$ from $r$ to $i$.
With $X$ equal to the set of right siblings $u$
of inner nodes on $P$ 
such that the $(j+1)$st bit is set in $u$'s spectrum,
determine whether
$X=\emptyset$ and, if so, return~0.
Otherwise proceed as follows:
Compute the node~$u$
in $X$ that follows $i$ most closely in
a preorder traversal of~$T$
(thus $u$ is the closest right sibling in $X$
of the node on $P$ of maximum depth among
those with right siblings in~$X$).
Carry out a partial descent in $T$ that starts
at $u$ and uses the selection rule
``color-seeking''$(j)$ throughout.
Let $k$ be the leaf reached and
use the algorithm of
Proposition~\ref{prop:semi} to compute $\pi(k)$
and the related information.
If $k$ is empty, return $(k-1)w'+1$
(we must have $j=0$).
Otherwise use $D_k$ (stored in $H_{\pi(k)}$) to return
$(k-1)w'+D_k.\Tvn{successor}(j,0)$.

\medskip\noindent
$\Tvn{setcolor}$:
Let us use the terms ``old'' and ``new'' to refer
to states before and after the update
under consideration, respectively.
To execute $\Tvn{setcolor}(j,\ell)$
for $j\in\{0,\ldots,c-1\}$ and $\ell\in\{1,\ldots,N w'\}$,
once more take
$i=\Tceil{{\ell/{w'}}}$ and $m=\ell-(i-1)w'$.
Find the old color $j_0$ of $\ell$,
determine whether $S_{j_0}\cap U_{i}=\emptyset$
after the update and use this
to compute the new spectrum of $i$.
Save the old root bits and
traverse the path $P$ in $T$ from $r$ to $i$, collecting
the concatenation $\Gamma$ of the navigation vectors
of all inner nodes on $P$.
Use $\Gamma$ to traverse $P$ backwards
and update the histories of all old top nodes
encountered and the root bits to reflect the change,
if any, in the spectrum of~$i$.
Since the spectrum of every inner node
in $T$ is a simple function of those of
its children, this is a straightforward
bottom-up computation.
Also explicitly compute the concatenation
$\Gamma'$
of the
new navigation vectors of the inner nodes on~$P$.
Guided by $\Gamma$ and $\Gamma'$, we can traverse $P$
in the forward and backward directions
in constant time per node visited, always
knowing both the old and the new navigation vector
of the current node, even though the
histories stored in $H$ may be temporarily
inconsistent during the update.
What remains is to actually record the new
color of $\ell$ and to modify the light paths
implicit in the values in $H$ accordingly.
After describing a procedure for achieving this,
we will argue that whenever the procedure needs the
(old and new) navigation vector of a node
outside of $P$, the navigation vector can
be obtained from the history stored in
a word that has not (yet) been modified in the
course of the update.

We consider three cases.
In Case~1 the set of light paths does not change.
In Case~2 the update creates a new light path,
which may shorten a single existing light path.
In Case~3 the update destroys
a light path, which may
lengthen a single existing light path.
The update either leaves invariant the status of
every node in $T$ with respect to being
empty, light or full,
or it causes the same transition from light
to empty or full or from empty or full
to light at all nodes
on a last (bottom) part of $P$, while the
other nodes on $P$ remain light and no other
node changes its status.

\medskip
\emph{Case 1}:
$i\in V_Q$
is true after the update
if and only if it was true before the update.\\
If the update changes $i$'s status with respect to
being empty, light or full, we must have one of
two situations:
Either $i\not\in V_Q$ even when $i$ is light,
in which case $i$ has at least one
light leaf as a left sibling,
or $i\in V_Q$ even when $i$ is not light,
in which case the status of $i$ switches between
light and empty and the light path that ends at
$i$ when $i$ is empty coincides with the
light path that ends at $i$ when $i$ is light
(informally, the switch to light of some nodes on the
path but not its first node---which
is always light---only makes
those nodes ``more preferred children'').
It is now easy to see that the update does
not change the set of light paths.

Use the algorithm of Proposition~\ref{prop:semi}
to compute $\pi(i)$ and the related information.
If $i$ is a proxy before and therefore also after
the update, convert $D_i$
(found in $H_{\pi(i)}$) to the standard representation
after saving the history stored in its free bits
in a temporary variable, then execute
$D_{i}.\Tvn{setcolor}(j,m)$, and finally reconvert
$D_i$ to the compact representation
(which may be the $\jj$-free representation for
a different $\jj$) and
store the history saved in its free bits.
If $i$ was empty but not a proxy before the update,
let $D_i$ be a newly initialized leaf dictionary, execute
$D_{i}.\Tvn{setcolor}(j,m)$ and store $D_i$ in $H_{\pi(i)}$.
In the remaining case only execute
$D_{i}.\Tvn{setcolor}(j,m)$.

\medskip
\emph{Case 2}:
$i\in V_Q$ holds
after the update, but not before it.\\
In this case $i$ is the last node of
a new light path~$P'$.
Again traverse $P$ backwards to find the
first node $v$ on $P$ (i.e., the node on
$P$ of maximum height) that did not belong to
a light path before the update.
The update changes the status of $i$ and~$v$
from empty or full to light.
Moreover, after the update every proper descendant of $v$
on $P$ is the only light child of its parent
and therefore its preferred child.

If $v$ has at least one light left sibling
or is the root $r$ of $T$,
$v$ is not a preferred child even
after the update and so
must be an inner node in $T$,
i.e., we cannot have $v=i$ (and
nontheless be in Case~2).
Then $v$ is the top node of $P'$,
$i$ is its proxy,
and the leftmost leaf descendant $h$
of $v$ is the historian of~$P'$ 
(see Fig.~\ref{fig:setcolor}(a)).
Since neither $i$ nor $h$ was the leftmost
leaf descendant
of a light node in $T$ or
belonged to a light path before
the update, neither was a proxy or a
historian before the update.
Thus $\pi$ changes into a permutation $\pi'$
of $\{1,\ldots,N\}$ that coincides with $\pi$, except that
$\pi(i)=i=\pi'(h)$ and
$\pi(h)=h=\pi'(i)$.
Accordingly carry out the following steps:
If $i$ was empty before the update,
let $D_i$ be a newly initialized leaf dictionary.
Subsequently, whether or not $i$ was empty,
execute $D_{i}.\Tvn{setcolor}(j,m)$
and convert $D_i$ to the compact representation.
Then store $D_i$ together with
the new history of $i$, which
is a suffix of~$\Gamma'$, in $H_h$
while
saving the old value of $H_h$ in $H_i$ if $h\not=i$.

\begin{figure}
\begin{center}
\epsffile{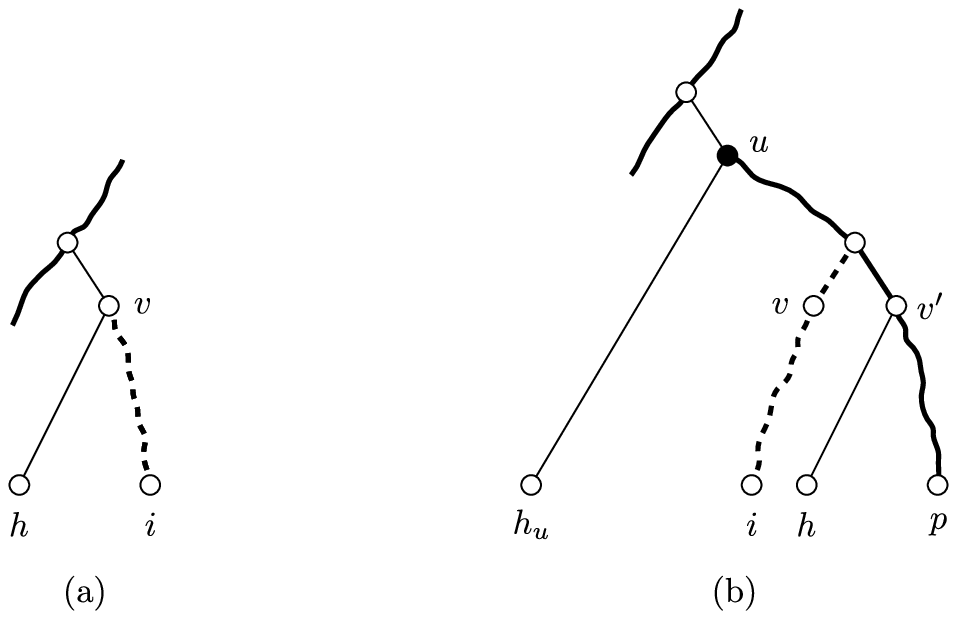}
\end{center}
\caption{(a) A new light path $P'$ (from $v$ to $i$)
is added without changes to the existing light paths.
(b) A new light path (from $u$ to $i$)
grabs an initial part of an old light path (from $u$ to $p$).}
\label{fig:setcolor}
\end{figure}

If $v$ has no light left sibling and is not $r$,
the situation is more complicated
(see Fig.~\ref{fig:setcolor}(b)).
This is because the update switches the
preferred child of the parent of $v$
from some sibling $v'$ of $v$ to~$v$.
Carry out a partial descent in $T$,
starting at $v'$ and with the
selection rule ``proxy-seeking'', to find the
end node $p$ of the old light path through $v'$
and let $h$ be the leftmost leaf descendant of~$v'$.
Also ascend in $T$ from $v$ until finding
the last (deepest) node $u$ on $P$
that was a top node before the update.
After the update $u$ is still a top node,
but with proxy $i$ instead of its old proxy $p$.
Assume first that $v'$ is a nonempty inner node in $T$.
Then $v'$ is a right sibling of $v$ and a new top node with
proxy $p$ and historian~$h$.
Let $h_u$ be the historian of $u$
(which does not change as a result of the update).
Before the update, $i$ was neither a proxy nor a
historian unless $i=h_u$,
and $h$ was neither a proxy nor a historian
unless $h=p$.
The update therefore changes $\pi$ into the permutation
$\pi'$ of $\{1,\ldots,N\}$ that coincides with $\pi$, except that

\medskip

\centerline{\vbox{\tabskip=1em\halign{$#$\hfil&#$=$&\hfil$#$\hfil&#$=$&\hfil$#$&
 \hskip 2em#\hfil\cr
\pi(h)&&h&&\pi'(p)&(only if $h\not=p$)\cr
\pi(p)&&h_u&&\pi'(i)\cr
\pi(i)&&i&&\pi'(h_u)&(only if $i\not=h_u$)\cr
\pi(h_u)&&p&&\pi'(h)\rlap{.}\cr
}}}

\medskip
\noindent
Accordingly,
move the old value of $H_h$ to $H_p$ (superfluous if $h=p$),
move the old value of $H_p$ to $H_i$
(superfluous if $i=h_u$),
let $D_i$ be a newly initialized leaf dictionary
if $i$ was empty before the update
and otherwise obtain $D_i$ from the old value of $H_{\pi(i)}$,
execute $D_i.\Tvn{setcolor}(j,m)$,
store the compact representation of $D_i$
together with the new history of $i$
(which is a suffix of $\Gamma'$) in $H_{h_u}$,
and finally move the old value of $H_{h_u}$ to $H_h$.
The latter value contains $D_{p}$ together with the
old history of $p$.
The new history of $p$ is a proper suffix of its
old history, but storing the latter does no harm
(the extra bits are considered unused anyway).

If $v'$ is empty, the steps executed are the same,
except that we refrain from moving
the old value of $H_{h_u}$ to $H_h$
(if $h=h_u$ the appropriate value was already stored in $H_h$,
and if $h\not=h_u$
the new value of $H_h$ can be arbitrary).
Finally, if $v'$ is a nonempty leaf,
the procedure is also the same, except that
at the end,
since $v'=p=h$ stops being a proxy, $D_p$
(stored in $H_p$) must be converted to the standard representation.
Its free bits contain no relevant information.

\medskip
\emph{Case 3}:
$i\in V_Q$ holds before the update, but not after it.\\
This case essentially entails undoing the steps
described for the previous case.
Traverse $P$ backwards to find the first node~$v$
on $P$ that ceases to belong to $V_Q$ as a result of the update.
The update changes the status of $i$ and $v$
from light to empty or full.

If $v$ was a top node before the update,
$i$ was its proxy, and neither $i$ nor the
old historian $h$ of $v$ is a proxy or a
historian after the update.
Thus $\pi$ changes into the permutation
$\pi'$ of $\{1,\ldots,N\}$ that coincides with $\pi$,
except that
$\pi(i)=h=\pi'(h)$ and
$\pi(h)=i=\pi'(i)$.
Accordingly obtain $D_i$ from $H_h$,
convert it to the standard representation,
execute $D_i.\Tvn{setcolor}(j,m)$
and store $D_i$ in $H_i$ while
saving the old value of $H_i$ in $H_h$ if $h\not=i$.

If $v$ was not a top node before the update,
let $v'$ be the preferred child of the parent of
$v$ after the update and let $h$
be the leftmost leaf descendant of~$v'$.
If $v'$ is empty take $p=h$, and otherwise let
$p$ be the leaf reached by
a partial descent in $T$ that starts at $v'$
and uses the selection rule ``proxy-seeking''.
Assume first that $v'$ is a nonempty inner node in~$T$.
Then $v'$ was a top node with proxy $p$
and historian $h$ before the update.
Ascend in $T$ from $v$ to find the last node $u$
on $P$ that was a top node before the update.
After the update $u$ is still a top node, but
with proxy $p$ instead of its old proxy~$i$,
and $v'$ is not a top node.
Let $h_u$ be the historian of $u$
(which does not change as a result of the update).
After the update, $i$ is neither a proxy nor a historian
unless $i=h_u$, and $h$ is neither a proxy nor
a historian unless $h=p$.
Therefore the update changes $\pi$ into the
permutation $\pi'$ of $\{1,\ldots,N\}$ that coincides
with $\pi$, except that

\medskip

\centerline{\vbox{\tabskip=1em\halign{$#$\hfil&#$=$&\hfil$#$\hfil&#$=$&\hfil$#$&
 \hskip 2em#\hfil\cr
\pi(h)&&p&&\pi'(h_u)\cr
\pi(h_u)&&i&&\pi'(i)&(only if $i\not=h_u$)\cr
\pi(i)&&h_u&&\pi'(p)\cr
\pi(p)&&h&&\pi'(h)&(only if $h\not=p$).\cr
}}}

\medskip
\noindent
Accordingly, 
move the old value of $H_h$ to $H_{h_u}$
after prefixing the history stored
in $H_h$ by the subsequence of $\Gamma'$ that pertains to
the part of $P$ from $u$ to the parent of $v'$,
convert $D_i$ (found in the old value of $H_{h_u}$)
to the standard representation,
execute $D_i.\Tvn{setcolor}(j,m)$
and store $D_i$ in $H_{\pi'(i)}$.
Finally move the old value of $H_i$ to $H_p$ if $i\not=h_u$
and move the old value of $H_p$ to~$H_h$
if $h\not=p$.

If $v'$ is empty, the steps executed are the same,
except that in place of the old value of $H_h$
(which can be arbitrary) we use a newly initialized
leaf dictionary, converted to the compact representation
and equipped with a history that shows every node
as being empty.
If $v'$ is a nonempty leaf, the procedure is also the same, except that
since $v'=h=p$ was not a proxy before the update, the old value of $H_h$,
before being moved to $H_{h_u}$, must be converted
to the compact representation and equipped
with a history equal to (the empty sequence prefixed by)
the subsequence of $\Gamma'$ that pertains to
the part of $P$ from $u$ to the parent of $v'$.

Observe that within each of Cases 1--3, all reading
from some of $H_1,\ldots,H_N$ can take place before
all writing to some of the same words.
Therefore the only possible source of inconsistency
in the data read
is the update of histories of nodes on~$P$
carried out before the computation
splits into Cases 1--3.
Most of the update conceptually happens on the
path $P$, on which we can navigate using
$\Gamma$ and $\Gamma'$.
The only occasion on which we need to navigate
outside of $P$ is during the partial descent
from $v'$ that takes place in Cases 2 and~3.
But if $v'$ is empty, the descent is trivial
and needs no inspection of histories, and if
$v'$ is light, it is necessarily to the
right of $v$, which implies that every history
inspected during the descent from $v'$ is stored
strictly to the right of every history of a
node on~$P$, and thus of every leaf
whose associated word might already have changed.

\medskip

Every operation of $D$ inspects or changes
$O(t)$ parts of histories stored in leaf
dictionaries, which takes $O(t)$ time.
In addition to this, $D.\Tvn{color}$
calls \Tvn{color} once in a leaf dictionary
and $D.\Tvn{successor}$ calls \Tvn{successor}
at most twice in a leaf dictionary.
Therefore these operations execute in $O(t+f)$
and $O(t+c)$ time, respectively.
$D.\Tvn{setcolor}$ carries out a constant number
of conversions between standard and
compact representations in
leaf dictionaries and therefore
altogether executes in $O(t+c)$ time.
If tables are available that allow the
leaf dictionaries to execute \Tvn{color}
and \Tvn{successor} in constant time,
$D$'s operations \Tvn{color} and \Tvn{successor}
work in $O(t)$ time (whereas the
time bound for \Tvn{setcolor} does not change).
This ends the proof of Lemma~\ref{lem:smalltree}.
\end{proof}

\begin{lemma}
\label{lem:unsystematic-tf}%
There is a choice dictionary that, for
arbitrary $n,f,t\in\TbbbN$,
can be initialized for universe size~$n$,
$c=2^f$ colors and tradeoff parameter $t$
in constant time and that subsequently occupies
$n f+O({{c n(c^2 f t}/w})^t+\log n)$ bits
and supports \Tvn{color} in $O(t+f)$ time and
\Tvn{setcolor},
\Tvn{choice} and, given $O(c\log n)$ additional bits,
robust iteration in $O(t+c)$ time.

Alternatively, for arbitrary fixed $\epsilon>0$,
if given access to tables of
$O(c^{\epsilon c^2})$ bits that can be computed in
$O(c^{\epsilon c^2})$ time and
depend only on~$c$, the data structure supports
\Tvn{color}, \Tvn{choice} and,
given $O(c\log n)$ additional bits,
robust iteration
in $O(t)$ time.
\end{lemma}

\begin{proof}
A data structure composed of $c$ instances of
the data structure of Theorem~\ref{thm:systematic-2},
one for each color, can support each
operation in constant time.
Assume therefore without loss of generality that
$c^2 f t<w$.
A word RAM with a word length of $w$ bits
can simulate one with a word length of $2 w$ bits
with constant slowdown.
This allows us to assume not only that
$c^2 f t<w$, but that $w$ is a multiple of $c^2 f t$
and that
the available
word length is in fact $2 w$.
Therefore define
$d={{2 w}/{(2 c^2 f t)}}={w/{(c^2 f t)}}\ge 2$ and $N=d^t$
and let $w'={{2 w}/f}$, in accordance
with the conventions used in this
section until this point.
Define $m=\Tceil{{n/{(N w')}}}$.
If $m=1$, the result follows directly from
Lemma~\ref{lem:smalltree}.
The latter assumes the universe size to be
exactly $N w'$, but a tree with fewer than
$N$ leaves
can be accommodated with straightforward
changes---essentially, each node should
adapt to
the number of its children---and
an ``incomplete leaf'',
one whose universe size is smaller than~$w'$,
can be handled separately with
Lemma~\ref{lem:atomic-c}.

If $m\ge 2$, we employ the trie-combination method of
Section~\ref{sec:trie} with
the degree sequence $(N w',n)$,
so that the overall trie is of height~2.
The 2-color choice dictionaries
of (the root of) the upper trie of height~1,
one for each of the $c$ colors and one
to keep track of initialization, are
instances of the data structure of
Theorem~\ref{thm:systematic-2}.
The number of bits needed
for these choice dictionaries is
$O(c m)=O(c n/{d^t})
=O(c n({{c^2 f t}/w})^t)$.

The choice dictionaries of (the roots of) the
lower tries, also of height~1,
are instances of the data
structure of Lemma~\ref{lem:smalltree}.
The same simple changes as above to reduce the
universe size yield a data structure
suitable for use at the rightmost lower trie.
Each instance uses 
a number of bits equal to $2$ plus
$f$ times its universe size, so the total
number of bits needed by all instances is $n f+2 m$.
This shows the space bounds of the theorem.
The time bounds follow from those of
Lemma~\ref{lem:smalltree}.
\end{proof}

\begin{theorem}
\label{thm:unsystematic-f}%
For every fixed $\epsilon>0$,
there is a choice dictionary that, for
arbitrary $n,f,t\in\TbbbN$,
can be initialized for universe size~$n$,
$c=2^f$ colors and tradeoff parameter $t$
in constant time and subsequently occupies
$n f+O({{c n(c^2 f t}/w})^t+c^{\epsilon c^2}+\log n)$ bits
and supports
\Tvn{setcolor} in $O(t+c)$ time and
\Tvn{color}, \Tvn{choice} and,
given $O(c\log n)$ additional bits,
robust iteration in $O(t)$ time.
In particular, if $c=O(\sqrt{{{\log n}}/{\log\log n}})$,
there is a data structure with the functionality
indicated that occupies
$n f+O({{c n(c^2 f t}/w})^t+n^\epsilon)$ bits.
\end{theorem}

\begin{proof}
We essentially use the data structure of
Lemma~\ref{lem:unsystematic-tf}, but
incorporate its tables 
into the data structure itself.
We need the tables exclusively to speed
up the operations \Tvn{color} and \Tvn{successor}
of the data structure of Lemma~\ref{lem:j-free}
from $O(c)$ time to constant time.
The only part of the realization of these two
operations described in the proof of
Lemma~\ref{lem:j-free} that needs more than
constant time without the use of
tables is a constant number of conversions
from the compact to the $\jj$-intermediate representation
within \emph{blocks}
of a certain number (approximately
${{\epsilon c}/2}$) of small groups.
Let $Y\Tsub C$ be the table, introduced in the proof of
Lemma~\ref{lem:j-free}, that realizes this conversion. We must prove
that $Y\Tsub C$ can
be computed in a lazy manner so that all
entries ever inspected have the correct value.
But this is easy:
Whenever a new block arises, it does
so in an execution of \Tvn{setcolor}, and the
time bound of \Tvn{setcolor} of
the present theorem allows for the
$O(c)$-time conversion of the block,
as described in the proof of Lemma~\ref{lem:j-free},
after which the relevant table entry can be filled in.
In the very first call of $\Tvn{setcolor}$,
we also use $O(c)$ initial steps to compute
the size of the table, so that the table can
be allocated before the other parts of the
data structure.
\end{proof}

Theorem~\ref{thm:unsystematic-f} can be
used to improve the lower-order terms of
Theorem~\ref{thm:systematic-2}.
Replacing all choice dictionaries at nodes of
height at least~2 in the construction of the proof of
Theorem~\ref{thm:systematic-2} by a single
instance of the choice dictionary of
Theorem~\ref{thm:unsystematic-f} with $c=2$,
we obtain a bound of
$n+{n/{(t w)}}+O(n({t/w})^t+\log n)$ bits.

Lemma \ref{lem:unsystematic-tf}
and Theorem~\ref{thm:unsystematic-f}
deal with the case in which the number $c$ of colors
is a power of~2.
We now turn to the case of general values of $c$
and first provide an analogue of
Lemma~\ref{lem:j-free}.

\begin{lemma}
\label{lem:j-free-c}%
Let $c$, $r$ and $K$ be positive integers
with $c\ge 2$ and $r\log c=O(w)$
and assume that $r$ is a multiple of $c$.
Then there is a choice dictionary $D$ with universe size
$2 r$ and for $c$ colors that can be initialized
in constant time and subsequently,
for integers $K'$, $q$ and $q'$ with $1\le K'\le K$,
$1\le q,q'\le\Tceil{{r/{(c K)}}}$
and $r=c q(K'-1)+c q'$, stores its
state as an element of
$\{0,\ldots,c^{2 c q}-1\}^{K'-1}\times\{0,\ldots,c^{2 c q'}-1\}$
and, if given access to tables of
$O(s)$ bits that can be computed in
$O(s)$ time and depend only on $c$, $r$ and $K$,
where $s=2^{\Tceil{\log c}\cdot 2 c\Tceil{{r/{(c K)}}}} r\log c
=O(c^{3 c+{{3 r}/K}} r)$,
executes
\Tvn{color}, \Tvn{setcolor} and
\Tvn{successor} in $O(K)$ time.
Moreover, during periods in which
$S_{\jj}=\emptyset$, where $(S_0,\ldots,S_{c-1})$
is $D$'s client vector and $\jj\in\{0,\ldots,c-1\}$,
$D$ supports $O(K)$-time conversion to and from
a $\jj$-free representation in which $D$ can
store ${r/c}$ unrelated bits, spaced apart by
gaps of $\Tfloor{2 c\log c}-1$ bits, that can be
read and written together in $O(K)$ time.
\end{lemma}

\begin{proof}
Write $r'={r/c}$ as $r'=(K'-1)q+q'$ for integers
$K'$, $q$ and $q'$ with $1\le K'\le K$ and
$1\le q,q'\le\Tceil{{r'}/K}$, which is clearly
possible (e.g., take $q=\Tceil{{r'}/K}$ and
choose $K'$ as large as possible).
$D$ operates with three types of representations.
In the \emph{standard representation}, the
$2 r$ digits, each drawn from $\{0,\ldots,c-1\}$,
are partitioned into \emph{segments} of consecutive
digits, $K'-1$ segments of $2 c q$ digits each and a
final segment of $2 c q'$ digits,
the digits within each segment are represented through an
integer in $\{0,\ldots,c^{2 c q}-1\}$ for the
$K'-1$ first segments and in
$\{0,\ldots,c^{2 c q'}-1\}$ for the final segment,
and $D$ stores the resulting $K$-tuple of integers.
In the \emph{compact representation}, i.e., the
$\jj$-free representation for some
$\jj\in\{0,\ldots,c-1\}$,
the digits are
partitioned into \emph{small groups} of
$2 c$ consecutive digits each, and the digits
in each small group are represented through
an integer in $\{0,\ldots,(c-1)^{2 c}-1\}$.
Since $2 c\log(c-1)\le 2 c\log c+2 c\ln(1-{1/c})
\le\Tfloor{2 c\log c}-1$, each small group can
be stored in a field of
$f=\Tfloor{2 c\log c}$ bits, with one bit
in the field left unused.
The summary bits of the proof of
Lemma~\ref{lem:j-free} are not needed
in the present data structure.
$D$ stores for each
segment the integer whose binary representation is
the concatenation of the bit sequences in the fields
of the small groups in the segment.
In the \emph{loose representation}, finally,
each of the $2 r$ digits is stored in
$\Tceil{\log c}$ bits, and the entire sequence
of $2 r$ digits occupies
$2 r\Tceil{\log c}=O(w)$ bits.

Conversion between the standard and compact
representations is carried out, segment by
segment, via the loose
representation and with the aid of tables.
This takes constant time per segment
and $O(K)$ time altogether.
As argued in the proof of Lemma~\ref{lem:j-free},
when $D$ is in the loose representation
and $\jj\in\{0,\ldots,c-1\}$,
the functions $\Tvn{skip}_{\jj}$ and $\Tvn{skip}^{-1}_{\jj}$
can be applied to all digits in constant time
using word parallelism.
Because of this, the tables that map to and from
the $\jj$-free representation can be made independent of~$\jj$.
As a consequence, the largest conversion tables
have at most $2^{\Tceil{\log c}\cdot 2 c\Tceil{r/{(c K)}}}$
entries of $O(r\log c)$ bits each,
so the tables are of total size
$O(s)$ bits and can be computed in $O(s)$ time.

When $D$ is in the loose representation, it
can execute \Tvn{color}, \Tvn{setcolor} and
\Tvn{successor} in constant time as shown
in the proof of Lemma~\ref{lem:j-free}.
Since each operation requires at most two
conversions between representations,
it can be carried out in $O(K)$ time.
There is one unused bit for every
small group, i.e., for every $2 c$ digits,
yielding a total of ${r/c}$ unused bits,
and within each segment the unused bits are spaced
apart by gaps of $f-1$ bits.
The unused bits can be read or written in
constant time per segment,
i.e., in $O(K)$ time altogether.
\end{proof}

Substituting Lemma~\ref{lem:j-free-c} for
Lemma~\ref{lem:j-free},
we can prove the following analogue
of Lemma~\ref{lem:unsystematic-tf}:

\begin{lemma}
\label{lem:unsystematic-tc}%
For every fixed $\delta>0$,
there is a choice dictionary that, for
arbitrary $n,c,t,r\in\TbbbN$ with $r\log c=O(w)$,
can be initialized for universe size~$n$,
$c$ colors and tradeoff parameters $t$ and $r$
in constant time and subsequently occupies
$n\log c+O(c n({{c^2 t}/r})^t+\log n+1)$
bits and,
if given access to tables of
$O(c^{\delta r+3 c})$ bits
that can be computed in $O(c^{\delta r+3 c})$
time and depend only on $c$ and $r$,
supports \Tvn{color}, \Tvn{setcolor},
\Tvn{choice} and,
given $O(c\log n)$ additional bits,
robust iteration in $O(t)$ time.
\end{lemma}

\begin{proof}
Without loss of generality assume that
$c^2 t<r$ and, since an arbitrary increase
of $r$ by at most a constant factor can be
``compensated for'' by a corresponding
decrease in $\delta$, that $r$ is a multiple of
$2 c^2 t$.

We use a similar construction as in the
proof of Lemma~\ref{lem:unsystematic-tf}.
Instead of storing ${w/{\log c}}$ digits to base~$c$
in a $w$-bit word maintained
in an instance of the data structure of
Lemma~\ref{lem:j-free}, however, we now
store $2 r$ digits to base~$c$ in an
instance of the data structure of
Lemma~\ref{lem:j-free-c}, 
called a \emph{leaf dictionary}
and
initialized with
$K$ chosen as an integer
constant larger than $3/\delta$.
This choice of $K$ ensures that the tables used
by the data structure of Lemma~\ref{lem:j-free-c}
are of
$O(c^{\delta r+3 c})$ bits,
for $r$ larger than a constant,
and can be computed in $O(c^{\delta r+3 c})$ time.

The integers that constitute the states of all
leaf dictionaries---one for each segment---are stored in an instance
of the data structure of
Lemma~\ref{lem:succincter-t}, initialized with $b=r$.
This needs $n\log c+O({n/{2^b}}+\log n+1)=
n\log c+O(n({{c^2 t}/r})^t+\log n+1)$ bits, requires a
table of $O(b^2)=O(c^{\delta r})$ bits and
allows us to read and write states of
leaf dictionaries in constant time.
As a result, a leaf dictionary can execute
every operation in constant time.
The data structure of
Lemma~\ref{lem:j-free-c} may employ segments of two
different sizes (namely $2 c q$ and $2 c q'$),
which, in the context of Lemma~\ref{lem:succincter-t},
translates into a sequence $(c_1,\ldots,c_p)$
that contains two different values.
Because Lemma~\ref{lem:succincter-t} tolerates only
a constant number of \emph{changes}, i.e.,
positions $i\in\{1,\ldots,p-1\}$ with
$c_i\not=c_{i+1}$, we present the values of
segments to the data structure in an order
that ensures that $(c_1,\ldots,c_p)$ has
at most one change.

Take $d={r/{(c^2 t)}}$ and $N=d^t$.
Similarly as in the proof of Lemma~\ref{lem:unsystematic-tf},
the $n$ color values are stored in
$m=\Tfloor{{n/{(2 r N)}}}$ complete $d$-ary trees
of depth $t$, ``surmounted'' by $c+1$ instances
of the choice dictionary of
Theorem~\ref{thm:systematic-2} that need
$O(c m)=O({{c n}/{d^t}})=
O(c n({{c^2 t}/r})^t)$ bits,
and possibly one incomplete tree that can
be dealt with as indicated
in the proof of
Lemma~\ref{lem:unsystematic-tf}.
If the rightmost leaf is ``incomplete'',
we consider it not to belong to any of the trees.
Instead we handle the associated color values separately,
storing them as up to $K$ integers in an  
``incomplete standard representation''
that are converted to a loose representation
whenever we need to operate on them.

The crucial inequality
$c d t\le {r/c}$ shows that the free storage
offered by a leaf dictionary in the compact representation
is sufficient to hold the history of a leaf.
Using the same algorithms as in the proof of
Lemma~\ref{lem:smalltree}, we can therefore
execute \Tvn{color}, \Tvn{setcolor}, \Tvn{choice}
and robust iteration in $O(t)$ time.
\end{proof}

\Fpasteinsection{\unsystematic}

\begin{proof}
We use two data structures, $D\Tsup T$ and $D$,
that interact in a way described in greater
detail in the proof of Theorem~\ref{thm:succincter}.
The first $c$ operations are served by $D\Tsup T$,
while an interleaved background process computes
certain quantities needed by~$D$.
After $c$ operations $D$
is ready,
and during the next $c$ operations $D\Tsup T$
and $D$ work in parallel while a background
process gradually transfers the elements in
$D\Tsup T$ of nonzero color to $D$.
After $2 c$ operations $D\Tsup T$ is dropped.

$D\Tsup T$ is an instance of the
choice dictionary of Theorem~\ref{thm:m-c}.
Since it is used only
during the first $2 c$ operations,
it fits in $O(c^2 n^\epsilon)$ bits,
a negligible quantity in the present context.

Assume without loss of generality that
$c^\epsilon\le n+1$ (we could even assume $c\le\log(n+1)$).
The second data structure, $D$, is closely
related to that of Lemma~\ref{lem:unsystematic-tc},
initialized
with $\delta=\epsilon^2/2$ and
with $r$ chosen
as an integer with
${r/2}\le{{\log(n+1)}/{(\epsilon \log c)}}\le r$,
so that $c^{\delta r}\le (n+1)^\epsilon=O(n^\epsilon)$.
We incorporate the tables used by the choice
dictionary of the lemma
into $D$ itself.
What remains is essentially to show
how to compute the tables sufficiently fast.

The preprocessing for $D$ serves to obtain
$c^{2 c}$ and, with it,
the quantity $f=\Tfloor{2 c\log c}$
used by the data structure of
Lemma~\ref{lem:j-free-c},
as well as $c^{3 c}$, needed
to estimate the size of its tables in
preparation for their allocation.

The tables employed by the data structure of
Lemma~\ref{lem:j-free-c} are used to convert segments
in the loose representation to and from the standard and
compact representations.
Since all computation takes place on segments
in the loose representation, the two other representations can
in fact be arbitrary encodings of segments,
except that they should fit in the available space.
We can therefore deal with both the standard and
the compact representation as described in the
proof of Theorem~\ref{thm:p} in the case of the
tables $Y\Tsub E$ and $Y^{-1}\Tsub E$, i.e., hand
out the codes $0,1,2,\ldots$ in that order and
compute the tables in a lazy fashion.

Finally, as concerns the data structure of
Lemma~\ref{lem:succincter-t}, we can simply replace
it by the data structure of Theorem~\ref{thm:succincter},
which uses no external tables.
The sequence $(c_1,\ldots,c_p)$
of the previous proof is $\epsilon$-balanced
for some fixed $\epsilon>0$
because either all segments are of the same size
or the larger segments are at least as
many as the smaller segments.
For $i=1,\ldots,p$, $c_i$ is either $c^{2 c q}$
or $c^{2 c q'}$, where $q$ and $q'$ are
defined in the proof of Lemma~\ref{lem:j-free-c}.
Since here we have
$q,q'=O(\log n)$, the sequence $(c_1,\ldots,c_p)$
can be communicated to the data structure
of Theorem~\ref{thm:succincter} as a sequence
of the form $(x^{y_1},\ldots,x^{y_p})$,
where $x=c^{2 c}$ and each of $y_1,\ldots,y_p$
is either $q$ or $q'$
(the computation of $x$ was considered above).
\end{proof}

\section{Applications of Choice Dictionaries}

When considering algorithmic problems, we assume
that the input is provided in read-only memory and
the output is sent to write-only memory
and count only the bits of working memory used.
When the input includes a graph $G=(V,E)$,
we make the standard assumption that $V=\{1,\ldots,|V|\}$.

For all integers
$n$ and $k$ with $1\le k\le n$, a \emph{$k$-permutation}
of $\{1,\ldots,n\}$ is
a sequence of $k$ pairwise distinct elements
of $\{1,\ldots,n\}$.

\begin{theorem}
For all fixed $\epsilon>0$ and
for arbitrary $n,k,t\in\TbbbN$ with $1\le k\le n$,
a $k$-permutation of $\{1,\ldots,n\}$
can be drawn uniformly at
random from the set of all $k$-permutations
of $\{1,\ldots,n\}$
and output in $O(t k)$ time using
$n+O({{n\log(t\log n)}/{(t\log n)}}+n^\epsilon)$ bits of
working memory.
\end{theorem}

\begin{proof}
Initialize a 2-color instance of the
choice dictionary of Theorem~\ref{thm:p}
for universe size $n$,
call its client set $S$ and,
$k$ times, draw an element
uniformly at random from $\overline{S}$,
output it and insert it in~$S$.
\end{proof}

\begin{corollary}
For all fixed $\epsilon>0$ and
for arbitrary $n,t\in\TbbbN$, a
permutation of $\{1,\ldots,n\}$
can be drawn uniformly at random
from the set of all permutations of $\{1,\ldots,n\}$
and output
in $O(t n)$ time using
$n+O({{n\log(t\log n)}/{(t\log n)}}+n^\epsilon)$ bits of
working memory.
\end{corollary}

Simple as the algorithm of the corollary is,
we can prove that it is close to using
the minimum possible amount of working memory.
Assume that for all $n\in\TbbbN$, $L_n$ is
a finite language of binary strings such that
no string in $L_n$ is a proper prefix of
another string in $L_n$ ($L_n$ is \emph{prefix-free})
and $g_n$ is a function from $L_n$ to
$\{1,\ldots,n\}$.
For the lower bound below, we relax the requirements
for what it means to output a permutation
$\pi$ of $\{1,\ldots,n\}$.
Rather than demanding that the output be
a sequence of exactly $n$ $w$-bit integers,
the $i$th of which is $\pi(i)$, for $i=1,\ldots,n$,
we allow the output to be a sequence
$(u_1,\ldots,u_m)$ of a possibly variable
number of bit strings of possibly variable lengths
such that
the concatenation $u_1\cdots u_m$ can be written
in the form $v_1\cdots v_n$, where $v_i\in L_n$
and $g_n(v_i)=\pi(i)$ for $i=1,\ldots,n$.
For example, this allows several values of $\pi$
to be output in the same word or a nonstandard
representation of integers to be used.

\begin{theorem}
Let $\mathcal{A}$ be a randomized algorithm that
outputs a permutation
of $\{1,\ldots,n\}$ for some $n\in\TbbbN$,
with each of the $n!$ such
permutations being output with positive
probability.
Then $\mathcal{A}$ uses at least $n-\log_2(n+1)$ bits
of working memory.
\end{theorem}

\begin{proof}
Suppose that the output of $\mathcal{A}$ is a
sequence $(u_1,\ldots,u_m)$ and
write $u_1\cdots u_m=v_1\cdots v_n$
as discussed above.
Let $k\in\{1,\ldots,m\}$ be minimal
with $|u_1\cdots u_k|\ge|v_1\cdots v_{\lfloor{n/2}\rfloor}|$
and consider the point in time just
before $u_k$ is output.
Now $v_1\cdots v_{\lfloor{n/2}\rfloor}$
determines $\pi(1),\ldots,\pi(\Tfloor{{n/2}})$,
where $\pi$ is the permutation computed
by~$\mathcal{A}$.
There are
${n\choose{\Tfloor{{n/2}}}}$ possibilities
for the set $\{\pi(1),\ldots,\pi(\Tfloor{{n/2}})\}$,
each of which occurs with positive probability,
and $u_k\cdots u_m$ cannot be the same for
any two distinct such possibilities.
Therefore, at the point
in time under consideration there must be
at least ${n\choose{\Tfloor{{n/2}}}}$ possibilities
for the state of $\mathcal{A}$.
Since ${n\choose{\Tfloor{{n/2}}}}\ge{n\choose i}$
for $i=0,\ldots,n$
and $\sum_{i=0}^n{n\choose i}=2^n$,
${n\choose{\Tfloor{{n/2}}}}\ge{{2^n}/{(n+1)}}$,
which implies that
the number of bits used by $\mathcal{A}$
is at least
$\log({{2^n}/{(n+1)}})=n-\log(n+1)$.
\end{proof}

Given a directed or undirected
$n$-vertex graph $G=(V,E)$
and a permutation 
$\pi$ of $V$, i.e., a
bijection from $\{1,\ldots,n\}$ to~$V$,
we define a \emph{spanning forest} of $G$
\emph{consistent with} $\pi$ to be a sequence
$F=(T_1,\ldots,T_q)$, where $T_1,\ldots,T_q$
are vertex-disjoint outtrees that are
subtrees of $G$
(if $G$ is directed) or of the directed version of~$G$
(if $G$ is undirected) and the union of whose vertex sets
is $V$, such that for each $v\in V$, the root of
the tree in $\{T_1,\ldots,T_q\}$ that contains~$v$
is the first vertex in the sequence
$(\pi(1),\ldots,\pi(n))$ from which $v$ is reachable in~$G$.
If, in addition, every path in the union
of $T_1,\ldots,T_q$
is a shortest path in $G$, $F$ is a
\emph{shortest-path spanning forest} of $G$
consistent with~$\pi$.
Thus a spanning forest of $G$ consistent with~$\pi$
can be produced, e.g., by a depth-first search that,
whenever its stack of partially processed vertices
is empty, picks its new start vertex
as the first undiscovered vertex in the order
prescribed by~$\pi$.
If a breadth-first search is used instead of the
depth-first search, the result will be a
shortest-path spanning forest of~$G$
consistent with~$\pi$.

In the following, by computing a spanning forest
$F=(T_1,\ldots,T_q)$
of an $n$-vertex graph $G=(V,E)$ consistent with
a permutation $\pi$ of $G$ we will mean producing a
sequence $((u_1,v_1,k_1),\ldots,$\break$(u_n,v_n,k_n))$ of
triples with $u_i\in V\cup\{0\}$, $v_i\in V$
and $k_i\in\TbbbN$ for $i=1,\ldots,n$ such that
$k_1\le\cdots\le k_n$
and such that for $j=1,\ldots,q$,
$\{v_i\mid 1\le i\le n$ and $k_i=j\}$ and
$\{(u_i,v_i)\mid 1\le i\le n$,
$k_i=j$ and $u_i\not=0\}$ 
are precisely the vertex and edge sets of~$T_j$,
respectively.
If, in addition, for each $\ell\in\{1,\ldots,n\}$
with $u_\ell\not=0$ there is an $i\in\{1,\ldots,\ell-1\}$
with $v_i=u_\ell$, we say that $F$ is computed in
\emph{top-down order}.
Thus for $j=1,\ldots,q$, the root and
the edges of $T_j$ are to
be output (in a top-down order),
each with the index $j$ of its tree~$T_j$.
The meaning of a shortest-path spanning forest
$F=(T_1,\ldots,T_q)$
of $G$ consistent with $\pi$ (in top-down order)
is analogous,
except that each triple $(u_i,v_i,k_i)$
is extended by a
fourth component equal to the depth of $v_i$
in $T_{k_i}$.
Of course, computing a spanning forest of an
undirected graph also solves the
(suitably defined) connected-components
problem.

\begin{theorem}
\label{thm:cc}%
Given a directed or undirected graph $G=(V,E)$
with $n$ vertices and $m$ edges,
a permutation $\pi$ of $V$ and a $t\in\TbbbN$,
a spanning forest of $G$ consistent with $\pi$
can be computed in
top-down order in
$O((n+m)t\log(t+1))$ time with
$n+{n/t}+O(\log n)$ bits of working memory.
In particular, for every fixed $\epsilon>0$,
a spanning forest of $G$ consistent with $\pi$
can be computed in $O(n+m)$ time
with at most $(1+\epsilon)n$ bits.
\end{theorem}

\begin{proof}
We use $n$ bits to mark each vertex as
\Tvn{unvisited} or \Tvn{visited}.
Initially all vertices are unvisited.
In addition, we store an
initially empty set $S$ of vertices in
an instance $D$ of the data structure
of Theorem~\ref{thm:mlog}.
We also maintain a current tree index $k$,
initially~0.

Compute $s\in\TbbbN$ with
$s=\Omega(1+{n/{(t\log(t+1))}})$
such that when $|S|\le s$, $D$ occupies
at most ${n/t}+O(\log n)$ bits.
We will ensure that $|S|\le s$ always holds,
so the space used by the algorithm
is as stated in the theorem.

In an outermost loop, we step through $V$ in the
order indicated by $\pi$, i.e.,
in the order $\pi(1),\ldots,\pi(n)$, and,
for each vertex $r$ found to be unvisited at this time,
increment~$k$
(a new tree $T_k$ is begun), mark $r$ as
visited, output
$(0,r,k)$ ($r$ is the root of $T_k$
and has no parent),
and insert $r$ in $S$.
Then, as long as $0<|S|<s$, we use
\Tvn{extract-choice} to delete a vertex $u$
from $S$ and \emph{process} $u$.
Processing $u$ means, for each unvisited (out)neighbor $v$ of~$u$,
marking $v$ as visited,
outputting $(u,v,k)$
($v$ belongs to $T_k$ and its parent is $u$), and
inserting~$v$ in~$S$.
If and immediately when $|S|$ reaches $s$, we abandon what
we are doing and start a \emph{global sweep}.
A global sweep reinitializes $D$ to reset
$S$ to $\emptyset$
(or achieves the same though a sequence of
calls of \Tvn{extract-choice})
and iterates over $V$, processing each visited vertex encountered.
If $|S|$ reaches $s$ during a global sweep,
the current global sweep is abandoned, and a new
global sweep is immediately begun.
Whenever $S$ becomes empty outside of a global sweep,
the current iteration of the outermost loop
terminates (no more vertices are reachable from $T_k$).

The algorithm is easily seen to be correct.
In particular, as long as $E$ contains an edge $(u,v)$
or $\{u,v\}$ such that $u$ is visited but $v$ is not,
$u$ belongs to $S$ or will eventually be processed
in a global sweep.
Outside of global sweeps, the running time of
the algorithm is $O(n+m)$.
Between any two global sweeps, at least $s$ vertices
are marked as visited.
Since this happens only once for each vertex,
the number of global sweeps is bounded by
$1+{n/s}=O(t\log(t+1))$.
A global sweep runs in $O(n+m)$ time, so the
total running time of the algorithm is
$O((n+m)t\log(t+1))$.
\end{proof}

\begin{theorem}
\label{thm:bfs}%
Given a directed or undirected graph $G=(V,E)$
with $n$ vertices and $m$ edges,
a permutation $\pi$ of $V$, a $t\in\TbbbN$
and a fixed $\epsilon>0$,
a shortest-path
spanning forest of $G$ consistent with $\pi$
can be computed in
top-down order in $O((n+m)t)$ time with
$n\log_2 3+O(n({t/{\log n}})^t+n^\epsilon)$ bits of working memory.
If $G$ is directed, its representation must
allow iteration over the inneighbors and
outneighbors of a given
vertex in time proportional to their number
plus a constant
(in the terminology
of~\cite{ElmHK15},
$G$ must be given with in/out adjacency lists).
\end{theorem}

\begin{proof}
Using a 3-color instance of the choice dictionary
of Theorem~\ref{thm:unsystematic-c},
we store for each vertex $v\in V$ a color:\
white, gray or black.
Initially all vertices are white.
We also store a current tree index, $k$,
initially 0, and a current distance counter, $d$.
In the following, the prefixes ``(in)'' and ``(out)''
are intended to apply if $G$ is directed;
if $G$ is undirected, they should be ignored.

In an outermost loop, we step through $V$ in
the order indicated by $\pi$ and,
for each vertex $r$ found to be white at this time,
increment~$k$
(a new tree $T_k$ is begun), set $d$ to 0,
color $r$ gray and output
$(0,r,k,0)$ ($r$ is the root $r_k$ of $T_k$, it
has no parent and its depth in $T_k$ is~$0$).
We also remember $r=r_k$
as the root
of the current tree.
Then, as long as at least one vertex is gray,
we carry out an \emph{exploration round}
followed by a \emph{consolidation round}
and then increment~$d$.

In the exploration round, we
iterate over the gray vertices.
For each gray vertex $u$, we test whether 
$u=r$ 
or
$u$ has one or more black (in)neighbors.
If this is the case, we
process all white (out)neighbors of $u$,
for each such vertex $v$
coloring $v$ gray and outputting $(u,v,k,d+1)$
($v$ belongs to $T_k$, its parent is $u$
and its depth in $T_k$ is $d+1$).
In the consolidation round, we again iterate
over the gray vertices, now coloring black each gray vertex
without white (out)neighbors.
If there are no gray vertices after a
consolidation round,
the current iteration of the outermost loop
terminates (no more vertices are reachable from $T_k$).

Consider a particular value of $k$ and
let $V_k$ be the set of vertices $v\in V$
reachable in $G$ from $r_k$
but not from any vertex before $r_k$ in the sequence
$(\pi(1),\ldots,\pi(n))$
(i.e.,
$V_k$ is the intended vertex set of the final tree $T_k$).
The following can be proved by induction on $d$:
Immediately before an exploration round,
suppose that $v\in V_k$
and let $d_v$ be the length of a shortest path
in $G$ from $r_k$ to $v$.
Then $v$ is black if $d_v<d$,
white if $d_v>d$,
gray or black if $d_v=d$ and $v$ has no
white (out)neighbors, and
gray if $d_v=d$ and $v$ has at least one
white (out)neighbor.
The exploration round may enumerate some vertices
at distance $d+1$ from $r_k$ that were colored
gray earlier in the same round, but
the test ``is gray and either equals $r_k$ or
has a black (in)neighbor''
employed in the exploration round
is passed precisely by the vertices that were
gray at the beginning of the
round.
It is now easy to see that the tuples output
by the exploration round are correct and that,
immediately after the exploration round,
$v$ is (still) black if $d_v<d$, gray
or black if $d_v=d$, gray
if $d_v=d+1$ and white if $d_v>d+1$.
The test ``is gray and has no white (out)neighbor'' employed
in the consolidation round is
passed precisely by those vertices in $V_k$
that either have distance $d$ from $r_k$ and
are not already black or have distance $d+1$ from $r_k$
and no white (out)neighbors.
Thus the induction hypothesis holds at the
beginning of the next exploration round, if any.

Since each vertex is gray for at most two
values of $d$, i.e., for at most two
exploration rounds and two consolidation rounds,
the running time of the algorithm is readily seen
to be $O((n+m)t)$.
The space bound follows from
Theorem~\ref{thm:unsystematic-c}.
\end{proof}

The vertex set $V^*$ of
a maximal clique in an undirected graph $G=(V,E)$ can be computed
greedily by starting with $V^*=\emptyset$ and
stepping through the vertices of~$G$, including 
each in $V^*$ if it is adjacent to all vertices already in~$V^*$.
With $n=|V|$ and $m=|E|$, this takes $O(n+m)$ time
and uses $n+O(\log n)$ bits.
Below we present an output-sensitive algorithm that is potentially
faster and uses only slightly more space.
For $u\in V$, let $N_G(u)$ be the neighborhood of~$u$,
i.e., $N_G(u)=\{v\in V\mid\{u,v\}\in E\}$.
Moreover, for $W\subseteq V$, denote by $\deg_G(W)$ the
total degree of the vertices in $W$,
i.e., $\deg_G(W)=\sum_{u\in W}|N_G(u)|$.

\begin{theorem}
\label{thm:clique}%
Given an undirected $n$-vertex graph $G=(V,E)$,
a $t\in\TbbbN$ and a fixed $\epsilon>0$,
the vertex set $V^*$ of a maximal clique in $G$
can be computed in $O(t\deg_G(V^*)+1)$ time
with $n\log_2 3+O(n({t/{\log n}})^t+n^\epsilon)$ bits of working memory.
If the adjacency lists of $G$ are sorted, i.e.,
each lists the neighbors of a vertex in sorted order,
the problem can be solved in $O(\deg_G(V^*)+1)$ time with
$n+O({{n(\log w)}/w}+\log n)$ bits of working memory.
\end{theorem}

\begin{proof}
We use the following algorithm:
Output the vertex 1 and initialize a set
$W$ to $N_G(1)$.
Then, as long as $W$ is nonempty,
output an element $u$ of $W$ and replace
$W$ by $W\cap N_G(u)$.
The correctness of the algorithm
is obvious---$W$ is always the set of neighbors common to all vertices that
were already output.

Without loss of generality assume in the
rest of the proof that $|N_G(1)|\ge 1$.
We store $W$ as the elements of
color~1 in a 3-color
instance of the choice dictionary
of Theorem~\ref{thm:unsystematic-c}.
To replace $W$ by $W\cap N_G(u)$,
temporarily color the elements of $W\cap N_G(u)$
with color~2 in a scan over $N_G(u)$
and subsequently replace
first the color 1 by 0 and then the
color 2 by 1 at all vertices.
Following the initialization,
the time needed is $O(t)$ times the sum of
$\deg_G(V^*)$
and the number of color decrements.
Since the number of color decrements is bounded by
the number of color increments,
which is at most $\deg_G(V^*)$,
the total running time is
$O(t\deg_G(V^*))$.

Assume now that the adjacency lists of $G$ are sorted.
Then, as long as $|W|\ge{n/w}$, we store $W$
differently, namely through its bit-vector representation.
We can find a vertex $u$ in $W$ in $O(1+{n/w})$ time,
and to replace $W$ by $W\cap N_G(u)$, we step
through the (sorted) adjacency list of $u$
and the bit-vector representation of $W$ in parallel,
clearing all bits in the latter that do not
correspond to neighbors of $u$.
This can be done in $O(\deg_G(u)+{n/w})$ time.
Since the procedure is carried out at most once
with $\deg_G(u)<{n/w}$, the total time used
until $|W|<{n/w}$ is $O(\deg_G(V^*))$.

When $|W|$ has dropped below ${n/w}$, we spend $O({n/w})$
time to extract $W$ from the bit vector and store the
set in a colorless instance $D$ of the choice dictionary
of Theorem~\ref{thm:unsystematic-f}.
To replace $W$ by $W\cap N_G(u)$, we scan
over $N_G(u)$ and store $W\cap N_G(u)$ in an instance $D'$
of the data structure of Theorem~\ref{thm:mlog},
after which we empty $D$, extract all elements
stored in $D'$ and insert them in~$D$.
The time spent in this part of the algorithm
is obviously $O(\deg_G(V^*))$.
Since $D'$ never contains more than $n/w$ elements,
it occupies $O({{n(\log w)}/w}+\log n)$ bits.
The space bound follows.
\end{proof}

\bibliography{all}

\end{document}